\documentclass{lmcs}
\pdfoutput=1

\usepackage{lastpage}
\lmcsdoi{21}{2}{18}
\lmcsheading{}{\pageref{LastPage}}{}{}%
{Jul.~19,~2023}{Jun.~05,~2025}{}


\usepackage[utf8]{inputenc}
\usepackage{mathtools}
\usepackage{hyperref}
\usepackage{cleveref}
\usepackage{multirow}	
\usepackage{braket} 	

\usepackage{tikz}
\usetikzlibrary{arrows}
\usetikzlibrary{automata}
\usetikzlibrary{decorations.pathmorphing}
\usetikzlibrary{positioning}
\usetikzlibrary{shapes, shapes.geometric}
\usetikzlibrary{fit}
\usetikzlibrary{calc}

\usepackage{pifont}
\newcommand{\cmark}{\textcolor{green}{\ding{51}}}
\newcommand{\xmark}{\textcolor{red}{\ding{55}}}
\newcommand{\qmark}{\textcolor{blue}{\textbf{?}}}

\newcommand{\ST}{~{\big |}\:}
\newcommand{\St}{~|~}
\newcommand{\con}{\cdot}
\newcommand{\tuple}[1]{\langle #1  \rangle}

\usepackage{xspace}		
\newcommand{\Nat}{\ensuremath{\mathbb{N}}\xspace}
\newcommand{\Rat}{\ensuremath{\mathbb{Q}}\xspace}

\newcommand{\A}{{\mathcal{A}}}
\newcommand{\B}{{\mathcal{B}}}
\ifdefined\C
\renewcommand{\C}{{\mathcal{C}}}
\else
\newcommand{\C}{{\mathcal{C}}}
\fi
\newcommand{\D}{{\mathcal{D}}}

\newcommand{\M}{{\mathcal{M}}}                      

\newcommand{\T}{{\mathcal{T}}}
\newcommand{\N}{{\mathcal{N}}}

\newcommand{\emptyword}{\varepsilon}

\usepackage{accents}
\newlength{\dhatheight}
\newcommand{\doublehat}[1]{%
	\settoheight{\dhatheight}{\ensuremath{\hat{#1}}}%
	\addtolength{\dhatheight}{-0.35ex}%
	\hat{\vphantom{\rule{1pt}{\dhatheight}}%
		\smash{\hat{#1}}}}

\newcommand{\Func}[1]{\mathtt{#1}}
\newcommand{\image}{\Func{image}}
\newcommand{\Gap}{\Func{gap}}
\newcommand{\Cost}{\Func{cost}}
\newcommand{\Val}{\Func{val}}

\newcommand{\inc}{\mbox{\sc inc}\xspace}
\newcommand{\dec}{\mbox{\sc dec}\xspace}
\newcommand{\goto}{\mbox{\sc goto }}
\newcommand{\halt}{\mbox{\sc halt}\xspace}
\newcommand{\jz}[3]{\mbox{\sc if $#1$=0 goto $#2$ else goto $#3$}\xspace}
\newcommand{\CMrun}{\psi}
\newcommand{\numof}[2]{\#(#1,#2)}
\newcommand{\qfr}{q_{\mathsf{freeze}}}
\newcommand{\qhalt}{q_{\mathsf{halt}}}
\newcommand{\qhc}{q_{\mathsf{HC}}}
\newcommand{\qzc}{q_{\mathsf{ZC}}^c}
\newcommand{\qpc}[1]{q_\mathsf{PC#1}^c}
\newcommand{\incdec}{\Sigma^{\inc\dec}}
\newcommand{\allgoto}{\Sigma^{\mbox{\sc goto}}}
\newcommand{\withouthalt}{\Sigma^{\mbox{\sc nohalt}}}
\newcommand{\pref}{\mbox{\sc pref}\xspace}
\newcommand{\checker}[2]{\vspace{5pt}\noindent{\bf #1 Checker#2.}}
\newcommand{\lchecker}[2]{\vspace{5pt}\noindent{\it #1 Checker#2.}}

\newcommand*\circled[1]{\tikz[baseline=(char.base)]{\node[shape=circle,draw,inner sep=1pt] (char) {#1};}}

\crefname{figure}{Figure}{Figures}
\crefname{defi}{Definition}{Definitions}
\crefname{thm}{Theorem}{Theorems}
\crefname{lem}{Lemma}{Lemmas}
\crefname{cor}{Corollary}{Corollaries}
\crefname{prop}{Proposition}{Propositions}
\crefname{enumi}{Item}{Items}
\crefname{rem}{Remark}{Remarks}

\begin{document}

\title{Discounted-Sum Automata with Multiple Discount Factors}
\titlecomment{The article extends \cite{BH21} and parts of \cite{BH23}.}

\author[U.~Boker]{Udi Boker\lmcsorcid{0000-0003-4322-8892}}
\author[G.~Hefetz]{Guy Hefetz\lmcsorcid{0000-0002-4451-6581}}
\thanks{Research supported by the Israel Science Foundation grant 2410/22.}
\address{Reichman University, Herzliya, Israel}	
\email{udiboker@runi.ac.il, ghefetz@gmail.com}






\begin{abstract}
Discounting the influence of future events is a key paradigm in economics and it is widely used in computer-science models, such as games, Markov decision processes (MDPs), reinforcement learning, and automata. 
While a single game or MDP may allow for several different discount factors, nondeterministic discounted-sum automata (NDAs) were only studied with respect to a single discount factor. 
It is known that every class of NDAs with an integer as the discount factor has good computational properties: It is closed under determinization and under the algebraic operations min, max, addition, and subtraction, and there are algorithms for its basic decision problems, such as automata equivalence and containment.
Extending the integer discount factor to an arbitrary rational number, loses most of these good properties.

We define and analyze nondeterministic discounted-sum automata in which each transition can have a different integral discount factor (integral \emph{NMDAs}). 
We show that integral NMDAs with an arbitrary choice of discount factors are not closed under determinization and under algebraic operations and that their containment problem is undecidable.
We then define and analyze a restricted class of integral NMDAs, which we call \emph{tidy NMDAs}, in which the choice of discount factors depends on the prefix of the word read so far. Among their special cases are NMDAs that correlate discount factors to actions (alphabet letters) or to the elapsed time. 
We show that for every function $\theta$ that defines the choice of discount factors, the class of $\theta$-NMDAs enjoys all of the above good properties of NDAs with a single integral discount factor, as well as the same complexity of the required decision problems. Tidy NMDAs are also as expressive as deterministic integral NMDAs with an arbitrary choice of discount factors.

All of our results hold for both automata on finite words and automata on infinite words.
\end{abstract}

\maketitle

\section{Introduction} \label{sec:intro}

Exponential growth and decay are natural physical phenomena (nuclear radiation over time, signal strength along distance, population along generations, etc.), and are central in economics (e.g., money value over time, considering inflation or interest rates.)

As a result, discounted summation, which formulates accumulation in the presence of exponential growth or decay (depending on whether we look forward or backward), is a central valuation function in various computational models, such as games (e.g.,\  \cite{ZP96,DiscountingInSystems,Andersson06,ACSU24}), Markov decision processes (e.g,\ \cite{DiscountedMarkov,DiscountedDeterministicMarkov,MultiObjectiveDiscountedReward,ECC22}), reinforcement learning (e.g,\ \cite{IntroductionToReinforcementLearning,KBKS19,PMLR20,Sensors22,PMLR22}), and automata (e.g,\ \cite{Skew,AlternatingWeightedAutomata,ExpressivenessQuantitativeLanguages,CDH10,ComparatorAutomataInQuantitativeVerification,AD24}).

A Nondeterministic Discounted-sum Automaton (NDA) is an automaton with rational weights on the transitions, and a fixed rational discount factor $\lambda > 1$. 
The value of a (finite or infinite) run is the discounted summation of the weights on the transitions, such that the weight in the $i$th position of the run is divided by $\lambda^i$. 
The value of a (finite or infinite) word is the minimal value of the automaton runs on it.
An NDA $\A$ expresses a function from words to real numbers, and we write $\A(w)$ for the value of $\A$ on a word $w$. (We further have, by \cite[Theorems 4.10 and 4.15, and Corollary 4.12]{BHMS23}, that this function is continuous.)

In the Boolean setting, where automata express languages, closure under the basic Boolean operations of union, intersection, and complementation is desirable, as it allows to use automata in formal verification, logic, and more. In the quantitative setting, where automata express functions from words to numbers, 
the above Boolean operations are naturally generalized to algebraic ones: union to $\min$, intersection to $\max$, and complementation to multiplication by $-1$ (depending on the function's co-domain). Likewise, closure under these algebraic operations, as well as under addition and subtraction, is desirable for quantitative automata, serving for quantitative verification.
Determinization is also very useful in automata theory, as it gives rise to many algorithmic solutions, and is essential for various tasks, such as synthesis and probabilistic model checking\footnote{In some cases, automata that are ``partially deterministic'' \cite{Bok22}, such as limit-deterministic \cite{Var85b}, good-for-games \cite{HP06}, or history-deterministic automata \cite{BL21} suffice.}.

NDAs cannot always be determinized \cite{CDH10}, they are not closed under basic algebraic operations \cite{BH14}, and basic decision problems on them, such as universality, equivalence, and containment, are not known to be decidable and relate to various longstanding open problems \cite{TDS}.
However, restricting NDAs to have integral discount factors, called \emph{integral NDAs}, provides for every discount factor $\lambda\in\mathbb{N}$ a robust class of automata that is closed under determinization and under the algebraic operations, and for which the decision problems of universality equivalence, and containment are decidable \cite{BH14}.

Various variants of NDAs are studied in the literature, among which are \emph{functional}, \emph{k-valued}, \emph{comparator}, \emph{probabilistic}, and more \cite{QuantitativeLanguagesDefinedByFunctionalAutomata,FiniteValuedWeightedAutomata,ComparatorAutomataInQuantitativeVerification,ProbabilisticWeightedAutomata}.
Yet, to the best of our knowledge, all of these models are restricted to have a single discount factor in an automaton. 
This is a significant restriction of the general discounted-summation paradigm, in which multiple discount factors are considered. For example, Markov decision processes and discounted-sum games allow for multiple discount factors within the same entity \cite{DiscountedMarkov, Andersson06}.

A natural extension to NDAs is to allow for different discount factors over the transitions (\Cref{fig:Schematic}), providing the ability to model systems in which each action (alphabet letter in the automaton) causes a different discounting, systems in which the discounting changes over time, and more.

Indeed, looking into the phenomena mentioned in the beginning of the Introduction, one may observe that while for some of them there is a constant ratio of exponential growth/decay, for others it varies. For example, while half life-time of a nuclear material is constant, WiFi signal attenuation depends on the medium it goes through (the attenuation ratio is different when traveling through walls, glass, and doors), and the value of money exponentially grows/decays according to the interest rate (which, for instance with the United States federal funds rate, might change 8 times a year).

Taking the automaton in \Cref{fig:Schematic}, one may view it as modelling the financial outcome of various scenarios -- each transition stands for an event ($a$ or $b$) that occurs each month and results with some income or expense ($w_1,w_2,w_3$). Yet, due to a monthly interest rate $\lambda$, a $\$$100 earned today (on the first transition) is worth more than a $\$$100 earned next month (on the second transition), which is worth more than a $\$$100 two month from now. The value of a $\$100$ earned on the $n$th month, in terms of today's value, is $\$100 / \lambda^{n-1}$. Yet, as the interest rate might change every month, we cannot have a single discount factor $\lambda$, but rather allow for multiple ones. In \Cref{fig:Schematic}, the discount factor depends on the input letter, which may be related to both the income/expense event and the considered interest rate. Then, the value of a $\$100$ on the $n$th transition, in terms of today's value, is $\$100 / (\lambda_1 \lambda_2 \cdots \lambda_{n-1})$, where $\lambda_i$ is the discount factor of the $i$th transition (month).

As integral NDAs provide robust automata classes, whereas non-integral NDAs do not, we look into extending integral NDAs into integral \emph{NMDAs} (\Cref{def:NMDA,fig:Schematic,fig:NMDAExample}), allowing multiple integral discount factors in a single automaton.

\begin{figure}
	\centering
	\begin{tikzpicture}[->,>=stealth',shorten >=1pt,auto,node distance=2.6cm, semithick, initial text=, every initial by arrow/.style={|->}]
		\node[initial, state] (q0) {$q_0$}; 
		\node[state] (q1) [right of=q0] {$q_1$};
		
		\node[right of=q1, align=left, xshift=4.2cm] 
		{\begin{tabular}{p{1.75cm}l}
			\multicolumn{2}{l}{The transition labels:}\\
			$a,b:$&Input letters (actions/events)\\
			$w_1,w_2,w_3:$&Weights (rewards/costs/\ldots)\\
			$\lambda, \lambda':$&Discount factors (growth/decay ratios)
		\end{tabular}};
		
		\path 
		
		(q0) edge	[loop above, out=70, in=120,looseness=5] node [above, xshift=0.0cm, yshift=0.0cm] {$a,w_1,\lambda$} (q0)

		(q0) edge [above, out=20,in=160] node[yshift=-0.07cm] {$b,w_2,\lambda'$} (q1)

		(q1) edge	[below, out=-160,in=-20] node {$b,w_3,\lambda'$} (q0)

		(q1) edge	[loop above, out=30, in=80, looseness=4] node [above, xshift=0.0cm, yshift=-0.0cm]{$a,w_2,\lambda$} (q1)
		(q1) edge	[loop below, out=-30, in=-80,looseness=4] node [below, xshift=0.0cm, 	yshift=0.0cm]{$b,w_1,\lambda'$} (q1)
		;
	\end{tikzpicture}
	\caption{\label{fig:Schematic}A nondeterministic discounted-sum automaton with multiple discount factors (NMDA).}
\end{figure}

We start with analyzing NMDAs in which the integral discount factors can be chosen arbitrarily. Unfortunately, we show that this class of automata does not allow for determinization, is not closed under the basic algebraic operations, and its containment problem is undecidable. 

For more restricted generalizations of integral NDAs, in which the discount factor depends on the transition's letter (\emph{letter-oriented} NMDAs) or on the elapsed time (\emph{time-oriented} NMDAs), we show that the corresponding automata classes do enjoy all of the good properties of integral NDAs, while strictly extending their expressiveness.

We further analyze a rich class of integral NMDAs that extends both letter-oriented and time-oriented NMDAs, in which the choice of discount factor depends on the word-prefix read so far (\emph{tidy} NMDAs). 
We show that their expressiveness is as of deterministic  integral NMDAs with an arbitrary choice of discount factors and that for every choice function $\theta: \Sigma^+ \to \Nat\setminus\{0,1\}$, the class of $\theta$-NMDAs enjoys all of the good properties of integral NDAs. (See \Cref{fig:introClasses}.)

Considering closure under algebraic operations, we further provide tight bounds on the size blow-up involved in the different operations (\Cref{tbl:blow-up}).
To this end, we provide new lower bounds also for the setting of NDAs, by developing a general scheme to convert every NFA to a corresponding NDA of linearly the same size, and to convert some specific NDAs back to corresponding NFAs.

As for the decision problems of tidy NMDAs, we provide a PTIME algorithm for emptiness and PSPACE algorithms for the other problems of exact-value, universality, equivalence, and containment. The complexities are with respect to the automaton (or automata) size, which is considered as the maximum between the number of transitions and the maximal binary representation of any discount factor or weight in it.
These new algorithms also improve the complexities of the previously known algorithms for solving the decision problems of NDAs, which were PSPACE with respect to unary representation of the weights.
For rational weights, we assume all of them to have the same denominator. (Omitting this assumption changes in the worst case the PSPACE algorithms into EXPSPACE ones.)

As general choice functions need not be finitely represented, it might upfront limit the usage of tidy NMDAs. 
Yet, we show that finite transducers (Mealy machines) suffice, in the sense that they allow to represent every choice function $\theta$ that can serve for a $\theta$-NMDA.
We provide a PTIME algorithm to check whether a given NMDA is tidy, as well as if it is a $\T$-NMDA for a given transducer $\T$.

We show all of our results for both automata on finite words and automata on infinite words. Whenever possible, we provide a single proof for both settings.
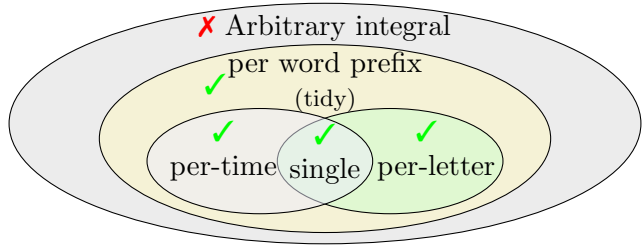
\begin{figure}
	\centering
	\begin{tikzpicture}
		\begin{scope}[fill opacity=0.5]
			\filldraw[fill=gray!30!white, draw=black]   ( 0:0) ellipse (4.2 and 1.6);
			\filldraw[fill=yellow!30!white, draw=black]   ( -90:0.2) ellipse (3 and 1.25);
			\filldraw[fill=green!20!white, draw=black] (-30:1) ellipse (1.5 and 0.69);
			\filldraw[fill=blue!7!white, draw=black]  (-150:1) ellipse (1.5 and 0.69);
		\end{scope}
		\begin{scope}
			\clip (-30:1) ellipse (1.5 and 0.69);
			\clip (-150:1) ellipse (1.5 and 0.69);
			\fill[red!7!white](0,0) circle(1.5cm);
		\end{scope}
		\node at ( 90:1.27)  {\xmark \text{} Arbitrary integral};
		\node at ( 160:1.55) {\cmark};
		\node at ( 90:0.55)  [align=center] {per word prefix \\ \footnotesize{(tidy)}};
		\node at ( -165:1.5)[align=center] {\cmark \\ per-time};
		\node at ( -15:1.4) [align=center] {\cmark \\  \text{ } per-letter};
		\node at ( -90:0.4) [align=center] {\cmark \\ single};
		\draw (-30:1) ellipse (1.5 and 0.69);
		\draw (-150:1) ellipse (1.5 and 0.69);
	\end{tikzpicture}
	\caption[Classes of integral NMDAs, defined according to the flexibility of choosing the discount factors.]{\label{fig:introClasses}Classes of integral NMDAs, defined according to the flexibility of choosing the discount factors.
	The class of NMDAs with arbitrary integral factors is not closed under algebraic operations and under determinization, and some of its decision problems are undecidable. The other classes (for a specific choice function) are closed under both algebraic operations and determinization, and their basic decision problems are decidable. Tidy NMDAs are as expressive as deterministic NMDAs with arbitrary integral discount factors.}
\end{figure}

We start, in \Cref{sec:Definitions}, with formal definitions of NMDAs, after which we analyze, in  \Cref{sec:ArbitraryFactors}, the properties of arbitrary integral NMDAs, showing that they do not enjoy algebraic closure, and that their containment problems are undecidable.
In \Cref{sec:TidyNMDAs} we introduce tidy-NMDAs and show that they are closed under algebraic operations. We then analyze, in \Cref{sec:tidyDecisionProblems}, their decision problems, and show that they are in the same complexity classes as the corresponding problems for NDAs with a single integral discount factor. We conclude and provide suggestions for future work in \Cref{sec:Conclusions}.

\subsection*{Related work}
As we extend integral NDAs to allow for multiple discount factors, our work naturally relates to existing works on integral NDAs, as well as to works on other computational models that already allow for multiple discount factors.

Most relevant to our positive results on tidy NMDAs is the work in \cite{BH14}, which considers integral NDAs, and whose techniques we extend in \Cref{sec:Determinizability} for the determinization procedure.

Also relevant to our positive results is the approach of ``comparators'' \cite{BCVCAV18,BVCAV19,ComparatorAutomataInQuantitativeVerification}, which are automata that read two infinite sequences of weights synchronously and relate their aggregate values. 
In particular, the containment problem of NDAs was proved in \cite{ComparatorAutomataInQuantitativeVerification} to be in PSPACE, using comparators to reduce the problem to language inclusion between B\"{u}chi automata.
Our approach for the containment problem of NMDAs is different, based on on-the-fly determinization of the union of the two considered automata (\Cref{sec:PSPaceProblems}). Our algorithm improves the complexity provided in \cite{ComparatorAutomataInQuantitativeVerification} for NDAs (having a single discount factor), as we refer to binary representation of weights, while \cite{ComparatorAutomataInQuantitativeVerification} assumes unary representation.\footnote{Rational weights are assumed to have a common denominator, both by us and by \cite{ComparatorAutomataInQuantitativeVerification}, where in the latter it is stated implicitly, by providing the complexity analysis with respect to transition weights that are natural numbers.}

Considering other computational models that allow for multiple discount factors, most relevant to NMDAs are discounted-payoff games with multiple discount factors (DPGs) \cite{Andersson06}. The two models share a common basic variant: NMDAs over a singleton alphabet are the same as one-player DPGs. We take advantage of this relation, using DPG algorithms to solve the problems of NMDA nonemptiness (\Cref{sec:Nonemptiness}).
The core difference between a nondeterminstic automaton and a two-player game is that the former allows for richer alphabets and unrestricted nondeterminism, while the latter allows for alternating turns between the players.
Hence, the choices made by an automaton can be based on the entire (infinite) input word, while player choices in a game are restricted to \emph{strategies}, which can only depend on past events (see, e.g., \cite{BKKS13}). 
Due to this difference, automata problems tend to have higher complexities than related game problems, sometimes resulting in undecidability of the former, which is the case with containment of limit-average automata \cite{DDGRT10}, compared to decidable problems of mean-payoff games \cite{ZP96}, as well as with containment of NMDAs, which we show to be undecidable (\Cref{sec:integralNMDAsComparison}), compared to decidable problems of DPGs \cite{Andersson06}.

Considering our aforementioned undecidability result, we provide a reduction from the halting problem of two-counter machines, following known schemes \cite{DDGRT10,ABK22}. Yet the crux of our proof is in simulating a counter within a discounting setting, as upfront an increment of a counter at a certain point of time cannot be compensated by a far-away discounted decrement. Nevertheless, we show that multiple discount factors allow in a sense to eliminate the influence of time, constructing automata in which wherever a letter appears in the word, it has the same influence on the automaton value (\Cref{sec:integralNMDAsComparison}).

\section{Discounted-Sum Automata with Multiple Integral Discount Factors}\label{sec:Definitions}

We define a discounted-sum automaton with arbitrary discount factors, abbreviated NMDA, by adding to an NDA a discount factor in each of its transitions. An NMDA is defined on either finite or infinite words. The formal definition is given in \Cref{def:NMDA}, and an example in \Cref{fig:NMDAExample}.

An \emph{alphabet} $\Sigma$ is an arbitrary finite set, and a \emph{word} over $\Sigma$ is a finite or infinite sequence of letters in $\Sigma$, with $\emptyword$ for the empty word. We denote the concatenation of a finite word $u$ and a finite or infinite word $w$ by $u\con w$, or simply by $uw$.
We define $\Sigma^+$ to be the set of all finite words except the empty word, i.e., $\Sigma^+=\Sigma^*\setminus\{\emptyword\}$. For a word $w=w(0) w(1) w(2) \ldots$, we denote the sequence of its letters starting at index $i$ and ending at index $j$ by $w[i..j]=w(i) w(i{+}1)\ldots w(j)$, and in general, for integers $i\leq j$, we denote the set $\{i, i{+}1, \ldots, j\}$ by $[i..j]$.

\begin{defi}\label{def:NMDA}
	A nondeterministic discounted-sum automaton with multiple discount factors (NMDA), on finite or infinite words, is a tuple $\A = \tuple{\Sigma, Q, \iota, \delta, \gamma, \rho}$ over an alphabet $\Sigma$, with a finite set of states $Q$, an initial set of states $\iota\subseteq Q$, a transition function $\delta \subseteq Q \times \Sigma \times Q$, a weight function $\gamma: \delta\to\Rat$, 
	and a discount-factor function $\rho: \delta \to \Rat\cap (1,\infty)$, assigning to each transition its discount factor, which is a rational greater than one.\footnote{Discount factors are sometimes defined in the literature as numbers between $0$ and $1$, under which setting weights are multiplied by these factors rather than divided by them.}
	\begin{itemize}
		\item A \emph{walk} in $\A$ from a state $p_0$ is a sequence of states and letters, $p_0, \sigma_0, p_1, \sigma_1, p_2, \cdots$, such that for every
		$i$, 
		$(p_i,\sigma_i,p_{i+1})\in\delta$.
		
		For example, $\psi=q_1,a,q_1,b,q_2$ is a walk of the NMDA $\A$ of \Cref{fig:NMDAExample} on the word $ab$ from the state $q_1$ .
		
		\item A run of $\A$ is a walk from an initial state. 
		
		\item The length of a walk $\psi$, denoted by $|\psi|$, is $n$ for a finite walk $\psi = p_0, \sigma_0, p_1, \allowbreak\cdots,
		\sigma_{n-1}, p_n$, and $\infty$ for an infinite walk.
		\item The $i$-th transition of a walk $\psi = p_0, \sigma_0, p_1, \sigma_1, \cdots$ is denoted by $\psi(i)=(p_{i},\sigma_i,p_{i+1})$.
		
		\item The \emph{value} of a finite or an infinite walk $\psi$ is
		$\A(\psi)=\sum_{i=0}^{|\psi|-1}{ \bigg(\gamma\big(\psi(i)\big) \cdot \prod_{j=0}^{i-1} \frac{1}{\rho\big(\psi(j)\big)}\bigg)}$.
		For example, the value of the walk $r_1=q_0,a,q_0,a,q_1,b,q_2$ (which is also a run) of $\A$ from \Cref{fig:NMDAExample} is 
		$\A(r_1)= 1 + \frac{1}{2}\cdot\frac{1}{3} + 2\cdot\frac{1}{2\cdot 3}=\frac{3}{2}$.
		\item The \emph{value} of $\A$ on a finite or infinite word $w$ is
		$\A(w) = \inf \{\A(r) \St r \text{ is a run of } \A \text{ on } w \}$.
		\item 
		
		In the case where $|\iota|= 1$ and for every $q \in Q$ and $\sigma \in \Sigma$, we have $|\{ q' \ST (q,\sigma,q')\in\delta \}| \leq 1$, we say that $\A$ is {\em deterministic}, denoted by DMDA,
		and view $\delta$ as a function to states.
		
		\item
		When all the discount factors are integers, we say that $\A$ is an \emph{integral} NMDA.
		
		\item For a given NMDA $\A$, we sometimes denote its weight function by $\delta_{\A}$  and its discount-factor function by $\rho_{\A}$.
	\end{itemize}
\end{defi}

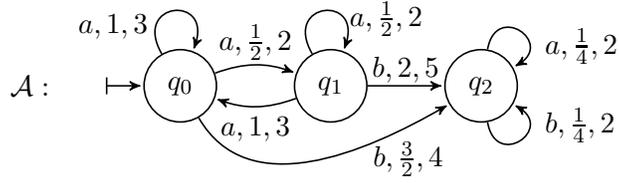
\begin{figure}
	\centering
	\begin{tikzpicture}[->,>=stealth',shorten >=1pt,auto,node distance=2cm, semithick, initial text=, every initial by arrow/.style={|->}]
		\node [midway] [ left of=q0] {$\A:$};
		\node[initial, state] (q0) {$q_0$};
		\node[state] (q1) [right of=q0] {$q_1$};
		\node[state] (q2) [right of=q1] {$q_2$};
		
		\path 
		(q0) edge	[loop above, out=120, in=70,looseness=5] node [left, xshift=-0.2cm, yshift=-0.2cm] {$a,1,3$} (q0)
		(q1) edge	[loop above, out=120,in=70,looseness=5] node [right, xshift=0.2cm, yshift=-0.1cm] {$a,\frac{1}{2},2$} (q1)
		(q2) edge	[loop above, out=80, in=30, looseness=4] node [right, xshift=0.2cm, yshift=-0.2cm]{$a,\frac{1}{4},2$} (q2)
		(q2) edge	[loop below, out=-80, in=-30,looseness=4] node [right, xshift=0.2cm, yshift=0.2cm]{$b,\frac{1}{4},2$} (q2)
		
		(q1) edge	[below, out=-160,in=-20] node {$a,1,3$} (q0)
		
		(q0) edge [above, out=20,in=160] node[yshift=-0.07cm] {$a,\frac{1}{2},2$} (q1)
		(q1) edge [above] node [yshift=-0.1cm] {$b,2,5$} (q2)
		(q0) edge	[below, out=-60,in=-150] node [xshift=1.3cm, yshift=0.4cm]{$b,\frac{3}{2},4$} (q2)
		;
	\end{tikzpicture}
	\caption[An NMDA example.]{\label{fig:NMDAExample}An NMDA $\A$. The labeling on the transitions indicate the alphabet letter, the weight of the transition, and its discount factor.}
\end{figure}
\noindent 
In the case where for every $q \in
Q$ and $\sigma \in \Sigma$, we have $|\{ q' \St (q,\sigma,q')\in\delta \}| \geq 1$, intuitively meaning that $\A$ cannot get stuck, we say that $\A$ is {\em complete}.
It is natural to assume that discounted-sum automata are complete, and we adopt this assumption, as dead-end states, which are equivalent to states with infinite-weight transitions, break the property of the decaying importance of future events. (Hence, the codomain of $\gamma$ is $\Rat$, not allowing for $\infty$.)

Automata $\A$ and $\A'$ are \emph{equivalent}, denoted by $\A\equiv \A'$, if for every word $w$, $\A(w) = \A'(w)$.

For every finite (infinite) walk $\psi = p_0, \sigma_0, p_1, \sigma_1, p_2, \cdots, \sigma_{n-1}, p_n$ ($\psi = p_0, \sigma_0, p_1, \cdots$),
and all integers $0 \leq i \leq j \leq |\psi|-1$ ($0 \leq i \leq j$), we define
the finite sub-walk from $i$ to $j$ as $\psi[i..j]=p_{i}, \sigma_i, p_{i+1}, \cdots, \sigma_j, p_{j+1}$.
For an infinite walk, we also define $\psi[i..\infty]=p_{i}, \sigma_i, p_{i+1}, \cdots$, namely the infinite suffix from position $i$.
For a finite walk, we also define the target state as $\delta(\psi)=p_n$ and the accumulated discount factor as $\rho(\psi)=\prod_{i=0}^{n-1}{\rho\big(\psi(i)\big)}$.

We extend the transition function $\delta$ to finite words in the regular manner: For a word $u\in\Sigma^*$ and a letter $\sigma\in\Sigma$, 
$
\delta(\emptyword)= \iota;
\delta(u\cdot \sigma)= \bigcup_{q\in\delta(u)}{\delta(q,\sigma)}
$.
For a state $q$ of $\A$, we denote by $\A^q$ the automaton that is identical to $\A$, except for having $q$ as its single initial state.

An NMDA may have rational weights, yet it is often convenient to consider an analogous NMDA with integral weights, achieved by multiplying all weights by their common denominator.
\begin{prop}\label{prop:multiply}
	For all constant $0<m\in \Rat$, NMDA $\A=\tuple{\Sigma,Q,\iota,\delta,\gamma,\rho}$, NMDA 
	$\A'=\tuple{\Sigma,Q,\iota,\delta,m\cdot \gamma,\rho}$ obtained from $\A$ by multiplying all its weights by $m$, and a finite or infinite word $w$, we have $\A'(w)=m\cdot \A(w)$.
\end{prop}

\begin{proof}
	Let $0<m\in \Rat$, $\A=\tuple{\Sigma,Q,\iota,\delta,\gamma,\rho}$ and $\A'=\tuple{\Sigma,Q,\iota,\delta,m\cdot \gamma,\rho}$ NMDAs, 
	and $w$ a finite or infinite word.
	
	For every run $r$ of $\A$ on $w$, we have that the same run in $\A'$ has the value of 
	\begin{align*}
		\A'(r) 
		&= \sum_{i=0}^{|w|-1}{ \Bigg(m\cdot \gamma(r(i)) \cdot \prod_{j=0}^{i-1} \frac{1}{\rho(r(j))}\Bigg)}
		= m\cdot \sum_{i=0}^{|w|-1}{ \Bigg( \gamma(r(i)) \cdot \prod_{j=0}^{i-1} \frac{1}{\rho(r(j))}\Bigg)}= m\cdot \A(r)
	\end{align*}
	
	Hence for every run of $\A$ with value $v_0$ we have a run of $\A'$ for the same word with value of $m\cdot v_0$. Symmetrically for every run of $\A'$ with value $v_1$ we have a run of $\A$ for the same word with value of $\frac{1}{m}\cdot v_1$. So,
	
	\begin{align*}
		\A'(w) &= \inf \{\A'(r) \ST r \text{ is a run of } \A' \text{ on } w \} 
		\geq \inf \{m\cdot \A(r) \ST r \text{ is a run of } \A \text{ on } w \} \\
		&= m\cdot \inf \{ \A(r) \ST r \text{ is a run of } \A \text{ on } w \} 
		= m \cdot \A(w)
	\end{align*} 
	and
	\begin{align*}
		\A(w) &= \inf \{\A(r) \ST r \text{ is a run of } \A \text{ on } w \} \\
		&\geq \inf \Big\{\frac{1}{m}\cdot \A'(r) \ST r \text{ is a run of } \A' \text{ on } w \Big\} = \frac{1}{m} \cdot \A'(w)
	\end{align*} 
	which leads to $\A'(w)=m\cdot \A(w)$.
\end{proof}

\paragraph{Size.} 
We define the size of $\A$, denoted by $|\A|$, as the maximum between the number of transitions and the maximal binary representation of any discount factor or weight in it. For rational weights, we assume all of them to have the same denominator. The motivation for a common denominator stems from the determinization algorithm (\Cref{thm:DetTidy}). Omitting this assumption will still result in a deterministic automaton whose size is only single exponential in the size of the original automaton, yet storing its states will require a much bigger space, changing our PSPACE algorithms (\Cref{sec:TidyNMDAs}) into EXPSPACE ones.

\paragraph{Algebraic operations.}
Given automata $\A$ and $\B$ over the same alphabet, and a non-negative scalar $c\in\Rat$, we define
\begin{itemize}
	\item $\C \equiv \min (\A, \B)$ if $\forall w$.	
	$\C(w) = \min \big(\A(w), \B(w)\big)$
	\item $\C \equiv \max (\A, \B)$ if $\forall w$.	
	$\C(w) = \max \big(\A(w), \B(w)\big)$
	\item $\C \equiv  \A + \B$ if $\forall w$.	
	$\C(w) = \A(w) + \B(w)$
	\item $\C \equiv  \A - \B$ if $\forall w$.	
	$\C(w) = \A(w) - \B(w)$
	\item $\C \equiv  c \cdot \A $ if $\forall w$.	
	$\C(w) = c \cdot \A(w)$
	\item $\C \equiv  -\A $ if $\forall w$.	
	$\C(w) = - \A(w)$
\end{itemize}

\paragraph{Decision problems.}
Given automata $\A$ and $\B$ and a threshold $\nu\in\Rat$, we consider the following properties, with strict (or non-strict) inequalities:
\begin{itemize}
\item 
\emph{Nonemptiness:} There exists a word $w$, s.t.\ $\A(w)<\nu$ (or $\A(w)\leq\nu$);
\item
\emph{Exact-value:} There exists a word $w$, s.t.\ $\A(w)=\nu$;
\item
\emph{Universality:} For all words $w$, $\A(w)<\nu$ (or $\A(w)\leq\nu$);
\item
\emph{Equivalence:} For all  words $w$, $\A(w) = \B(w)$;
\item
\emph{Containment:} For all  words $w$, $\A(w) > \B(w)$ (or $\A(w) \geq \B(w)$). \footnote{Considering quantitative containment as a generalization of language containment, and defining the ``acceptance'' of a word $w$ as having a small enough value on it, we define that $\A$ is contained in $\B$ if for every word $w$, $\A$'s value on $w$ is at least as big as $\B$'s value. (Observe the $>$ and $\geq$ signs in the definition.)}
\end{itemize}

\paragraph{Finite and infinite words.}
Results regarding NMDAs on finite words that refer to 
the existence of an equivalent automaton (``positive results'') can be extended to NMDAs on infinite words due to \Cref{lemma:finiteToInfinite} below.
Likewise, results that refer to non-existence of an equivalent automaton (``negative results'') can be extended from NMDAs on infinite words to NMDAs on finite words.
Accordingly, if not stated otherwise, we prove the positive results for automata on finite words and the negative results for automata on infinite words, getting the results for both settings.
\begin{lem}\label{lemma:finiteToInfinite}
	For all NMDAs $\A$ and $\B$, if for all finite word $u\in\Sigma^+$, we have $\A(u)=\B(u)$, then also for all infinite word $w\in\Sigma^\omega$, we have $\A(w)=\B(w)$.
\end{lem}
\begin{proof}
	The proof extends \cite[Lemma 3.3]{BH14} from NDAs to NMDAs.
	
	We start by making a key observation: For any NMDA $\C$ and every $\epsilon > 0$, there exists $n_\epsilon \in\Nat$, such that the contribution of any infinite suffix word to any $n_\epsilon$-sized prefix run, is less than $\epsilon$ in magnitude.
		This is seen as follows. Denote the supremum of the absolute value of $\C$ on any infinite word by $W=\sup_{w\in\Sigma^\omega}|\C (w)|$. Let $m$ be the highest absolute value of a transition weight in $\C$, and $\lambda$ be the lowest discount factor in $\C$. Note that $W$ cannot be higher than the case when we choose $m$ for the weight of all transitions and $\lambda$ as their discount factor, that is $W\leq\sum_{i=0}^{\infty}\frac{m}{\lambda^i}$. Since for every $\lambda>1$, we have  $\sum_{i=0}^{\infty}\frac{1}{\lambda^i}= \frac{1}{1-\frac{1}{\lambda}}=\frac{\lambda}{\lambda - 1}$, it follows that $W\leq m\cdot\frac{\lambda}{\lambda - 1}$.
	The maximal contribution of an infinite suffix word to an $n_\epsilon$-sized prefix run cannot thus be higher than $\frac{1}{\lambda^{n_\epsilon}}\cdot W \leq \frac{m}{\lambda^{n_\epsilon -1} \cdot(\lambda -1)}$. Since $m$ and $\lambda$ are fixed, this contribution can be made arbitrary small by choosing arbitrary large values of $n_\epsilon$.
	
	Now, suppose for the sake of contradiction that NMDAs $\A$ and $\B$ agree on all finite words, but there exists an infinite word $w$ such that $|\A (w) - \B (w)| = c >0$. W.l.o.g., we shall assume that $\A (w) < \B (w)$. We take $\epsilon = c/3$, and $n_\epsilon$ such that the contribution of any infinite suffix word to any $n_\epsilon$-sized prefix run is less than $\epsilon$ in both $\A$ and $\B$.
	Let $w_n$ be the $n_\epsilon$-sized prefix word of $w$. The value of the preferred run of $\A$ on $w_n$ cannot be more than $\A (w) + c/3$ while the value of the preferred run of $\B$ on $w_n$ cannot be less than $\B (w) - c/3$, resulting in $\A(w_n)\leq A(w) + c/3 < B(w) - c/3 \leq \B(w_n)$ and a contradiction.
\end{proof}

Notice that the converse of \Cref{lemma:finiteToInfinite} does not hold, namely there are automata equivalent w.r.t.\ infinite words, but not w.r.t.\ finite words. 
(See an example in \Cref{fig:lemmaConverse}.)

\begin{figure}
	\centering
	\begin{tikzpicture}[->,>=stealth',shorten >=1pt,auto,node distance=2.6cm, semithick, initial text=, every initial by arrow/.style={|->}]
		\node[initial, state] (q0) {$q_0$}; 
		\node[state] [left of=q0, draw=none,xshift=0.9cm] {$\B:$};
		
		\node[state] (q1) [right of=q0] {$q_1$};
		
		\node[initial, state] (p0) [left of=q0, xshift=-3cm] {$p_0$};
		
		\node[state] [left of=p0, draw=none,xshift=0.9cm] {$\A:$};		
		\path 
		
		(q0) edge node [above] {$\Sigma,2,2$} (q1)
		
		(q1) edge [right, in=30,out=-30, loop, looseness=4] node {$\Sigma,0,2$} (q1)

		(p0) edge [right, in=30,out=-30, loop, looseness=4] node {$\Sigma,1,2$} (p0)
		;
	\end{tikzpicture}
	\caption{\label{fig:lemmaConverse} \cite[Figure 2]{BH14} The automata $\A$ and $\B$ are equivalent with respect to infinite words, while not equivalent with respect to finite words.}
\end{figure}

\section{Arbitrary Integral NMDAs}\label{sec:ArbitraryFactors}

Unfortunately, we show that the family of integral NMDAs in which discount factors can be chosen arbitrarily is not closed under determinization (\Cref{sec:integralNMDAsNotDeterminizable}) and under basic algebraic operations (\Cref{sec:integralNMDAsNonClosureAlgebraic}), and its containment problem is undecidable (\Cref{sec:integralNMDAsComparison}); A summary of the concrete negative results is given in \Cref{tbl:generalIntegral}.

\begin{table}
	\begin{center}
		\begin{tabular}{|c||c|c|}
			\hline
			& Finite words & Infinite words \\ [0.5ex]
			\hline \hline
			Determinization
			& \multicolumn{2}{c|}{\xmark \text{} Not closed {\footnotesize (\Cref{thm:integralNMDAsNotDeterminizableInfinite})}}
			\\ [0.5ex]
			\hline
			Algebraic operations (max, addition)
			& \multicolumn{2}{c|}{\xmark \text{} Not closed {\footnotesize (\Cref{thm:NoClosureGeneralIntegralNMDA})}}
			\\ [0.5ex]
			\hline
			Containment ($>$) & \xmark \text{} Undecidable & \qmark \text{} Open question \\ [0.5ex]
			\cline{1-1}
			\cline{3-3}
			Containment ($\geq$) & {\footnotesize (\Cref{cl:ContainmentFiniteWordsUndecidable})} & \xmark \text{} Undecidable {\footnotesize (\Cref{cl:NonStrictContainmentInfiniteWordsUndecidable})} \\ [0.5ex]
			\hline
			Equivalence
			& \multicolumn{2}{c|}{\xmark \text{} Undecidable {\footnotesize (\Cref{cl:EquivalenceFiniteWordsUndecidable})}}
			\\ [0.5ex]
			\hline
		\end{tabular}
	\end{center}
	\caption{Negative results for the family of integral NMDAs in which discount factors can be chosen arbitrarily.} \label{tbl:generalIntegral}
\end{table}

\subsection{Non-closure under determinization}\label{sec:integralNMDAsNotDeterminizable}
\begin{thm}\label{thm:integralNMDAsNotDeterminizableInfinite}
	There exists an integral NMDA that no integral DMDA is equivalent to, with respect to both finite and infinite words.
\end{thm} 

\begin{proof}
	Let $\B$ be the integral NMDA depicted in \Cref{fig:NonDetNMDAInfinityWords} over the alphabet $\Sigma=\{a,b,c\}$. 
	We first show that for every $n\in\Nat$, $\B(a^{n}b^\omega)=1-\frac{1}{2^{n+1}}$ and $\B(a^{n}c^\omega)=1+\frac{1}{3^{n+1}}$.	
			
	Note that the only nondeterminism in $\B$ is in the option to start from either $q_0$ or $q_1$.
	Intuitively, for an infinite word for which the first non-$a$ letter is $b$, the best choice for $\B$ would be to start from $q_0$, while if the first non-$a$ letter is $c$, the best choice would be to start from $q_1$. 
	
	Formally, for each $n\in\Nat\setminus\{0\}$, observe that for the finite word $a^n$, the run $r_1$ starting from $q_0$ will have the accumulated value of $\B(r_1) = \sum_{k=0}^{n-1}\frac{1}{2}\cdot\frac{1}{2^k}=\frac{1}{2}\cdot\frac{1-\frac{1}{2^{n}}}{1-\frac{1}{2}}= 1-\frac{1}{2^{n}}$,
	and an accumulated discount factor of $2^n$; and the run $r_2$ starting from $q_1$ the value
	$\B(r_2)=\sum_{k=0}^{n-1}\frac{2}{3}\cdot\frac{1}{3^k}=\frac{2}{3}\cdot\frac{1-\frac{1}{3^{n}}}{1-\frac{1}{3}}= 1-\frac{1}{3^{n}}$, 
	and an accumulated discount factor of $3^n$. Thus,
	the value of $\B$, which is the minimum of the two runs, is
	$\B(a^n)=\min\Big\{1-\frac{1}{2^{n}},1-\frac{1}{3^{n}}\Big\}=1-\frac{1}{2^{n}}$.
	
	Accordingly, we have that for every $n\in\Nat$,
	\begin{align*}
		\B(a^{n}b^\omega)&=\min\Big\{
		1-\frac{1}{2^{n}}+\frac{1}{2}\cdot\frac{1}{2^n},
		1-\frac{1}{3^{n}}+2\cdot\frac{1}{3^n}
		\Big\} =1-\frac{1}{2^{n+1}}
		\\
		\B(a^{n}c^\omega)&=\min\Big\{
		1-\frac{1}{2^{n}}+2\cdot\frac{1}{2^n},
		1-\frac{1}{3^{n}}+\frac{4}{3}\cdot\frac{1}{3^n}
		\Big\} =1+\frac{1}{3^{n+1}}	
	\end{align*} 
	
	We continue with assuming toward contradiction the existence of an integral DMDA $\D=\tuple{\Sigma, Q_\D, p_0, \delta_\D, \gamma_\D, \rho_\D}$ such that $\B\equiv\D$.
	Suppose $\D$ reaches some state $s(n)=\delta_\D(a^n)$ upon reading $a^n$, whence it starts reading the suffixes $b^\omega$ and $c^\omega$. 
	We use the following notations:
	\begin{enumerate}
		\item Let the accumulated discount factor at this point be $\Pi_n=\rho_\D(a^n)$. Observe that, by definition, $\Pi_n$ must be an integer.
		\item Let the suffixes weights be $W_b(n)=\D^{s(n)}(b^\omega)$ and $W_c(n)=\D^{s(n)}(c^\omega)$. Observe that there are up to $|Q_\D|$ different (rational) values to all of $W_b(n)$ and $W_c(n)$, namely over all $n\in\Nat$, $W_b(n)$ and $W_c(n)$ depend only on the state $s(n)$, and not on $n$ itself.
	\end{enumerate}
	We now show a combinatorial claim to prove that it is impossible for $\D$ to yield the values $1-\frac{1}{2^{n+1}}$ and $1+\frac{1}{3^{n+1}}$, for all $n\in\Nat$. Under the above notation,
	\begin{align*}
		\D(a^{n}b^\omega)&=
		\D(a^n)+\frac{W_b(n)}{\Pi_n}=
		1-\frac{1}{2^{n+1}} \text{ , so }
		1-D(a^n)= \frac{W_b(n)}{\Pi_n}+\frac{1}{2^{n+1}} \text{ ; and}
		\\
		\D(a^{n}c^\omega)&=
		\D(a^n)+\frac{W_c(n)}{\Pi_n}
		=1+\frac{1}{3^{n+1}} \text{ , so }
		1-D(a^n)= \frac{W_c(n)}{\Pi_n} - \frac{1}{3^{n+1}}\text{ .}
	\end{align*}
	Hence,
	\begin{align*}
		\frac{W_b(n)}{\Pi_n}+\frac{1}{2^{n+1}} = \frac{W_c(n)}{\Pi_n}-\frac{1}{3^{n+1}}\ ,
	\end{align*}
	from which it follows that
	\begin{align*}
		\Pi_n = \Big(W_c(n)-W_b(n) \Big)\cdot \frac{2^{n+1}\cdot 3^{n+1}}{2^{n+1}+3^{n+1}}\ .
	\end{align*}
	
	Since for every $n$, $\Pi_n$ must be an integer, it follows that for every $n$, the denominator  $F(n) = 2^{n+1}+3^{n+1}$ divides the numerator. Since neither $2$ nor $3$ divides $F(n)$, it is clear that $F(n)$ does not share any prime divisors with $2^{n+1}\cdot 3^{n+1}$. Thus, for every $n$, it is necessary that $F(n)$ divides $W_c(n)-W_b(n)$. However, the latter only takes up to $|Q_\D|$ different values, over all $n\in\Nat$, and can never be zero. Hence, it cannot be divisible by $F(n)$ for arbitrary large $n$, leading to contradiction.
	By \Cref{lemma:finiteToInfinite}, we also have that no DMDA is equivalent to $\B$ with respect to finite words.	
\end{proof}

	\begin{figure}
	\centering
	\begin{tikzpicture}[->,>=stealth',shorten >=1pt,auto,node distance=2cm, semithick, initial text=, every initial by arrow/.style={|->}]
		\node [state] [left of=q0, xshift=0.5cm, draw=none] {$\B:$};
		\node[initial left, state] (q0) {$q_0$};

		\node[state] (q2) [right of=q0, xshift=0.8cm] {$q_2$};
		\node[initial right, state] (q1) [right of=q2, xshift=0.8cm] {$q_1$};

		\path 
		(q0) edge	[loop above, out=120, in=70,looseness=5] node [above left, xshift=-0.1cm, yshift=-0.2cm] {$a,\frac{1}{2},2$} (q0)
		(q1) edge	[loop above, out=70, in=120,looseness=5] node [above right, xshift=0.2cm, yshift=-0.2cm] {$a,\frac{2}{3},3$} (q1)

		(q2) edge[loop above, out=120, in=70,looseness=5] node [left,xshift=-0.18cm, yshift=-0.15cm, align=center] {$a,0,2$\\[-3pt]$b,0,2$\\[-3pt]$c,0,2$} (q2)
		(q0) edge [above right] node [below, align=center] {$b,\frac{1}{2},2$\\$c,2,2$} (q2)
		(q1) edge [below left] node [below, align=center] {$b,2,3$\\$c,\frac{4}{3},3$} (q2)
		
		;
	\end{tikzpicture}
	\caption[An integral NMDA that cannot be determinized.]{\label{fig:NonDetNMDAInfinityWords}An integral NMDA $\B$ on infinite words that cannot be determinized.}
\end{figure}
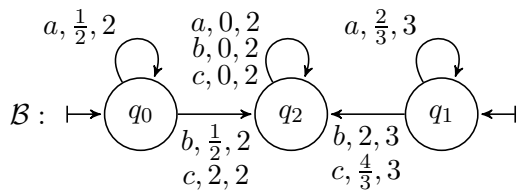

\subsection{Non-closure under algebraic operations}\label{sec:integralNMDAsNonClosureAlgebraic}
In the following proof that integral NMDAs are not closed under algebraic operations, we cannot assume toward contradiction a candidate deterministic automaton, and thus, as opposed to the proof of \Cref{thm:integralNMDAsNotDeterminizableInfinite}, we cannot assume a specific accumulative discount factor for each word prefix. 
Yet, we analyze the behavior of a candidate nondeterministic automaton on an infinite series of words, and build on the observation that there must be a state that appears in ``the same position of the run'' in infinitely many optimal runs of the automaton on these words.

\begin{thm}\label{thm:NoClosureGeneralIntegralNMDA}
	There exist integral NMDAs (even deterministic integral NDAs) $\A$ and $\B$ over the same alphabet, such that no integral NMDA is equivalent to $\max(\A,\B)$, and no integral NMDA is equivalent to $\A+\B$, with respect to both finite and infinite words.
\end{thm}
\begin{proof}
We show the result with respect to infinite words, and it also holds by  \Cref{lemma:finiteToInfinite} to finite words.	
Consider the NMDAs $\A$ and $\B$ depicted in \Cref{fig:noMaxGeneral}, and
assume towards contradiction that there exists an integral NMDA $\C'$ such that for every $n\in\Nat$,
$$\C'(a^n b^\omega)=\max(\A,\B)(a^n b^\omega)=\Big(\A+\B\Big)(a^n b^\omega)=
\begin{cases}
	\frac{1}{2^{n}} & n \text{ is odd}\\
	\frac{1}{3^{n}} & n \text{ is even}
\end{cases}
$$	
\begin{figure}
	\centering
	\begin{tikzpicture}[->,>=stealth',shorten >=1pt,auto,node distance=2.6cm, semithick, initial text=, every initial by arrow/.style={|->}]
		\node[initial, state] (q0) {$q_0$}; 
		\node[state] [left of=q0, draw=none,xshift=0.9cm] {$\B:$};
		
		\node[state] (q1) [right of=q0] {$q_1$};
		\node[state] (hole) [below of=q1, yshift=1cm] {$q_3$};
		
		\node[state] (q2) [right of=q1] {$q_2$};
		
		\node[state] (p1) [left of=q0, xshift=-1cm] {$p_1$};
		
		\node[initial, state] (p0) [left of=p1] {$p_0$};
		
		\node[state] [left of=p0, draw=none,xshift=0.9cm] {$\A:$};
		\node[state] (hole1) [below of=p0, yshift=1cm] {$p_2$};
		
		\path 
		
		(q1) edge [in=150,out=30] node [below, yshift=0.08cm]{$a,\frac{1}{3},3$} (q2)
		(q2) edge [above, in=-30,out=-150] node [above, yshift=-0.07cm] {$a,-1,3$} (q1)
		(q0) edge node [below] {$a,0,3$} (q1)
		
		(q1) edge node [left]{$b,0,3$} (hole)
		(q2) edge [bend left=20] node [right, xshift=0.2cm] {$b,0,3$} (hole)
		(hole) edge [right, in=-60,out=0, loop, looseness=4] node [align=center,yshift=0.1cm]{$a,0,3$\\[-3pt]$b,0,3$} (hole)

		(p0) edge [in=150,out=30] node [below, yshift=0.08cm] {$a,\frac{1}{2},2$} (p1)
		(p1) edge [above, in=-30,out=-150] node [above, yshift=-0.07cm] {$a,-1,2$} (p0)
		
		(p0) edge node [left]{$b,0,2$} (hole1)
		(p1) edge [bend left=20] node [right, xshift=0.2cm] {$b,0,2$} (hole1)
		(hole1) edge [right, in=-60,out=0, loop, looseness=4] node [align=center,yshift=0.1cm]{$a,0,2$\\[-3pt]$b,0,2$} (hole1)
		;
	\end{tikzpicture}
	\caption{\label{fig:noMaxGeneral}Deterministic integral NDAs that no integral NMDA is equivalent to their max or addition.}
\end{figure}
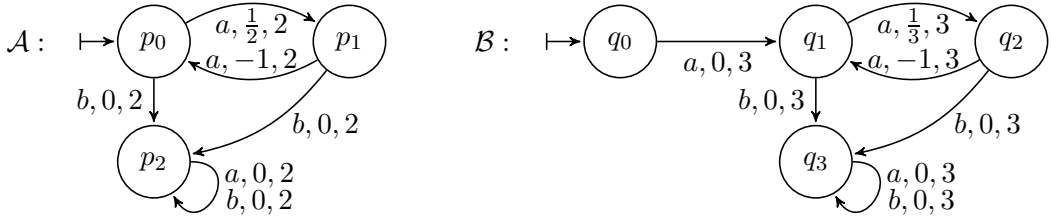
	Let $d\in\Nat$ be the least common denominator of the weights in $\C'$, and
consider the NMDA $\C=\tuple{\Sigma,Q,\iota,\delta,\gamma,\rho}$ created from $\C'$ by multiplying all its weights by $d$. Observe that all the weights in $\C$ are integers. 
According to \Cref{prop:multiply}, for every $n\in\Nat$, we have
$\C(a^n b^\omega)=d\cdot \C'(a^n b^\omega)=
\begin{cases}
	\frac{d}{2^{n}} & n \text{ is odd}\\
	\frac{d}{3^{n}} & n \text{ is even}
\end{cases}
$

For every even $n\in\Nat$, let $w_n=a^n b^\omega$, and $r_n$ a run of $\C$ on $w_n$ that entails the minimal value of $\frac{d}{3^{n}}$.
Since $\C$ is finite, there exists a state $q\in Q$ such that for infinitely many even $n\in\Nat$, the target state of $r_n$ after $n$ steps is $q$, i.e, $\delta(r_n[0..n-1])=q$.
We now show that the difference between $U_b=\C^q(b^\omega)$ and $U_a=\C^q(a\con b^\omega)$, the weights of the $b^\omega$ and $a\con b^\omega$ suffixes starting at $q$, discounted by $\Pi_n=\rho(r_n[0..n-1])$, which is the accumulated discount factor of the prefix of $r_n$ up to $q$, is approximately $\frac{1}{2^n}$ (See \Cref{fig:noClosureProof} for the notations). 
Since the weights of the prefixes are constant, for large enough $n$ we will conclude that $m_1\cdot 2^n\geq \Pi_n$ for some positive constant $m_1$.
\begin{figure}
	\centering{}
	\begin{tikzpicture}[->,>=stealth',shorten >=1pt,auto,node distance=3cm, semithick, initial text=, every initial by arrow/.style={|->}]
		\node [midway] [left of=q, xshift=2cm] {$\C:$};
		\node[state, draw=none] (q0) {};
		\node[state, inner sep=0.1cm,minimum size=0.4cm] (q) [right of=q0] {$q$}; 
		\node[state, draw=none] (qinf) [above right of=q, yshift=-1.5cm] {}; 
		\node[state, draw=none] (qinf2) [below right of=q, yshift=1.5cm] {}; 
		\path 
		(q0) edge [out=0,in=160, swap, looseness=1] node [above] {$a^n,W_n,\Pi_n$} (q)
		(q) edge node [above left, near end, xshift=0.5cm] {$b^\omega,U_b,-$} (qinf)
		(q) edge node [below left, near end, xshift=1cm, yshift=-0.1cm] {$a\con b^\omega,U_a,-$} (qinf2)
		;
	\end{tikzpicture}
	\caption[The notations from the proof of \Cref{thm:NoClosureGeneralIntegralNMDA}.]{\label{fig:noClosureProof}The state $q$ and the notations from the proof of \Cref{thm:NoClosureGeneralIntegralNMDA}, for two different even $n\in\Nat$ such that $\delta(r_n[1..n])=q$.	
		The labels on the walks indicate the input word and the accumulated weight and discount factors.}
\end{figure}
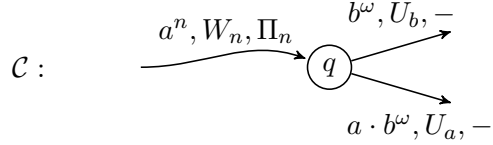

For every such $n\in\Nat$, let 
$W_n=\C(r_n[0..n-1])$, and since $\C(r_n) = \frac{d}{3^{n}}$, we have
\begin{align}W_n + \frac{U_b}{\Pi_n}=\frac{d}{3^n}
	\label{eqn:noClosureUb}
\end{align}  
Since the value of every run of $\C$ on $a^{n+1} b^\omega$ is at least $\frac{d}{2^{n+1}}$, we have 
$W_n+\frac{U_a}{\Pi_n}\geq \frac{d}{2^{n+1}} $. 
Hence,
$
\frac{d}{3^n}-\frac{U_b}{\Pi_n}+\frac{U_a}{\Pi_n}\geq \frac{d}{2^{n+1}}
$
resulting in
$
\frac{U_a-U_b}{\Pi_n}\geq d\cdot\Big(\frac{1}{2^{n+1}}-\frac{1}{3^{n}}\Big)
$.
But for large enough $n$, we have $3^n>2^{n+2}$, hence we get
$\frac{1}{2^{n+1}}-\frac{1}{3^n}>\frac{1}{2^{n+1}}-\frac{1}{2^{n+2}}=\frac{1}{2^{n+2}}$, resulting in $\frac{U_a-U_b}{d} \cdot 2^{n+2} \geq\Pi_n $. And indeed, there exists a positive constant $m_1=\frac{U_a-U_b}{d} \cdot 2^{2}$ such that
$ m_1 \cdot 2^n \geq\Pi_n $.

Now, $U_b$ is a rational constant, otherwise \Cref{eqn:noClosureUb} cannot hold, as the other elements are rationals. Hence, there exist $x\in\mathbb{Z}$ and $y\in\Nat$ such that $U_b=\frac{x}{y}$, and
$
\frac{1}{3^n}=\frac{W_n\cdot \Pi_n + U_b}{d\cdot \Pi_n}=
\frac{W_n\cdot \Pi_n + \frac{x}{y}}{d\cdot \Pi_n}=
\frac{W_n\cdot \Pi_n\cdot y+x}{d\cdot y\cdot \Pi_n}
$ .
Since the denominator and the numerator of the right-hand side are integers, we conclude that there exists a positive constant $m_2=d\cdot y$, such that $m_2 \cdot \Pi_n \geq 3^n$. Eventually, we get $m_1\cdot m_2 \cdot 2^n\geq 3^n$, for some positive constants $m_1$ and $m_2$, and for infinitely many $n\in\Nat$. But this stands in contradiction with $\lim_{n\to\infty}{\Big(\frac{2}{3}\Big)}^n=0$.
\end{proof}

Observe that DMDAs can be complemented, i.e., the complements $-\A,-\B$ of the DMDAs depicted in \Cref{fig:noMaxGeneral} can easily be constructed (by negating all the weights). Also, we can easily construct an NMDA for $\min(-\A,-\B)$ (by joining both automata). However, the complement of $\min(-\A,-\B)$ does not exist, for if it did, it would precisely be $\max(\A,\B)$.

\subsection{Undecidability of the containment problem}\label{sec:integralNMDAsComparison}
We show that it is undecidable to resolve the equivalence and containment problems of integral NMDAs. More precisely, for given integral NMDA $\N$ and integral DMDA $\D$, on both finite and infinite words, it is undecidable to resolve whether $\N \leq \D$ (\Cref{cl:ContainmentFiniteWordsUndecidable,cl:NonStrictContainmentInfiniteWordsUndecidable}), and on finite words it is also undecidable to resolve  whether $\N < \D$ (\Cref{cl:ContainmentFiniteWordsUndecidable}).
For given integral NMDAs $\N_1$ and $\N_2$, on both finite and infinite words, it is undecidable to resolve whether $\N_1 \equiv \N_2$ (\Cref{cl:NonStrictContainmentInfiniteWordsUndecidable}).
We also sketch, in \Cref{rem:PCP}, the undecidability of a problem we do not define in \Cref{sec:Definitions} and do not formally consider in the paper: Given integral DMDAs $\A$ and $\B$, does there exist a finite word $w$, such that $\A(w)=\B(w)$? 

We prove the undecidability result by reduction from the halting problem of two-counter machines.
The general scheme follows similar reductions, such as in \cite{DDGRT10,ABK22}, yet the crux is in simulating a counter by integral NMDAs. 
Upfront, discounted summation is not suitable for simulating counters, since a current increment has, in the discounted setting,  a much higher influence than of a far-away decrement. However, we show that multiple discount factors allow in a sense to eliminate the influence of time, having automata in which no matter where a letter appears in the word, it will have the same influence on the automaton value. (See \Cref{lem:undecidabilityContainment,fig:undecidabilityContainment_A}). Another main part of the proof is in showing how to nondeterministically adjust the automaton weights and discount factors in order to ``detect'' whether a counter is at a current value $0$. (See \Cref{fig:NegativeCounter,fig:balancedCounters,fig:zeroJumpChecker,fig:positiveJumpChecker}.)

We start with introducing the halting problem of two-counter machines (\Cref{sec:twoCountersMachines}), continue with a lemma on the accumulated value of certain series of discount factors and weights (\Cref{sec:specialWeightsAndFactors}), present the reduction (\Cref{sec:TheReduction}) and show the undecidability proof (\Cref{sec:Undecidability}).

\subsubsection{Two-counter machines}\label{sec:twoCountersMachines}
A two-counter machine \cite{Min67} $\M$ is a sequence $(l_1,\ldots,l_n)$ of commands, for some $n\in\Nat$, involving two counters $x$ and $y$. We refer to
$\set{1,\ldots,n}$ as the {\em locations} of the machine. For every $i\in\set{1,\ldots,n}$ we refer to $l_i$ as the {\em command in location $i$}. There are five possible forms of commands:
$$\inc(c),\ \dec(c),\ \goto l_k,\  \jz{c}{l_k}{l_{k'}},\  \halt,$$
where $c\in \set{x,y}$ is a counter and $1\le k,k'\le n$ are locations. 
For not decreasing a zero-valued counter $c\in\set{x,y}$, every $\dec(c)$ command is preceded by the command  $\jz{c}{\text{<current\_line>}}{\text{<next\_line>}}$, and there are no other direct goto-commands to it.\footnote{Notice that this conditional-blocking command keeps the model's halting problem undecidable -- if the original program properly halts then it does not decrease a zero counter, so there is no blocking, and the adapted program also halts; and if the original program does not halt then so does the adapted one, either because of blocking or because of following the original program without blocking.}
The counters are initially set to $0$.
An example of a two-counter machine is given in \Cref{fig:machineExample}.
\begin{figure}
	\vspace*{-\baselineskip}
	\centering
	\setlength{\belowcaptionskip}{-\baselineskip}
	\fbox{\begin{minipage}{20em}
			\begin{enumerate}[itemsep=0pt]
				\item[$l_1$.] $\inc(x)$
				\item[$l_2$.] $\inc(x)$
				\item[$l_3$.] $\jz{x}{l_3}{l_4}$
				\item[$l_4$.] $\dec(x)$
				\item[$l_5$.] $\jz{x}{l_6}{l_3}$
				\item[$l_6$.] $\halt$
			\end{enumerate}
	\end{minipage}}
	\caption{\label{fig:machineExample}An example of a two-counter machine.}
\end{figure}

Let $L$ be the set of possible commands in $\M$, then a {\em run} of $\M$ is a sequence
$\CMrun=\CMrun_1,\ldots,\CMrun_m\in (L\times\Nat\times\Nat)^*$ such that the following holds:
\begin{enumerate}
	\item $\CMrun_1=\tuple{l_1,0,0}$.
	\item For all $1< i\le m$, let $\CMrun_{i-1}=(l_j,\alpha_x,\alpha_y)$ and $\CMrun_{i}=(l',\alpha_x',\alpha_y')$. Then, the following hold.
	\begin{itemize}
		\item If $l_j$ is an $\inc(x)$ command (resp. $\inc(y)$), then $\alpha_x'=\alpha_x+1$, $\alpha_y'=\alpha_y$ (resp. $\alpha_y=\alpha_y+1$, $\alpha_x'=\alpha_x$), and $l'=l_{j+1}$.
		\item If $l_j$ is $\dec(x)$ (resp. $\dec(y)$) then $\alpha_x'=\alpha_x-1$, $\alpha_y'=\alpha_y$ (resp. $\alpha_y=\alpha_y-1$, $\alpha_x'=\alpha_x$), and $l'=l_{j+1}$.
		\item If $l_j$ is $\goto l_k$ then $\alpha_x'=\alpha_x$, $\alpha_y'=\alpha_y$, and $l'=l_k$.
		\item If $l_j$ is $\jz{x}{l_k}{l_{k'}}$ then $\alpha_x'=\alpha_x$, $\alpha_y'=\alpha_y$, and $l'=l_k$ if $\alpha_x=0$, and $l'=l_{k'}$ otherwise.
		\item If $l_j$ is $\jz{y}{l_k}{l_{k'}}$ then $\alpha_x'=\alpha_x$, $\alpha_y'=\alpha_y$, and $l'=l_k$ if $\alpha_y=0$, and $l'=l_{k'}$ otherwise.		
		\item If $l'$ is $\halt$ then $i=m$, namely a run does not continue after $\halt$.
	\end{itemize}
\end{enumerate}
If, in addition, we have that $\CMrun_m=\tuple{l_j,\alpha_x,\alpha_y}$ such that $l_j$ is a $\halt$ command, we say that $\CMrun$ is a \emph{halting run}. We say that a machine $\M$ 0-halts if its run is halting and ends in $\tuple{l,0,0}$.
We say that a sequence of commands $\tau\in L^*$ {\em fits} a run $\CMrun$, if $\tau$ is the projection of $\CMrun$ on its first component.

The {\em command trace} $\pi=\sigma_1,\ldots,\sigma_{m}$ of a halting run $\CMrun=\CMrun_1,\ldots,\CMrun_m$ describes the flow of the run, including a description of whether a counter $c$ was equal to $0$ or larger than $0$ in each occurrence of an $\jz{c}{l_k}{l_{k'}}$ command. It is formally defined as follows.
$\sigma_{m}=\halt$ and for every $1< i\le m$, we define $\sigma_{i-1}$ according to $\CMrun_{i-1}=(l_j,\alpha_x,\alpha_y)$ in the following manner:
\begin{itemize}
	\item $\sigma_{i-1}=l_j$ if $l_j$ is not of the form $\jz{c}{l_k}{l_{k'}}$.
	\item 
	$\sigma_{i-1}=(\goto l_k,c=0)$ for $c\in\{x,y\}$, if $\alpha_c=0$ and the command $l_j$ is of the form $\jz{c}{l_k}{l_{k'}}$.
	\item 
	$\sigma_{i-1}=(\goto l_{k'},c>0)$ for $c\in\{x,y\}$, if $\alpha_c>0$ and the command $l_j$ is of the form $\jz{c}{l_k}{l_{k'}}$.
\end{itemize}
For example, the command trace of the halting run of the machine in \Cref{fig:machineExample} is $\inc(x)$, $\inc(x)$, $(\goto l_4, x>0)$, $\dec(x)$, $(\goto l_3, x>0)$, $(\goto l_4, x>0)$, $\dec(x)$, $(\goto \! l_6, x=0)$, $\halt $.

Deciding whether a given counter machine $\M$ halts is known to be undecidable \cite{Min67}. Deciding whether $\M$ halts with both counters having value $0$, termed the {\em $0$-halting problem}, is also undecidable. 
Indeed, the halting problem can be reduced to the latter by adding some commands that clear the counters, before every \halt command. 

\subsubsection{Auxiliary lemma for simulating counters}\label{sec:specialWeightsAndFactors}
We present a lemma on the accumulated value of certain series of discount factors and weights. Observe that by the lemma, no matter where the pair of discount-factor $\lambda\in\Nat\setminus\{0,1\}$ and weight $w=\frac{\lambda-1}{\lambda}$ appear along the run, they will have the same effect on the accumulated value. This property will play a key role in simulating counting by NMDAs.
\begin{lem}\label{lem:undecidabilityContainment}
	For every sequence $\lambda_1,\cdots,\lambda_{m}$ of integers larger than $1$ and weights $w_1,\cdots,w_{m}$ such that $w_i=\frac{\lambda_i-1}{\lambda_i}$, we have
	$\sum_{i=1}^{m}{ \big(w_i \cdot \prod_{j=1}^{i-1} \frac{1}{\lambda_j}\big)}=
	1-\frac{1}{\prod_{j=1}^{m}\lambda_j}
	$.
\end{lem}
\begin{proof}
	We show the claim by induction on $m$.
	
	\noindent
	The base case, i.e., $m=1$, is trivial.
	For the induction step we have
	\begin{align*}
		\sum_{i=1}^{m+1}{ \big(w_i \cdot \prod_{j=1}^{i-1} \frac{1}{\lambda_j}\big)}&=
		\sum_{i=1}^{m}{ \big(w_i \cdot \prod_{j=1}^{i-1} \frac{1}{\lambda_j}\big)} + w_{m+1} \cdot \prod_{j=1}^{m} \frac{1}{\lambda_j}\\
		&= 1-\frac{1}{\prod_{j=1}^{m}\lambda_j} + \frac{\lambda_{m+1}-1}{\lambda_{m+1}} \cdot \prod_{j=1}^{m} \frac{1}{\lambda_j}\\
		&= 1-\frac{\lambda_{m+1}}{\prod_{j=1}^{m+1}\lambda_j} + \frac{\lambda_{m+1}-1}{\prod_{j=1}^{m+1}\lambda_j}
		= 1-\frac{1}{\prod_{j=1}^{m+1}\lambda_j}
		\qedhere
	\end{align*}
\end{proof}

\subsubsection{The Reduction}\label{sec:TheReduction}
We turn to our reduction from the halting problem of two-counter machines to the problem of NMDA containment. ``Halting with zero counter values'' means that on each counter, there are as many increment operations as there are decrement operations. We can detect violations of such a cumulative property courtesy of \Cref{lem:undecidabilityContainment}.
Furthermore, the assumption that over the entire run, the increments and decrements balance out, allows us to also detect violations in the control flow. Critically, after every branch taken due to a counter being 0, there must be as many increments to that counter as decrements. Dually, after every branch taken due to a counter being positive, there must be more decrements to that counter than increments. If there is a violation, the offending misprediction can be signaled out, to entail a less expensive run on the checking NMDA.

We provide below the construction and the correctness lemma with respect to automata on finite words. We later show, in \Cref{sec:Undecidability}, how to use the same construction also for automata on infinite words.

Given a two-counter machine $\M$ with the commands $(l_1,\ldots,l_n)$,
we construct an integral DMDA $\A$ and an integral NMDA $\B$ on finite words, such that $\M$ $0$-halts iff there exists a word $w\in\Sigma^+$ such that $\B(w)\geq \A(w)$ iff there exists a word $w\in\Sigma^+$ such that $\B(w) > \A(w)$.

The automata $\A$ and $\B$ operate over the following alphabet $\Sigma$, which consists of   $5n+5$ letters, standing for the possible elements in a command trace of $\M$:
\begin{align*}
	\incdec = \ &\set{\inc(x),\dec(x),\inc(y),\dec(y)} \\ 
	\allgoto =\ &\big\{\goto\ l_k: k\in \{1,\ldots,n\}\big\}\cup\\
	&\big\{(\goto\ l_k,c=0): k\in \{1,\ldots,n\},c\in\{x,y\}\big\}\cup\\
	&\big\{(\goto\ l_{k'},c>0): k'\in \{1,\ldots,n\},c\in\{x,y\}\big\} \\
	\withouthalt =\ &\incdec \cup \allgoto \\
	\Sigma =\ &\withouthalt \cup \big\{\halt\big\}
\end{align*}

When $\A$ and $\B$ read a word $w\in\Sigma^+$, they intuitively simulate a sequence of commands $\tau_u$ that induces the command trace $u=\pref_{\halt}(w)$. 
If $\tau_u$ fits the actual run of $\M$, and this run 0-halts, then the minimal run of $\B$ on $w$ has a value strictly larger than $\A(w)$. 
If, however, $\tau_u$ does not fit the actual run of $\M$, or it does fit the actual run but it does not 0-halt, then the violation is detected by $\B$, which has a run on $w$ with value strictly smaller than $\A(w)$.

In the construction, we use the following partial discount-factor functions $\rho_p,\rho_d:\withouthalt\to \Nat$ and partial weight functions $\gamma_p,\gamma_d:\withouthalt\to \Rat$. 
$$
\rho_p(\sigma)=\begin{cases}
	5 & \sigma=\inc(x)\\
	4 & \sigma=\dec(x)\\
	7 & \sigma=\inc(y)\\
	6 & \sigma=\dec(y)\\
	15 & \text{otherwise}
\end{cases} ~~~~~
\rho_d(\sigma)=\begin{cases}
	4 & \sigma=\inc(x)\\
	5 & \sigma=\dec(x)\\
	6 & \sigma=\inc(y)\\
	7 & \sigma=\dec(y)\\
	15 & \text{otherwise}
\end{cases}
$$
$\gamma_p(\sigma)=\frac{\rho_p(\sigma)-1}{\rho_p(\sigma)}$, and $\gamma_d(\sigma)=\frac{\rho_d(\sigma)-1}{\rho_d(\sigma)}$.
We say that $\rho_p$ and $\gamma_p$ are the \emph{primal} discount-factor and weight functions, while $\rho_d$ and $\gamma_d$ are the \emph{dual} functions.
Observe that for every $c\in\{x,y\}$ we have that 
\begin{align}
	\rho_p(\inc(c))=\rho_d(\dec(c))>\rho_p(\dec(c))=\rho_d(\inc(c)) \label{eqn:primalDual}
\end{align}

Intuitively, we will use the primal functions for $\A$'s discount factors and weights, and the dual functions for identifying violations.
Notice that if changing the primal functions to the dual ones in more occurrences of $\inc(c)$ letters than of $\dec(c)$ letters  along some run, then by \Cref{lem:undecidabilityContainment} the run will get a value lower than the original one.

We continue with their formal definitions.
$\A=\tuple{\Sigma,\{q_\A,q_\A^h\},\{q_\A\},\delta_\A,\gamma_\A,\rho_\A}$ is an integral DMDA consisting of two states, as depicted in \Cref{fig:undecidabilityContainment_A}.
Observe that the initial state $q_\A$ has self loops for every alphabet letter in $\withouthalt$ with weights and discount factors according to the primal functions, and a transition $(q_\A,\halt, q_\A^h)$ with weight of $\frac{14}{15}$ and a discount factor of $15$.

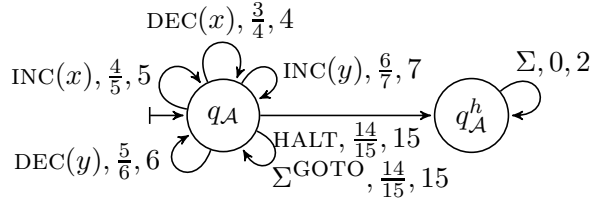
\begin{figure}
	\vspace*{-\baselineskip}
	\centering
	\setlength{\belowcaptionskip}{-\baselineskip}
	\begin{tikzpicture}[->,>=stealth',shorten >=1pt,auto,node distance=2cm, semithick, initial text=, every initial by arrow/.style={|->}]
		\node[initial, state] (q0) {$q_\A$};
		\node[state] (q1) [right of=q0, xshift=1.3cm] {$q_{\A}^h$};
		\path
		(q0) edge [loop left, in=125, out=170, looseness=5] node {$\inc(x),\frac{4}{5},5$} (q0)
		(q0) edge [loop above, in=70, out=115, looseness=5] node [yshift=-0.07cm]{$\dec(x),\frac{3}{4},4$} (q0)
		(q0) edge [loop right, in=25, out=60, looseness=5] node {$\inc(y),\frac{6}{7},7$} (q0)
		(q0) edge [loop right, in=-60, out=-25, looseness=5] node [yshift=-0.25cm, xshift=-0.15cm] {$\allgoto,\frac{14}{15},15$} (q0)
		(q0) edge [loop left, in=-150, out=-110, looseness=5] node [xshift=-0.1cm, yshift=0.1cm] {$\dec(y),\frac{5}{6},6$} (q0)
		
		(q0) edge [below] node [yshift=0.05cm] {$\halt,\frac{14}{15},15$} (q1)
		(q1) edge [loop above, in=0, out=40, looseness=5] node [xshift = 0.2cm, yshift=0.05cm]{$\Sigma,0,2$} (q1)
		;
	\end{tikzpicture}
	\caption{\label{fig:undecidabilityContainment_A}The DMDA $\A$ constructed for the proof of \Cref{cl:undecidabilityContainment}.}
\end{figure}

The integral NMDA $\B=\tuple{\Sigma,Q_\B,\iota_\B,\delta_\B,\gamma_\B,\rho_\B}$ is the union of the following eight gadgets (checkers), each responsible for checking a certain type of violation in the description of a 0-halting run of $\M$.
It also has the states $\qfr,\qhalt\in Q_\B$ such that for all $\sigma\in \Sigma$, there are 0-weighted transitions $(\qfr,\sigma,\qfr)\in\delta_\B$ and $(\qhalt,\sigma,\qhalt)\in\delta_\B$ with an arbitrary discount factor.
Observer that in all of $\B$'s gadgets, the transition over the letter \halt to $\qhalt$ has a weight higher than the weight of the corresponding transition in $\A$, so that when no violation is detected, the value of $\B$ on a word is higher than the value of $\A$ on it.

\checker{1. Halt}{}
This gadget, depicted in \Cref{fig:haltChecker}, checks for violations of non-halting runs.
Observe that its initial state $\qhc$ has self loops identical to those of $\A$'s initial state, a transition to $\qhalt$ over \halt with a weight higher than the corresponding weight in $\A$, and a transition to the state $q_{\mathsf{last}}$ over every letter that is not \halt, ``guessing'' that the run ends without a \halt command.

\begin{figure}
	\vspace*{-\baselineskip}
	\centering
	\setlength{\belowcaptionskip}{-\baselineskip}
	\begin{tikzpicture}[->,>=stealth',shorten >=1pt,auto,node distance=2cm, semithick, initial text=, every initial by arrow/.style={|->},
		every state/.style={ 
			inner sep=0pt}]
		\node[initial, state] (q0) {$\qhc$};
		\node[state] (q1) [right of=q0, xshift=1.3cm] {$\qhalt$};
		
		\node[state] (q2) [below of=q1, xshift=-1.2cm,yshift=0.5cm] {$q_{\mathsf{last}}$};
		\node[state] (q3) [right of=q2, xshift=0.5cm] {$\qfr$};
		
		\path
		(q0) edge [loop left, in=125, out=170, looseness=5] node {$\inc(x),\frac{4}{5},5$} (q0)
		(q0) edge [loop above, in=70, out=115, looseness=5] node {$\dec(x),\frac{3}{4},4$} (q0)
		(q0) edge [loop right, in=25, out=60, looseness=5] node {$\inc(y),\frac{6}{7},7$} (q0)
		(q0) edge [loop below, in=-95, out=-55, looseness=5] node [align=center] {\footnotesize$\allgoto$,\\$\frac{14}{15},15$} (q0)
		(q0) edge [loop left, in=-150, out=-110, looseness=5] node [xshift=-0.1cm, yshift=0.1cm] {$\dec(y),\frac{5}{6},6$} (q0)
		
		(q0) edge [below] node [yshift=0.05cm]{\footnotesize$\halt$,$\frac{15}{16},16$} (q1)
		(q1) edge [loop above, in=0, out=40, looseness=5] node [xshift = 0.2cm, yshift=0.05cm]{\footnotesize$\Sigma,0,2$} (q1)
		
		(q0) edge [right] node [xshift = 0.3cm, yshift=-0.1cm] {\footnotesize$\withouthalt,0,2$} (q2)
		(q2) edge [below] node {\footnotesize$\Sigma,2,2$} (q3)
		(q3) edge [loop above, in=0, out=40, looseness=5] node [xshift = 0.2cm, yshift=0.05cm]{\footnotesize$\Sigma,0,2$} (q3)
		;
	\end{tikzpicture}
	\caption{\label{fig:haltChecker}The Halt Checker in the NMDA $\B$.}
\end{figure}
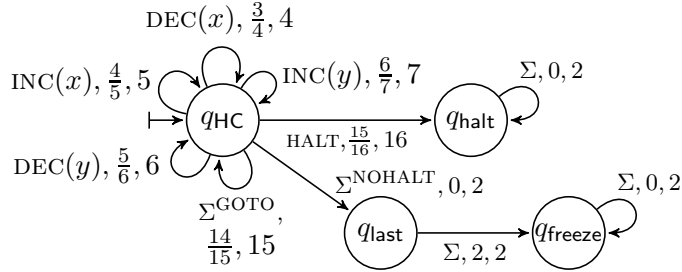

\checker{2. Negative-Counters}{}
The second gadget, depicted in \Cref{fig:NegativeCounter}, checks that the input prefix $u$ has no more $\dec(c)$ than $\inc(c)$ commands for each counter $c\in\{x,y\}$.
It is similar to $\A$, however having self loops in its initial states that favor $\dec(c)$ commands when compared to $\A$. 

\begin{figure}
	\centering
	\setlength{\belowcaptionskip}{-\baselineskip}
	\begin{tikzpicture}[->,>=stealth',shorten >=1pt,auto,node distance=2cm, semithick, initial text=, every initial by arrow/.style={|->}, every state/.style={inner sep=0pt, minimum size=0.6cm}]
		\node[initial, state] (q0) {$q_{\mathsf{Nx}}$};
		\node[state] (q1) [right of=q0, xshift=1.3cm] {$\qhalt$};
		\path
		(q0) edge [loop left, in=125, out=170, looseness=5] node {$\inc(x),\frac{9}{10},10$} (q0)
		(q0) edge [loop above, in=70, out=115, looseness=5] node {$\dec(x),\frac{1}{2},2$} (q0)
		(q0) edge [loop right, in=25, out=60, looseness=5] node {$\inc(y),\frac{6}{7},7$} (q0)
		(q0) edge [loop below, in=-80, out=-40, looseness=5] node {\footnotesize$\allgoto$,$\frac{14}{15},15$} (q0)
		(q0) edge [loop left, in=-150, out=-110, looseness=5] node [xshift=-0.1cm, yshift=0.1cm] {$\dec(y),\frac{5}{6},6$} (q0)
		
		(q0) edge [below] node [yshift=0.05cm] {$\halt,\frac{15}{16},16$} (q1)
		;
		
		\node[initial right, state] (q0) [right of=q1, xshift=1.3cm]{$q_{\mathsf{Ny}}$};
		\path
		(q0) edge [loop left, in=125, out=170, looseness=5] node [yshift=0.2cm]{$\inc(x),\frac{4}{5},5$} (q0)
		(q0) edge [loop above, in=70, out=115, looseness=5] node {$\dec(x),\frac{3}{4},4$} (q0)
		(q0) edge [loop right, in=25, out=60, looseness=5] node {$\inc(y),\frac{13}{14},14$} (q0)
		(q0) edge [loop below, in=-80, out=-40, looseness=5] node [xshift=0.7cm] {\footnotesize$\allgoto$,$\frac{14}{15},15$} (q0)
		(q0) edge [loop left, in=-150, out=-110, looseness=5] node [xshift=0.3cm, yshift=-0.3cm] {$\dec(y),\frac{2}{3},3$} (q0)
		
		(q0) edge [below] node [yshift=0.05cm] {$\halt,\frac{15}{16},16$} (q1)
		;
		
	\end{tikzpicture}
	\caption{\label{fig:NegativeCounter}The negative-counters checker, on the left for $x$ and on the right for $y$, in the NMDA $\B$.}
\end{figure}
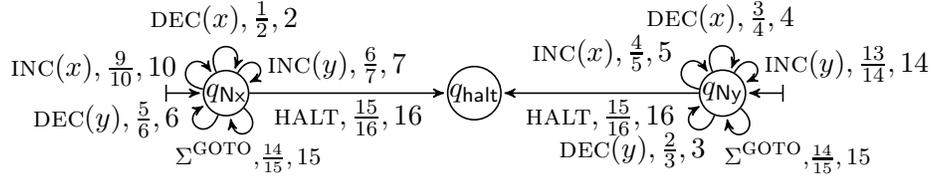

\checker{3. Positive-Counters}{}
The third gadget, depicted in \Cref{fig:balancedCounters}, checks that for every $c\in\{x,y\}$, the input prefix $u$ has no more $\inc(c)$ than $\dec(c)$ commands.
It is similar to $\A$, while having self loops in its initial state according to the dual functions rather than the primal ones. 

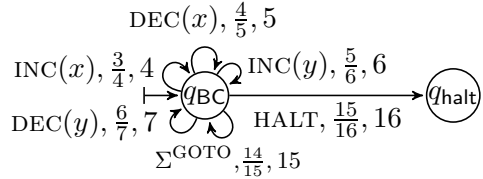
\begin{figure}
	\vspace*{-\baselineskip}
	\centering
	\setlength{\belowcaptionskip}{-\baselineskip}
	\begin{tikzpicture}[->,>=stealth',shorten >=1pt,auto,node distance=2cm, semithick, initial text=, every initial by arrow/.style={|->}, every state/.style={inner sep=0pt, minimum size=0.6cm}]
		\node[initial, state] (q0) {$q_{\mathsf{BC}}$};
		\node[state] (q1) [right of=q0, xshift=1.3cm] {$\qhalt$};
		\path
		(q0) edge [loop left, in=125, out=170, looseness=5] node {$\inc(x),\frac{3}{4},4$} (q0)
		(q0) edge [loop above, in=70, out=115, looseness=5] node {$\dec(x),\frac{4}{5},5$} (q0)
		(q0) edge [loop right, in=25, out=60, looseness=5] node {$\inc(y),\frac{5}{6},6$} (q0)
		(q0) edge [loop below, in=-80, out=-40, looseness=5] node {\footnotesize$\allgoto$,$\frac{14}{15},15$} (q0)
		(q0) edge [loop left, in=-150, out=-110, looseness=5] node [xshift=-0.1cm, yshift=0.1cm] {$\dec(y),\frac{6}{7},7$} (q0)
		
		(q0) edge [below] node [yshift=0.05cm] {$\halt,\frac{15}{16},16$} (q1)
		;
	\end{tikzpicture}
	\caption{\label{fig:balancedCounters}The Positive-Counters Checker in the NMDA $\B$.}
\end{figure}

\checker{4. Command}{}
The next gadget checks for local violations of successive commands. That is, it makes sure that the letter $w_i$ represents a command that can follow the command represented by $w_{i-1}$ in $\M$, ignoring the counter values. 
For example, if the command in location $l_2$ is $\inc(x)$, then from state $q_2$, which is associated with $l_2$, we move with the letter $\inc(x)$ to $q_3$, which is associated with $l_3$.
The test is local, as this gadget does not check for violations involving illegal jumps due to the values of the counters. 
An example of the command checker for the counter machine in \Cref{fig:machineExample} is given in \Cref{fig:commandCheckerExample}.

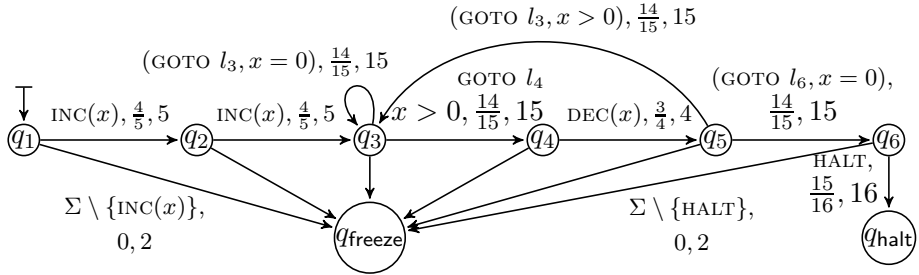
\begin{figure}
	\vspace*{-\baselineskip}
	\centering
	\setlength{\belowcaptionskip}{-\baselineskip}
	\begin{tikzpicture}[->,>=stealth',shorten >=1pt,auto,node distance=2.3cm, semithick, initial text=, every initial by arrow/.style={|->}, every state/.style={ 
			inner sep=0pt, minimum size=0.4cm}]
		\node[initial above, state] (q1) {$q_1$};
		\node[state] (q2) [right of=q1] {$q_2$};
		\node[state] (q3) [right of=q2] {$q_3$};
		\node[state] (q4) [right of=q3] {$q_4$};
		\node[state] (q5) [right of=q4] {$q_5$};
		\node[state] (q6) [right of=q5] {$q_6$};
		\node[state] (hlt) [below of=q6, yshift=1cm] {$\qhalt$};
		\node[state] (fr) [below of=q3, yshift=1cm] {$\qfr$};
		\path
		(q1) edge [above] node [align=center]{\footnotesize$\inc(x),\frac{4}{5},5$} (q2)
		(q2) edge [above] node [align=center, xshift=-0.1cm]{\footnotesize$\inc(x),\frac{4}{5},5$} (q3)
		(q3) edge [loop above,out=80,in=130, looseness=11] node [align=center, xshift=-1.2cm]{\footnotesize{$(\goto l_3,x=0),\frac{14}{15},15$}} (q3)
		(q3) edge [above] node [align=right, xshift=0.15cm]{\footnotesize{$\mbox{\sc goto } l_4$}\\$x>0,\frac{14}{15},15$} (q4)
		(q4) edge [above] node [align=center]{\footnotesize$\dec(x),\frac{3}{4},4$} (q5)
		(q5) edge [above] node [align=center]{\footnotesize$(\goto l_6,x=0),$\\$\frac{14}{15},15$} (q6)
		(q5) edge [above,out=120,in=60] node [align=center, xshift=0.4cm]{\footnotesize$(\goto l_3,x>0),\frac{14}{15},15$} (q3)
		(q6) edge [left] node [align=center]{\footnotesize$\halt,$\\$\frac{15}{16},16$} (hlt)
		(q1) edge [below] node [align=center,xshift=-0.7cm]{\footnotesize$\Sigma\setminus\{\inc(x)\},$\\\footnotesize$0,2$} (fr)
		(q2) edge [below] node [align=center,xshift=-0.7cm]{} (fr)
		(q3) edge [below] node [align=center,xshift=-0.7cm]{} (fr)
		(q4) edge [below] node [align=center,xshift=-0.7cm]{} (fr)
		(q5) edge [below] node [align=center,xshift=-0.7cm]{} (fr)
		(q6) edge [below] node [align=center,xshift=0.7cm]{\footnotesize$\Sigma\setminus\{\halt\},$\\\footnotesize$0,2$} (fr)
		;
	\end{tikzpicture}
	\caption{\label{fig:commandCheckerExample}The command checker that corresponds to the counter machine in \Cref{fig:machineExample}.}
\end{figure}

The command checker, which is a DMDA, consists of states $q_1,\ldots,q_n$ that correspond to the commands $l_1,\ldots,l_n$, and the states $\qhalt$ and $\qfr$. 
For two locations $j$ and $k$, there is a transition from $q_j$ to $q_k$ on the letter $\sigma$ iff $l_k$ can {\em locally follow\/} $l_j$ in a run of $\M$ that has $\sigma$ in the corresponding location of the command trace. 
That is, either $l_j$ is a $\goto l_k$ command (meaning $l_j=\sigma=\goto l_k$), $k$ is the next location after $j$ and $l_j$ is an $\inc$ or a $\dec$ command (meaning $k=j+1$ and $l_j=\sigma\in\incdec$), $l_j$ is an $\jz{c}{l_k}{l_{k'}}$ command with $\sigma=(\goto l_k,c=0)$, or $l_j$ is an $\jz{c}{l_s}{l_k}$ command with $\sigma=(\goto l_k,c>0)$.
The weights and discount factors of the $\withouthalt$ transitions mentioned above are according to the primal functions $\gamma_p$ and $\rho_p$ respectively.
For every location $j$ such that $l_j=\halt$, there is a transition from $q_j$ to $\qhalt$ labeled by the letter $\halt$ with a weight of $\frac{15}{16}$ and a discount factor of $16$.
Every other transition that was not specified above leads to $\qfr$ with weight $0$ and some discount factor.

\checker{5,6. Zero-Jump}{s}
The next gadgets, depicted in \Cref{fig:zeroJumpChecker}, check for violations in conditional jumps. In this case, we use a different checker instance for each counter $c\in\{x,y\}$, ensuring that for every $\jz{c}{l_k}{l_{k'}}$ command,  if the jump $\goto l_k$ is taken, then the value of $c$ is indeed $0$.

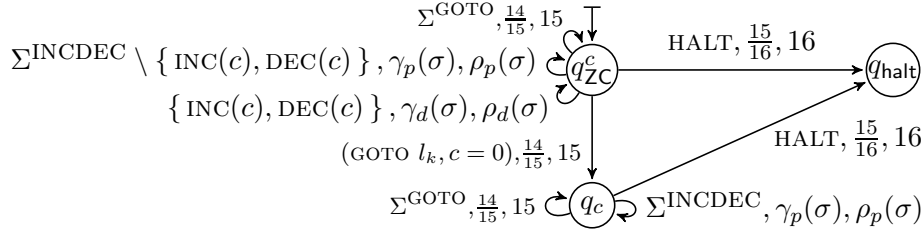
\begin{figure}
	\vspace*{-\baselineskip}
	\centering
	\setlength{\belowcaptionskip}{-\baselineskip}
	\begin{tikzpicture}[->,>=stealth',shorten >=1pt,auto,node distance=4cm, semithick, initial text=, every initial by arrow/.style={|->}, every state/.style={inner sep=0pt, minimum size=0.6cm}]
		\node[initial above, state] (q0) {$\qzc$};
		\node[state] (q1) [below of=q0, yshift=2.2cm] {$q_c$};
		\node[state] (hlt) [right of=q0] {$\qhalt$};
		\path
		(q0) edge [loop left, in=110, out=140, looseness=5] node [align=center, yshift=0.2cm, xshift=0.1cm] {\footnotesize$\allgoto$,$\frac{14}{15},15$} (q0)
		(q0) edge [loop left, in=160, out=-170, looseness=5] node [yshift=0.07cm] {$\incdec\setminus\set{\inc(c), \dec(c)},\gamma_p(\sigma),\rho_p(\sigma)$} (q0)
		(q0) edge [loop left, in=-150, out=-120, looseness=5] node [yshift=-0.1cm] {$\set{\inc(c), \dec(c)},\gamma_d(\sigma),\rho_d(\sigma)$} (q0)
		(q0) edge [left] node [align=center,yshift=-0.2cm] {\footnotesize$(\goto l_k,c=0),$\footnotesize$\frac{14}{15},15$} (q1)
		
		(q1) edge [loop right, in=-20, out=10, looseness=6] node {$\incdec,\gamma_p(\sigma),\rho_p(\sigma)$} (q1)
		
		(q1) edge [loop left, in=165, out=-160, looseness=6] node [align=center] {\footnotesize$\allgoto$,$\frac{14}{15},15$} (q1)
		
		(q0) edge [above] node [align=center] {$\halt,\frac{15}{16},16$} (hlt)
		
		(q1) edge [right] node [align=center, xshift=0.3cm] {$\halt,\frac{15}{16},16$} (hlt)
		;
	\end{tikzpicture}
	\caption{\label{fig:zeroJumpChecker}The Zero-Jump Checker (for a counter $c\in\set{x,y}$) in the NMDA $\B$.}
\end{figure}

Intuitively, $\qzc$ profits from words that have more $\inc(c)$ than $\dec(c)$ letters, while $q_c$ continues like $\A$.
If the move to $q_c$ occurred after a balanced number of $\inc(c)$ and $\dec(c)$, as it should be in a real command trace, neither the prefix word before the move to $q_c$, nor the suffix word after it result in a profit. 
Otherwise, provided that the counter is $0$ at the end of the run (as guaranteed by the negative- and positive-counters checkers), both prefix and suffix words get profits, resulting in a smaller value for the run. 

\checker{7,8. Positive-Jump}{s}
These gadgets, depicted in \Cref{fig:positiveJumpChecker}, are dual to the zero-jump checkers, checking for the dual violations in conditional jumps. 
Similarly to the zero-jump checkers, we have a different instance for each counter $c\in\{x,y\}$, ensuring that for every $\jz{c}{l_k}{l_{k'}}$ command, if the jump $\goto l_{k'}$ is taken, then the value of $c$ is indeed greater than $0$.

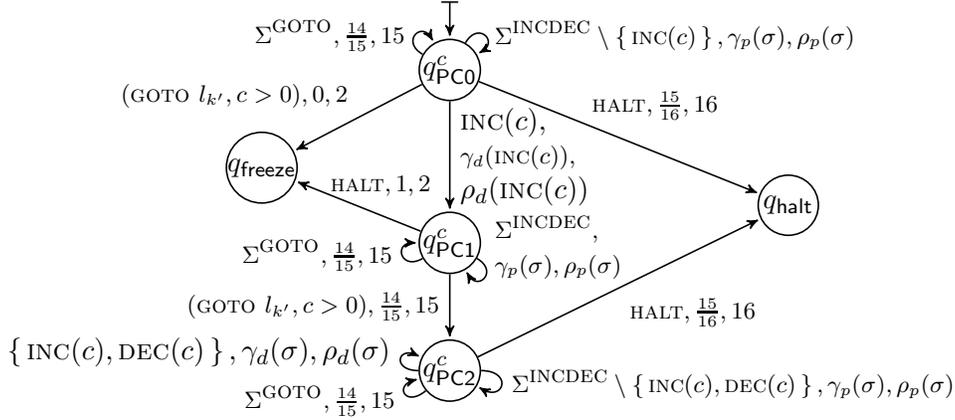
\begin{figure}
	\centering
	\begin{tikzpicture}[->,>=stealth',shorten >=1pt,auto,node distance=2.5cm, semithick, initial text=, every initial by arrow/.style={|->},
		every state/.style={ 
			inner sep=0pt, minimum size=0.8cm}]
		\node[initial above, state] (q0) {$\qpc{0}$};
		\node[state] (q1) [below of=q0, yshift=0.2cm] {$\qpc{1}$};
		\node[state] (q2) [below of=q1, yshift=0.8cm] {$\qpc{2}$};
		\node[state] (qfr) [left of=q1,yshift=1cm] {$\qfr$};
		\node[state] (hlt) [right of=q1,xshift=2cm,yshift=0.5cm] {$\qhalt$};
		\path
		(q0) 
		edge 
		[loop left, out=150, in=120, looseness=4] 
		node 
		[align=center]
		{\footnotesize$\allgoto,\frac{14}{15},15$} 
		(q0)
		
		(q0) 
		edge 
		[loop right, out=60, in=30, looseness=4] 
		node 
		[align=center] {\footnotesize$\incdec\setminus\set{\inc(c)},\gamma_p(\sigma),\rho_p(\sigma)$}
		(q0)
		
		(q0) 
		edge 
		[right] 
		node 
		[align=left] 
		{$\inc(c),$\\\footnotesize$\gamma_d(\inc(c)),$\\$\rho_d(\inc(c))$}
		(q1)
		
		(q0) 
		edge 
		[above right] 
		node 
		[align=center, xshift=-0.5cm, yshift=0.1cm] 
		{\footnotesize$\halt,\frac{15}{16},16$}
		(hlt)
		
		(q0) 
		edge 
		[above left] 
		node 
		[align=center] 
		{\footnotesize$(\goto l_{k'},c>0),0,2$} 
		(qfr)
		
		(q1) 
		edge 
		[loop right, out=-30, in=-60, looseness=4] 
		node 
		[align=left,yshift=0.4cm] 
		{\footnotesize$\incdec,$\\\footnotesize$\gamma_p(\sigma),\rho_p(\sigma)$} 
		(q1)
		
		(q1) 
		edge 
		[loop left, out=-150, in=-180, looseness=4] 
		node 
		[align=center] 
		{\footnotesize$\allgoto,\frac{14}{15},15$} 
		(q1)
		
		(q1) 
		edge 
		[left] 
		node 
		[align=center] 
		{\footnotesize$(\goto l_{k'},c>0),\frac{14}{15},15$} 
		(q2)
		
		(q2) 
		edge 
		[loop right, out=0, in=-30, looseness=5] 
		node 
		[align=center] 
		{\footnotesize$\incdec\setminus\set{\inc(c),\dec(c)},\gamma_p(\sigma),\rho_p(\sigma)$} 
		(q2)
		
		(q2) 
		edge 
		[loop left, out=-140, in=-165, looseness=5] 
		node 
		[align=center, yshift=-0.1cm] 
		{\footnotesize$\allgoto,\frac{14}{15},15$} 
		(q2)
		
		(q2) 
		edge 
		[loop left, out=180, in=155, looseness=5] 
		node 
		[align=center, yshift=0.1cm]
		{$\set{\inc(c),\dec(c)},\gamma_d(\sigma),\rho_d(\sigma)$} 
		(q2)
		
		(q2) 
		edge 
		[below right] 
		node 
		[align=center] 
		{\footnotesize$\halt,\frac{15}{16},16$} 
		(hlt)
		
		(q1) 
		edge 
		[above] 
		node 
		[align=center,xshift=0.3cm] 
		{\footnotesize$\halt,1,2$} 
		(qfr)
		;
	\end{tikzpicture}
	\caption{\label{fig:positiveJumpChecker}The Positive-Jump Checker (for a counter $c$) in the NMDA $\B$.}
\end{figure}

Intuitively, if the counter is $0$ on a $(\goto l_{k'},c>0)$ command when there was no $\inc(c)$ command yet, the gadget benefits by moving from $\qpc{0}$ to $\qfr$. If there was an $\inc(c)$ command, it benefits by having the dual functions on the move from $\qpc{0}$ to $\qpc{1}$ over $\inc(c)$ and the primal functions on one additional self loop of $\qpc{1}$ over $\dec(c)$.

\begin{lem}
	\label{cl:undecidabilityContainment}
	Given a two-counter machine $\M$, we can compute an integral DMDA $\A$ and an integral NMDA $\B$ on finite words, such that $\M$ $0$-halts iff there exists a word $w\in\Sigma^+$ such that $\B(w)\geq \A(w)$ iff there exists a word $w\in\Sigma^+$ such that $\B(w) > \A(w)$.
\end{lem}
\begin{proof}
	Given a two-counter machine $\M$, consider the DMDA $\A$ and the NMDA $\B$ constructed in \Cref{sec:TheReduction}, and an input word $w$. Let $u=\pref_{\halt}(w)$.
	
	We prove the claim by showing that I) if $u$ correctly describes a 0-halting run of $\M$ then $\B(w)>\A(w)$, and II) if $u$ does not fit the actual run of $\M$, or if it does fit it, but the run does not 0-halt, then the violation is detected by $\B$, in the sense that $\B(w)<\A(w)$.
	
	\vspace{5pt}\noindent{\bf I.} 	
	We start with the case that $u$ correctly describes a 0-halting run of $\M$, and show that $\B(w)>\A(w)$.
	
	Observe that in all of $\B$'s checkers, the transition over the \halt command to the $\qhalt$ state has a weight higher than the weight of the corresponding transition in $\A$. Thus, if a checker behaves like $\A$ over $u$, namely uses the primal functions, it generates a value higher than that of $\A$.
	
	We show below that each of the checkers generates a value higher than the value of $\A$ on $u$ (which is also the value of $\A$ on $w$), also if it nondeterministically ``guesses a violation'', behaving differently than $\A$.
	
	\lchecker{1. Halt}{}
	Since $u$ does have the \halt command, the run of the halt checker on $u$, if guessing a violation, will end in the pair of transitions from $\qhc$ to $q_{\mathsf{last}}$ to $\qfr$ with discount factor $2$ and weights $0$ and $2$, respectively. 
	
	Let $D$ be the accumulated discount factor in the gadget up to these pair of transitions. 
	According to \Cref{lem:undecidabilityContainment}, the accumulated weight at this point is $1-\frac{1}{D}$, hence the value of the run will be $1-\frac{1}{D} + \frac{1}{D}\cdot 0 + \frac{1}{2D}\cdot 2 = 1$, which is, according to \Cref{lem:undecidabilityContainment}, larger than the value of $\A$ on any word.
	
	\lchecker{2,3. Negative- and Positive-Counters}{s}
	Since $u$ has the same number of $\inc(c)$ and $\dec(c)$ letters, by \Cref{eqn:primalDual,lem:undecidabilityContainment}, these gadgets and $\A$ will have the same value on the prefix of $u$ until the last transition, on which the gadgets will have a higher weight.
	
	\lchecker{4. Command}{}
	As this gadget is deterministic, it cannot ``guess a violation'', and its value on $u$ is larger than $\A(u)$ due to the weight on the \halt command.
	
	\lchecker{5,6. Zero-Jump}{s}
	Consider a counter $c\in\set{x,y}$ and a run $r$ of the gadget on $u$.
	If $r$ did not move to $q_c$, we have $\B(r)>\A(w)$, similarly to the analysis in the negative- and positive-counters checkers.
	Otherwise, denote the transition that $r$ used to move to $q_c$ as $t$.
	Observe that since $u$ correlates to the actual run of $\M$, we have that $t$ was indeed taken when $c=0$.
	In this case the value of the run will not be affected, since before $t$ we have the same number of $\inc(c)$ and $\dec(c)$ letters, and after $t$ we also have the same number of $\inc(c)$ and $\dec(c)$ letters. Hence, due to the last transition over the \halt command, we have $\B(r)>\A(u)$.
	
	\lchecker{7,8. Positive-Jump}{s}
	Consider a counter $c\in\set{x,y}$ and a run $r$ of the gadget on $u$.
	If $r$ never reaches $\qpc{1}$, it has the same sequence of weights and discount factors as $\A$, except for the higher-valued \halt transition.
	If $r$ reaches $\qpc{1}$ but never reaches $\qpc{2}$, since $u$ ends with a $\halt$ letter, we have that $r$ ends with a transition to $\qfr$ that has a weight of $1$, hence $\B(r)=1>\A(w)$.

	If $r$ reaches $\qpc{2}$, let $u=y\con\inc(c)\con z\con v$ where $y$ has no $\inc(c)$ letters, $t=r[|y|+1+|z|]$ is the first transition in $r$ targeted at $\qpc{2}$, and $\alpha_c\geq 1$ is the value of the counter $c$ when $t$ is taken.
	We have that $1+\numof{\inc(c)}{z}=\numof{\dec(c)}{z}+\alpha_c$. 
	Since $u$ is balanced, we also have that $\numof{\dec(c)}{v}=\numof{\inc(c)}{v}+\alpha_c$.
	For the first $\inc(c)$ letter, $r$ gets a discount factor of $\rho_d(\inc(c))=\rho_p(\dec(c))$.
	All the following $\inc(c)$ and $\dec(c)$ letters contribute discount factors according to $\rho_p$ in $z$ and according to $\rho_d$ in $v$.
	Hence, $r$ gets the discount factor $\rho_p(\dec(c))$ a total of 
	\begin{align*}
		1+\numof{\dec(c)}{z}+\numof{\inc(c)}{v}&=
		1+1+\numof{\inc(c)}{z}-\alpha_c+\numof{\inc(c)}{v}\\
		&=
		\numof{\inc(c)}{u}+1-\alpha_c\\&\leq \numof{\inc(c)}{u}=\numof{\dec(c)}{u}
	\end{align*}
	times, and the discount factor $\rho_p(\inc(c))$ a total of 
	\begin{align*}
		\numof{\inc(c)}{z}+\numof{\dec(c)}{v}&=
		\numof{\inc(c)}{z}+\numof{\inc(c)}{v}+\alpha_c\\&=
		\numof{\inc(c)}{u}-1+\alpha_c\geq \numof{\inc(c)}{u}
	\end{align*}
	times.
	
	Therefore, the value of $r$ is at least as big as the value of $\A$ on the prefix of $u$ until the \halt transition, and due to the higher weight of $r$ on the latter, we have $\B(r) > \A(u)$.

	\vspace{5pt}\noindent{\bf II.} 	
	We continue with the case that $u$ does not correctly describe a 0-halting run of $\M$, and show that $\B(w) <\A(w)$.	
	Observe that the incorrectness must fall into one of the following cases, each of which results in a lower value of one of $\B$'s gadgets on $u$, compared to the value of $\A$ on $u$:
	\begin{itemize}
		\item {\it The word $u$ has no \halt command.} In this case the minimal-valued run of the halt checker on $u$ will be the same as of $\A$ until the last transition, on which the halt checker will have a $0$ weight, compared to a strictly positive weight in $\A$.
		
		\item {\it The word $u$ does not describe a run that ends up with value $0$ in both counters.}
		Then there are the following sub-cases:
		\begin{itemize}
			\item {\it The word $u$ has more $\dec(c)$ than $\inc(c)$ letters for some counter $c\in\{x,y\}$}.
			For $c=x$, in the negative-counters checker, more discount factors were changed from $4$ to $2$  than those changed from $5$ to $10$, compared to their values in $\A$, implying that the total value of the gadget until the last letter will be lower than of $\A$ on it.
			For $c=y$, we have a similar analysis with respect to the discount factors $6;3$, and $7;14$.
			\item {\it The word $u$ has more $\inc(c)$ than $\dec(c)$ letters for some counter $c\in\{x,y\}$}.
			By \Cref{eqn:primalDual,lem:undecidabilityContainment}, the value of the positive-counters checker until the last transition will be lower than of $\A$ until the last transition.
		\end{itemize}
		
		Observe, though, that the weight of the gadgets on the \halt transition ($16$) is still higher than that of $\A$ on it ($15$). 
		Nevertheless, since a ``violation detection'' results in replacing at least one discount factor from $4$ to $2$, from $6$ to $3$, from $5$ to $4$, or from $7$ to $6$ (and replacing the corresponding weights, for preserving the $\frac{\rho-1}{\rho}$ ratio), and the ratio difference between $16$ and $15$ is less significant than between the other pairs of weights, we have that the gadget's value and therefore $\B$'s value on $u$ is smaller than $\A(u)$.
		Indeed, by \Cref{lem:undecidabilityContainment} $\A(u)=1-\frac{1}{D_\A}$, where $D_\A$ is the multiplication of the discount factors along $\A$'s run, and $\B(u)\leq 1-(\frac{1}{D_\A}\cdot \frac{7}{6}\cdot\frac{15}{16}) < 1-\frac{1}{D_\A} = \A(u)$.
		
		\item {\it The word $u$ does not correctly describe the run of $\M$}. Then there are the following sub-cases:
		\begin{itemize}
			\item {\it The incorrect description does not relate to conditional jumps}. 
			Then the command-checker has the same weights and discount factors as $\A$ on the prefix of $u$ until the incorrect description, after which it has $0$ weights, compared to strictly positive weights in $\A$.
			\item {\it The incorrect description relates to conditional jumps}. Then there are the following sub-sub-cases:
			\begin{itemize}
				\item {\it A counter $c>0$ at a position $i$ of $\M$'s run, while $u[i]=\goto l_k,c=0$}. Let $v=u[0..i{-}1]$ and $u=v\con v'$, and consider the run $r$ of the zero-jump checker on $u$ that moves to $q_c$ after $v$.
				Then $\numof{\inc(c)}{v}>\numof{\dec(c)}{v}$ and $\numof{\inc(c)}{v'}<\numof{\dec(c)}{v'}$. (We may assume that the total number of $\inc(c)$ and $\dec(c)$ letters is the same, as otherwise one of the previous checkers detects it.)
				
				All the $\inc(c)$ and $\dec(c)$ transitions in $r[0..i{-}1]$ have weights and discount factors according to the dual functions, and those transitions in $r[i..|w|{-}1]$ have weights and discount factors according to the primal functions.
				Therefore, compared to $\A$, more weights changed from $\gamma_p(\inc(c))$ to $\gamma_d(\inc(c))=\gamma_p(\dec(c))$ than weights changed from $\gamma_p(\dec(c))$ to $\gamma_d(\dec(c))=\gamma_p(\inc(c))$, resulting in a lower  total value of $r$ than of $\A$ on $u$. (As shown for the negative- and positive-counters checkers, the higher weight of the $\halt$ transition is less significant than the lower values above.)
				
				\item {\it A counter $c=0$ at a position $i$ of $\M$'s run, while $u[i]=\goto l_k,c>0$}.
				Let $r$ be a minimal-valued run of the positive-jump checker on $u$.
				
				If  there are no $\inc(c)$ letters in $u$ before position $i$, $r$ will have the same weights and discount factors as $\A$ until the $i$'s letter, on which it will move from $\qpc{1}$ to $\qfr$, continuing with $0$-weight transitions, compared to strictly positive ones in $\A$.
				
				Otherwise, we have that the first $\inc(c)$ letter of $u$ takes $r$ from $\qpc{0}$ to $\qpc{1}$ with a discount factor of $\rho_d(\inc(c))$. 
				Then in $\qpc{1}$ we have more $\dec(c)$ transitions than $\inc(c)$ transitions, and in $\qpc{2}$ we have the same number of $\dec(c)$ and $\inc(c)$ transitions. (We may assume that $u$ passed the previous checkers, and thus has the same total number of $\inc(c)$ and $\dec(c)$ letters.)
				Hence, we get two more discount factors of $\rho_d(\inc(c))$ than $\rho_p(\inc(c))$, resulting in a value smaller than $\A(u)$. (As in the previous cases, the higher value of the \halt transition is less significant.)
				\qedhere
			\end{itemize}
			\qedhere
		\end{itemize}
		\qedhere
	\end{itemize}
\end{proof}

\subsubsection{Undecidability of arbitrary integral NMDAs containment}\label{sec:Undecidability}
For finite words, the undecidability result directly follows from \Cref{cl:undecidabilityContainment} and the undecidability of the 0-halting problem of counter machines \cite{Min67}.
\begin{thm}\label{cl:ContainmentFiniteWordsUndecidable}
	Strict and non-strict containment of (integral) NMDAs on finite words are undecidable. More precisely, the problems of deciding for given integral NMDA $\N$ and integral DMDA $\D$ whether $\N(w) \leq \D(w)$ for all finite words $w$ and whether $\N(w) < \D(w)$ for all finite words $w$.
\end{thm}
\noindent 
For infinite words, undecidability of non-strict containment also follows from the reduction given in \Cref{sec:TheReduction}, as the reduction considers prefixes of the word until the first \halt command. 
We leave open the question of whether strict containment is also undecidable for infinite words. The issue with the latter is that a \halt command might never appear in an infinite word $w$ that incorrectly describes a halting run of the two-counter machine, in which case both automata $\A$ and $\B$ of the reduction will have the same value on $w$. 
On words $w$ that have a \halt command but do not correctly describe a halting run of the two-counter machine we have $\B(w)<\A(w)$, and on a word $w$ that does correctly describe a halting run we have $\B(w)>\A(w)$. Hence, the reduction only relates to whether $\B(w)\leq\A(w)$ for all words $w$, but not to whether $\B(w) <\A(w)$ for all words $w$.

\begin{thm}\label{cl:NonStrictContainmentInfiniteWordsUndecidable}
	Non-strict containment of (integral) NMDAs on infinite words is undecidable. More precisely, the problem of deciding for given integral NMDA $\N$ and integral DMDA $\D$ whether $\N(w) \leq \D(w)$ for all infinite words $w$.
\end{thm}
\begin{proof}
	The automata $\A$ and $\B$ in the reduction given in \Cref{sec:TheReduction} can operate as is on infinite words, ignoring the Halt-Checker gadget of $\B$ which is only relevant to finite words.
	
	Since the values of both $\A$ and $\B$ on an input word $w$ only relate to the prefix $u=\pref_{\halt(w)}$ of $w$ until the first \halt command, we still have that $\B(w)>\A(w)$ if $u$ correctly describes a halting run of the two-counter machine $\M$ and that $\B(w)<\A(w)$ if $u$ is finite and does not correctly describe a halting run of $\M$.
	
	Yet, for infinite words there is also the possibility that the word $w$ does not contain the \halt command. In this case, the value of both $\A$ and the command checker of $\B$ will converge to $1$, getting $\A(w)=\B(w)$.
	
	Hence, if $\M$ 0-halts, there is a word $w$, such that $\B(w)>\A(w)$ and otherwise, for all words $w$, we have $\B(w)\leq \A(w)$. 
\end{proof}

Observe that for NMDAs, equivalence and non-strict containment are interreducible.

\begin{thm}\label{cl:EquivalenceFiniteWordsUndecidable}
	Equivalence of (integral) NMDAs on finite as well as infinite words is undecidable. That is, the problem of deciding for given integral NMDAs $\A$ and $\B$ on finite or infinite words whether $\A(w) = \B(w)$ for all words $w$.
\end{thm}
\begin{proof}
	Assume toward contradiction the existence of a procedure for equivalence check of $\A$ and $\B$. 
	We can use the nondeterminism to obtain an automaton $\C=\A \cup \B$, having $\C(w)\leq \A(w)$ for all words $w$. We can then check whether $\C$ is equivalent to $\A$, which holds if and only if $\A(w) \leq \B(w)$ for all words $w$. Indeed, if $\A(w) \leq \B(w)$ then $\A(w) \leq \min(\A(w), \B(w)) = \C(w)$, while if there exists a word $w$, such that $\B(w)<\A(w)$, we have $\C(w) = \min(\A(w), \B(w))  <\A(w)$, implying that $\C$ and $\A$ are not equivalent. Thus, such a procedure contradicts the undecidability of non-strict containment, shown in \Cref{cl:ContainmentFiniteWordsUndecidable,cl:NonStrictContainmentInfiniteWordsUndecidable}.
\end{proof}

\begin{rem}\label{rem:PCP}
	One can provide a much simpler undecidability result for a problem we do not define in \Cref{sec:Definitions} and do not formally consider in the paper: Given two integral NMDAs, or even integral DMDAs, $\A$ and $\B$, does there exist a finite word $w$, such that $\A(w)=\B(w)$? 
	
	We provide below a sketch of the proof, which goes by reduction from the Post Correspondence Problem (PCP). 
	Recall that in a PCP there are two finite lists, $L_1=\alpha_1, \ldots, \alpha_k$ and $L_2=\beta_1, \ldots, \beta_k$, of finite words over some alphabet $X$, and a solution to the problem is a squence $(i_j)_{j\in[1..N]}$, for some $N\in\Nat\setminus\{0\}$, where $i_j\in[1..k]$ for all $j$, such that $\alpha_{i_1} \alpha_{i_2} \cdots \alpha_{i_N} = \beta_{i_1} \beta_{i_2} \cdots \beta_{i_N}$.
	
	In the reduction, the alphabet of the DMDAs is $\Sigma = \{1,2,\ldots,k\}$, namely the indices of the PCP, while each letter in the PCP alphabet $X$ is assigned a unique integer between $1$ and $|X|$. Then, each word $u$ in $L_1$ and $L_2$ is assigned a value $\Val(u)$ between $0$ and $1$ according to the value of $0.u$ in base $X{+}1$.	
	For example, consider a PCP with alphabet $X$ of size 9, a word $u=aafagb$ in its lists, and the letter assignment $a=1, b=2, f=6, g=7$. Then the value assigned to $u$ is $\Val(u)=0.116172$ in base $10$. 
	
	Now, each of the DMDAs $\A$ and $\B$ consists of a single state; in $\A$ the transition over the letter $i$ has value $\Val(\alpha_i)$ and discount factor $|X{+}1|^{|\alpha_i|}$ (where $|\alpha_i|$ is the length of $\alpha_i$), and analogously in $\B$ the transition over the letter $i$ has value $\Val(\beta_i)$ and discount factor $|X{+}1|^{|\beta_i|}$. Observe that $\A$ and $\B$ have the same value on a finite input word $w$ if and only if $w$ is a sequence of indices that is a solution to the corresponding PCP.	
\end{rem}

\section{Tidy NMDAs} \label{sec:TidyNMDAs}
We present the family of ``tidy NMDAs'' and show that it is as expressive as integral DMDAs.
Intuitively, an integral NMDA is tidy if the choice of discount factors depends on the word prefix read so far. 
We further show that for every choice function $\theta$, the class of all $\theta$-NMDAs is closed under determinization and algebraic operations, and enjoys decidable algorithms for its decision problems.

The family of tidy NMDAs contains various natural subfamilies, each strictly extending the expressive power of integral NDAs. Among which are integral NMDAs whose discount factors are chosen per letter (action) or per the elapsed time. We elaborate on these subfamilies at the end of the section. 

\begin{defi}\label{def:ThetaNMDA}
	An integral NMDA $\A$ over an alphabet $\Sigma$ and with discount-factor function $\rho$ is \emph{tidy} if there exists a function
	$\theta: \Sigma^+ \to \Nat \setminus \{0,1\}$, such that for every finite word $u=\sigma_1\ldots\sigma_n\in\Sigma^+$, and every run $q_0,\sigma_1,\cdots,q_n$ of $\A$ on $u$, we have 
	$\rho(q_{n-1},\sigma_n,q_n)=\theta(u)$.
	
	In this case we say that $\A$ is a $\theta$-NMDA.
\end{defi}

For example, the NMDAs in \Cref{fig:detExample,fig:LetterOrrientedExample,fig:TimeOrientedExample} are tidy, whereas the ones in \Cref{fig:NMDAExample,fig:NonDetNMDAInfinityWords} are not. (In the NMDA of \Cref{fig:NMDAExample}, for the word ``a'' (of length one) and the runs $r_1=(q_0,a,q_1)$ and $r_2=(q_0,a,q_0)$ on it, we have $\rho(q_0,a,q_1)=2 \neq 3 = \rho(q_0,a,q_0)$, and likewise in the NMDA of \Cref{fig:NonDetNMDAInfinityWords}, relating to the runs $r_1=(q_0,a,q_0)$ and  $r_2=(q_1,a,q_1)$.)

Notice that while the notion of tidiness is declarative, checking whether a given NMDA is tidy can be done in quadratic time, as it reduces to checking a reachability problem on a Cartesian product of the NMDA with itself (see \Cref{sec:Tidiness}).

\begin{defi}\label{def:choiceFunction}
	For an alphabet $\Sigma$, a function $\theta: \Sigma^+ \to \Nat \setminus \{0,1\}$ is a \emph{choice function} if there exists an integral NMDA that is a $\theta$-NMDA.
\end{defi}

For example, the function $\theta$ defined by ``$\theta(u)=2$ if $|u|$ is odd and $3$ if $|u|$ is even'' is a choice function, since the integral NMDA $\A$ from \Cref{fig:TimeOrientedExample} is a $\theta$-NMDA. On the other hand, the function $f:  \Sigma^+ \to \Nat \setminus \{0,1\}$ defined by $\theta(u) = |u|$ is not a choice function, as its image is unbounded, reflecting infinitely many discount factors in an NMDA, which is not possible, as an NMDA has finitely many transitions. The function ``$\theta(u)=2$ if $u$ encodes (via some standard encoding) a halting Turing machine, and $3$ otherwise'' is also not a choice function, even though its image is bounded, since an NMDA with such a choice function would have solved the undecidable halting problem of Turing machines.
In fact, even though a general function $\theta: \Sigma^+ \to \Nat \setminus \{0,1\}$ might require an infinite representation, every choice function has a finite representation as a transducer (see \Cref{sec:RepresentingChoiceFunctions}). 

For choice functions $\theta_1$ and $\theta_2$, the classes of $\theta_1$-NMDAs and of $\theta_2$-NMDAs are \emph{equivalent} if they express the same functions, namely if for every $\theta_1$-NMDA $\A$, there exists a $\theta_2$-NMDA $\B$ equivalent to $\A$ and vice versa.

For every tidy NMDA $\A$ and finite word $u$, all the runs of $\A$ on $u$ entail the same accumulated discount factor. We thus use the notation $\rho(u)$ to denote $\rho(r)$, where $r$ is any run of $\A$ on $u$.

\subsection{Determinizability}\label{sec:Determinizability}

We show that for every tidy NMDA $\N$, we can construct an equivalent integral DMDA $\D$. The defining feature of a tidy NMDA is that each run of a word is discounted identically. We use this property to construct $\D$ such that $\D(w)=\N(w)$ for all finite words $w$. We then use \Cref{lemma:finiteToInfinite} to lift this equivalence to all infinite words.

\begin{figure}
	\centering
	
	\begin{tikzpicture}[->,>=stealth',shorten >=1pt,auto,node distance=2.5cm, semithick, initial text=, every initial by arrow/.style={|->}]
		\node[initial, state, inner sep=0.1cm,minimum size=0.4cm] (q1) {$q_1$};
		\node[state, inner sep=0.1cm,minimum size=0.4cm] (q2) [right of=q1] {$q_2$}; 
		\node[state] (dummy) [below of =q2, draw=none, xshift=-1.2cm, yshift=1.2cm] {\Huge $\Downarrow$};

		\node[state] [left of=q1, draw=none,xshift=0.9cm] {$\N:$};
		
		\node[initial,state, inner sep=0.01cm,minimum size=0.5cm] (p0) [below of=q1, yshift=-0.5cm] {$0,\infty$};
		\node[state, inner sep=0.05cm,minimum size=0.5cm] (p1) [right of=p0] {$0,2$}; 
		\node[state, inner sep=0.05cm,minimum size=0.5cm] (p2) [right of=p1] {$2,0$}; 
		\node[state, inner sep=0.01cm,minimum size=0.5cm] (p3) [right of=p2] {$\infty,0$}; 

		\node[state] [left of=p0, draw=none,xshift=0.9cm] {$\D:$};
		
		\path 
		(q1) edge [out=0, in=180] node [above,align=center] {$a,5,2$} (q2)
		(q1) edge [out=130, in=70,loop, looseness=5] node [above,align=center] {$a,4,2$} (q1)
		(q2) edge [out=130, in=70,loop, looseness=5] node [above,align=center] {$a,1,2$} (q2)
		
		(p0) edge [out=0, in=180] node [above,align=center] {$a,4,2$} (p1)
		(p1) edge [out=0, in=180] node [above,align=center] {$a,3,2$} (p2)
		(p2) edge [out=0, in=180] node [above,align=center] {$a,1,2$} (p3)
		(p3) edge [out=130, in=70,loop, looseness=5] node [above,align=center] {$a,1,2$} (p3)
		;
	\end{tikzpicture}
	
	\caption[A simple example of the determinization procedure.]{\label{fig:detExampleSimple}A simple example of the determinization procedure presented in \Cref{sec:Determinizability}.} (A more involved and detailed example is given in \Cref{fig:detExample}.)
\end{figure}

The following technique is a generalization of the determinization algorithm presented in \cite{BH14} for NDAs. 
We give the basic ideas underlying the construction through a very simple example of an NMDA $\N$, depicted in \Cref{fig:detExampleSimple}, over the alphabet $\{a\}$:
\begin{itemize}
\item The fundamental building block of the construction is to store in each state of the DMDA $\D$ a tuple $(x_1, x_2)$, where $x_1$ and $x_2$ represent the ``gaps'' that the states $q_1$ and $q_2$ of $\N$, respectively, have. Intuitively, a gap $x_i$ stands for how much more ``expensive'' it is for the nondeterministic automaton to reach $q_i$ upon reading the current word prefix than to follow the optimal run on the word prefix. Notice that in our example, $\N$ favors the self loop $(q_1,a,q_1)$ for the word ``$a$'' (while favoring $(q_1,a,q_2)$ as the initial transition for all other words). 
Hence, after reading ``$a$'' the gap of $q_1$ is $0$ (a preferred run reaches it), and the gap of $q_2$ should express ``how much more expensive is it to reach $q_2$''. 
In absolute values, the run to $q_2$ is more expensive by $1$. 
Yet, since the value of every continuation is discounted (i.e., divided) by $2$, if the run through $q_2$ is going to be optimal for a prolonged word, its continuation must be cheaper by $1 \cdot 2=2$ than the run trough $q_1$. Hence, we set the gap of a state after reading a word prefix $u$ to be the extra cost of reaching it, multiplied by the accumulative discount factor along $u$.
\item Notice that the weight of a transition in $\D$, between reading words $u$ and $ua$, for a letter $a$, should be $\N(ua)-\N(u)$, multiplied by the accumulative discount factor along $u$. Back to our example, the transition from the initial state of $\D$ upon reading $a$ (back to $q_1)$ needs to incur a weight of $4$.
Reading an additional $a$, the NMDA $\N$ has a different preferred path, which is simulated by the gap calculations in $\D$: it is cheaper to pay $q_2$'s gap of $2$ and the transition weight of $1$ (so $3$ in total) than to stay with the original run through $q_1$ and pay the transition weight of $4$. Thus, the equivalent transition in $\D$ will have weight of $3$. After reading the word $u$=``$aa$'', the gap of $q_1$ should be: (the difference between $\N(u)$ and $\N$'s best run on $u$ ending in $q_1$) multiplied by the accumulative discount factor along $u$. Observe that we can calculate the gap by the information in $\D$'s states, without considering the prefix word $u$: the gap of $q_1$ is equivalent to the previous gap of the path that leads to $q_1$ ($0$ in the example) plus the cost of continuing this path to $q_1$ ($4$ in our case) minus $\D$'s transition weight ($3$ here), multiplied by the discount factor of the last transition, namely in the example it is $(0+4-3)\cdot 2 = 2$.
\item To argue that such a DMDA is guaranteed to be finite, we make use of two observations: \begin{enumerate}
	\item A gap larger than $2T$, where $T$ is the maximal difference between two weights in the NMDA, can never be recovered, namely a prefix run that has such a gap will never be part of an optimal run, and hence the gap can be set to $\infty$.
	\item All other gaps are integer multiplications of $1/d$, where $d$ is the least common divisor of weights in the NMDA.
\end{enumerate}
\end{itemize}
\noindent 
The determinization construction forms the basis for both algebraic closure (\Cref{sec:tidyAlgebraic}) and decidability of the decision problems (\Cref{sec:tidyDecisionProblems}).
Though the constructed deterministic automaton can be of exponential size compared to the original nondeterministic automaton, each of its states is of only polynomial size (\Cref{lem:Termination}). Hence, we can solve the decision problems in PSPACE, performing the determinization on-the-fly. 

\vspace{0.5cm}
\noindent\emph{The formal construction.}
Consider a tidy NMDA $\A = \tuple{\Sigma, Q, \iota, \delta, \gamma, \rho}$.

For every finite word $u\in\Sigma^*$ and state $q\in Q$, we define $S(q,u)$ to be the set of runs of $\A$ on $u$ ending in $q$, and $r_{(q,u)}$ to be a \emph{preferred run} that entails the minimal value among all the runs in $S(q,u)$.
Observe that every prefix of a preferred run is also a preferred run.
Hence, given the values of all the preferred runs on a certain finite word $u$, i.e., $\A(r_{(q,u)})$ for every $q\in Q$, we can calculate the values of the preferred runs on every word $u\con\sigma$ by 
$\A(r_{(q',u\con\sigma)})=\min\big\{\A(r_{(q,u)}) + \gamma(t) \ST t=(q,\sigma,q')\in\delta \big\}$.

Intuitively, every state of $\D$ that is reached after reading $u$ stores for each $q\in Q$ its ``gap'', which is the difference between $\A(u)$ and $\A(r_{(q,u)})$, ``normalized'' by multiplying it with the accumulated discount factor $\rho(u)$, and ``truncated'' if reached a threshold value (which can no longer be recovered). 

Formally, for a state $q\in Q$, and a finite word $u$, we define 
\begin{itemize}
	\item 
	The \emph{cost} of reaching $q$ over $u$ as \\
	$\Cost(q,u)
	= \min \big\{\A(r) \ST r \text{ is a run of } \A \text{ on } u \text{ s.t.\ } \delta(r) = q\big\} 
	=\min \big\{\A(r) \ST r \in S(q,u)\big\}$, 
	where $\min\emptyset = \infty$.
	\item
	The \emph{gap} of $q$ over $u$ as $\Gap(q,u)=\rho(u)\big(\Cost(q,u) - \A(u)\big)$.
	Intuitively, the gap stands for the value that a walk starting in $q$ should have, compared to a walk starting in $u$'s optimal ending state, in order to make a run through $q$ optimal.
\end{itemize}
Let $T$ be the maximum difference between the weights in $\A$, That is, 
$T = \max \big(|x-y| \ST x,y \in \image(\gamma) \big)$.
Since the difference between two infinite runs of $\A$ on a word $w$ is bounded by
$\sum_{i=0}^\infty \frac{T}{\prod_{j=0}^{i-1}{\rho\big(w[j]\big)}} \leq \sum_{i=0}^\infty \frac{T}{2^i} = 2T$, 
we define
the set of possible \emph{recoverable-gaps}
$G = \big\{ v \ST v\in\Rat \mbox{ and } 0 \leq v \leq 2T \big\} \cup \{\infty \}$.
The $\infty$ element denotes a non-recoverable gap, and behaves as the standard infinity element in the algebraic operations that we will be using. Note that our NMDAs do not have infinite weights and the infinite element is only used as an internal component of the construction.

We will inductively construct 
$\D = \tuple{\Sigma, Q', q'_{in}, \delta', \gamma', \rho'}$ as follows.
A state of $\D$ extends the standard subset construction by assigning a gap to each state of $\A$. That is, for $Q=\{q_1, \cdots,
q_n \}$, a state $p\in Q'$ is a tuple $\tuple{g_1, \cdots, g_n}$, where $g_h\in G$ for every $h\in[1..n]$. 
Once a gap is obviously not recoverable, by being larger than $2T$, it is truncated by setting it to be $\infty$.

In the integral $\rho$ function case, the construction only requires finitely many elements of $G$, as shown in \Cref{lem:Termination}, and thus it is guaranteed to terminate.

For simplicity, we assume that $\iota=\{q_1,q_2,\cdots,q_{|\iota|}\}$ and extend $\gamma$ with $\gamma(q_i,\sigma,q_j)=\infty$ for every $({q_i,\sigma,q_j})\not\in\delta$. 
The initial state of $\D$ is $q'_{in} = \tuple{0,\cdots,0,\infty,\cdots,\infty}$, in which the left $|\iota|$ elements are $0$, meaning that the initial states of $\A$ have a $0$ gap and the others are currently not relevant.

We inductively build the desired automaton
$\D$ using the intermediate automata 
$\D_i= \tuple{\Sigma, Q'_i, q'_{in}, \delta'_i, \gamma'_i, \rho'_i}$. We start with
$\D_1$, in which $Q'_1 = \{ q'_{in} \}$, $\delta'_1 = \emptyset$, $\gamma'_1 = \emptyset$ and $\rho'_1 = \emptyset$, and proceed from $\D_i$ to $\D_{i+1}$, such that
$Q'_i \subseteq Q'_{i+1}$, $\delta'_i \subseteq \delta'_{i+1}$, $\gamma'_i \subseteq \gamma'_{i+1}$ and $\rho'_i \subseteq \rho'_{i+1}$. 
The construction is completed once
$\D_i = \D_{i+1}$, finalizing the desired deterministic automaton $\D = \D_i$.

In the induction step, $\D_{i+1}$ extends $\D_i$ by (possibly) adding, for every state $q'=\tuple{g_1, \cdots, g_n}\in Q'_i$ and letter $\sigma\in\Sigma$, a state $q'':= \tuple{x_1, \cdots, x_n}$, and a transition $t:=(q', \sigma, q'')$ as follows:
\begin{itemize}
	\item 
	Weight: For every $h\in[1..n]$ define, 
	
	$c_h := \min \big\{g_j + \gamma(q_j, \sigma, q_h) \ST j\in[1..n]  \big\}$, and add a new weight, $\gamma'_{i+1}(t)= \min\limits_{1 \leq h \leq n}(c_h)$.
	
	\item
	Discount factor: By the induction construction, if $\D_i$ running on a finite word $u$ ends in $q'$, there is a run of $\A$ on $u$ ending in $q_h$, for every $h\in[1..n]$ for which the gap $g_h$ in $q'$ is not $\infty$.
	Since $\A$ is tidy, all the transitions from every such state $q_h$ over $\sigma$ have the same discount factor, which we set to the new transition $\rho'_{i+1}(t)$.
	
	\item 
	Gap: For every $h\in[1..n]$, set $x_h:=\rho'_{i+1}(t)\cdot\big(c_h-\gamma'_{i+1}(t)\big)$. If $x_h > 2T$ then set $x_h:=\infty$.
\end{itemize}

See \Cref{fig:detExample} for an example of the determinization process.

\begin{figure}
	\centering
	
	\begin{tikzpicture}[->,>=stealth',shorten >=1pt,auto,node distance=2.5cm, semithick, initial text=, every initial by arrow/.style={|->}]
		\node[initial above, state, inner sep=0.1cm,minimum size=0.4cm] (q1) {$q_1$};
		\node[state, inner sep=0.1cm,minimum size=0.4cm] (q2) [right of=q1] {$q_2$}; 
		\node[state] (dummy) [right of =q2, draw=none, xshift=-0.8cm] {\Huge $\Rightarrow$};
		
		\node[initial,state, inner sep=0.01cm,minimum size=0.4cm] (p0inf) [right of=q1, xshift=4cm] {$0,\infty$};
		\node[state, inner sep=0.05cm,minimum size=0.4cm] (p02) [right of=p0inf] {$0,2$}; 
		\node[state, inner sep=0.01cm,minimum size=0.4cm] (pinf0) [below of=p0inf] {$\infty,0$}; 
		\node[state] (b13callout) [above of=pinf0, xshift=-1cm, yshift=-1.2cm, draw=none] {};
		\node[state] (aminus2callout) [below right of=pinf0, xshift=0cm, yshift=1.3cm, draw=none] {};
		\node[state, inner sep=0.05cm,minimum size=0.4cm] (p01) [above of=p02] {$0,1$};
		\node[state, inner sep=0.05cm,minimum size=0.4cm] (p20) [below of=p02] {$2,0$};
		\node[state] (callout3) [above of=p0inf, xshift=-0.5cm, yshift=-0.5cm, draw=none] {}; 
		
		\node[overlay,rectangle callout,fill=gray!20,align=left, scale=0.6, callout absolute pointer={(callout3)}] 
		at ($(4cm,2.4cm)$) {$c_1=\min({0-\frac{1}{3},\infty+1})=-\frac{1}{3}$\\
			$c_2=\min({0+0,\infty-2})=0$\\
			$c=\min(-\frac{1}{3},0)=-\frac{1}{3}$\\
			$x_2=3(0-(-\frac{1}{3})) = 1$};
		
		\node[overlay,rectangle callout,fill=gray!20,align=left, scale=0.6, callout absolute pointer={(b13callout)}] 
		at ($(3.7cm,-1.5cm)$) {$c_1=\min({\infty-1,0+1})=1$\\
			$c_2=\min({\infty+0,0+\infty})=\infty$\\
			$c=\min(1,\infty)=1$\\
			$x_2=2(\infty-1)=\infty$};
		
		\node[rectangle callout,fill=gray!20,align=left, scale=0.6, callout absolute pointer={(aminus2callout)}] 
		at ($(11cm,-3.3cm)$) {
			$c_1=\min({2-\frac{1}{3},0+1})=1$\\
			$c_2=\min({2+0,0-2})=-2$\\
			$c=\min(1,-2)=-2$\\
			$x_1=3\big(1-(-2)\big)=9\rightsquigarrow\infty$};
		
		\path 
		(q1) edge [out=30, in=150] node [above,align=center] {$a,0,3$\\$b,0,2$} (q2)
		(q2) edge [out=-150, in=-30] node [below,align=center] {$a,1,3$\\$b,1,2$} (q1)
		(q1) edge [out=-70, in=-130,loop, looseness=5] node [below left,align=center] {$a,-\frac{1}{3},3$\\$b,-1,2$} (q1)
		(q2) edge [out=110, in=50,loop, looseness=5] node [above right,align=center] {$a,-2,3$} (q2)
		
		(p0inf) edge [out= 80, in = -170] node [left] {$a,-\frac{1}{3},3$} (p01)
		(p0inf) edge node [above] {$b,-1,2$} (p02)
		(p01) edge [out=-60, in=60] node [right, near end] {$b,-1,2$} (p02)
		(p02) edge [out=120, in=-120] node [left, yshift=-0.15cm] {$a,-\frac{1}{3},3$} (p01)
		(p01) edge [out=0, in=0] node [right] {$a,-1,3$} (p20)
		(p20) edge node [below] {$a,-2,3$} (pinf0)
		(p20) edge node [left] {$b,1,2$} (p02)
		(pinf0) edge node [left] {$b,1,2$} (p0inf)
		(pinf0) edge [out=-150, in=150,loop, looseness=5] node [left] {$a,-2,3$} (pinf0)
		(p02) edge [out=-0, in=-40,loop, looseness=5] node [below, xshift=0.3cm, yshift=-0.1cm] {$b,-1,2$} (p02)
		;
	\end{tikzpicture}
	
	\caption[An example of the determinization procedure.]{\label{fig:detExample}An example of the determinization procedure, as per \Cref{thm:DetTidy}. The gray rectangles detail some of the intermediate calculations.}
\end{figure}
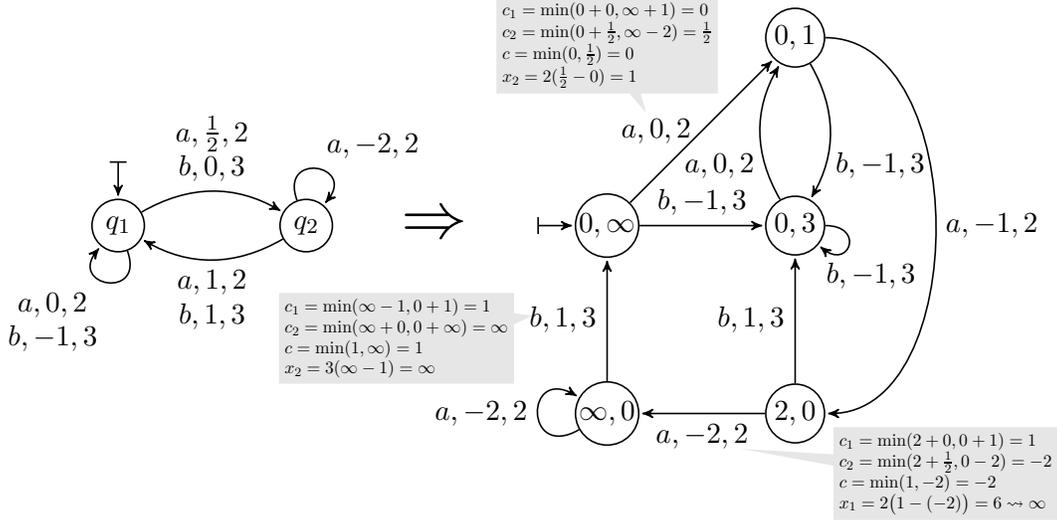
\noindent 
We prove below that the procedure always terminates for a tidy NMDA, and that every state of the generated DMDA can be represented in PSPACE. 
The proof is similar to the corresponding proof in \cite{BH14} with respect to NDAs, adding the necessary extensions for tidy NMDAs.
\begin{lem}\label{lem:Termination}
	The above determinization procedure	always terminates for a tidy NMDA $\A$. 
	Every state of the resulting deterministic automaton $\D$ can be represented in space polynomial in $|\A|$, and $|\D|\in 2^{O(|\A|)}$.
\end{lem}
\begin{proof}
	The induction step of the construction, extending $\D_i$ to $\D_{i+1}$, only depends on $\A$, $\Sigma$ and $Q'_i$. 
	Furthermore, for every $i\geq 0$, we have that $Q'_i \subseteq Q'_{i+1}$. 
	Thus, for showing the termination of the construction, it is enough to show that there is a general bound on the size of the sets $Q'_i$.  
	We do it by showing that the inner values, $g_1, \ldots, g_n$, of every state $q'$ of every set $Q'_i$ are from the finite set $\bar{G}$, defined below. 
	
	Let $d\in\Nat$ be the least common denominator of the weights in $\A$, and let $T\in\Nat$ be the maximal difference between the weights. 
	We define the set $\bar{G}$ as
	$$\bar{G} = \Big\{ \frac{k}{d} \ST k\in\Nat \text{ and } \frac{k}{d} \leq 2T \Big\} \cup \{\infty \}$$
	
	We start with the first set of states
	$Q'_1$, which satisfies the property that the inner values, $g_1, \ldots, g_n$, of every state $q'\in Q'_1$ are from $\bar{G}$, as $Q'_1=\{  \tuple{0,\cdots,0,\infty,\cdots,\infty}  \}$. We proceed by induction on the construction steps, assuming that $Q'_i$ satisfies the property. 
	By the construction, an inner value of a state $q''$ of $Q'_{i+1}$ is derived by four operations on elements of $\bar{G}$: addition, subtraction ($x-y$, where $x \geq y$), multiplication by $\lambda\in\image(\rho)\subset\Nat$, and taking the minimum. 
	
	One may verify that applying these four operations on $\infty$ and numbers of the form $\frac{k}{d}$, where $k\in\Nat$, results in $\infty$ or in a number $\frac{k'}{d}$, where $k'\in\Nat$.
	Recall that once an inner value exceeds $2T$, it is replaced by the procedure with $\infty$, meaning that $\frac{k'}{d}\leq2T$, or the calculated inner value is $\infty$. Concluding that all the inner values are in $\bar{G}$.
	
	Observe that $|\bar{G}|\leq 2Td + 2 $, meaning that every state in $\D$ has up to $2Td + 2$ possible values for each of the $|Q|$ inner elements, and that there are up to $(2Td + 2)^{|Q|}$ states in $\D$. 
	In particular, the procedure is guaranteed to terminate, the size of $\D$ is in $2^{O(|\A|)}$, and each of its states can be represented in space polynomial in $|\A|$.
\end{proof}

We will now show the correctness of the determinization procedure. According to \Cref{lemma:finiteToInfinite}, it is enough to show the equivalence $\D\equiv\A$ with respect to finite words.
\begin{lem}\label{lem:DetCorrectness}
	Consider a tidy NMDA $\A$ over $\Sigma^+$ and a DMDA $\D$, constructed from $\A$ by the above determinization procedure.
	Then, for every $u\in\Sigma^+$, we have 
	\begin{enumerate}[label=\roman*.,align=left]
		\item $\A(u) = \D(u)$.
		\item For every $h\in[1..n]$,  $g_h=\Gap(q_h,u)$ if $\Gap(q_h,u) \leq 2T$ and $\infty$ otherwise.
	\end{enumerate}
	\noindent 
	where $\tuple{g_1, \cdots, g_n}$ is the target state of the run of $\D$ on $u$.
\end{lem}
\begin{proof}
	Let $\A = \tuple{\Sigma, Q, \iota, \delta, \gamma, \rho}$ be the input NMDA, 
	$\D = \tuple{\Sigma, Q', \iota', \delta', \gamma', \rho'}$ the DMDA constructed from $\A$, and $T$ be the maximal difference between the weights in $\A$.
	
	For a finite word $u$, let $\delta'(u)=\tuple{g_1, \cdots, g_n}\in Q'$ be the target state of $\D$'s run on $u$. We show the claims $i$ an $ii$ above by induction on the length of the input word $u$.
	The assumptions obviously hold for the initial step, where $u$ is the empty word. As for the induction step, we assume they hold for $u$ and show that for every $\sigma\in\Sigma$, they hold for $u\con\sigma$. Let $\delta'(u\con\sigma)=\tuple{x_1, \cdots, x_n}\in Q'$ be the target state of $\D$'s run on $u\con\sigma$. 
	
	We start by proving the claim with respect to an \emph{infinite-state} automaton $\D'$ that is constructed as in the determinization procedure,
	except for not changing any gap to $\infty$. Afterwards, we shall argue that changing all gaps that exceed $2T$ to $\infty$ does not harm the correctness.
	\noindent 
	\begin{enumerate}[label=\roman*.,align=left]
		\item
		By the definitions of $\Cost$ and $\Gap$, we have for every $h\in[1..n]$, 
		\begin{align}
			\Cost(q_h, u\con\sigma)&=
			\min\limits_{j\in[1..n]} \Bigg(\Cost(q_j,u) + \frac{\gamma(q_j, \sigma, q_h)}{\rho(u)}\Bigg)  
			\nonumber\\
			&=\min\limits_{j\in[1..n]} \Bigg(\frac{\Gap(q_j,u)}{\rho(u)} + \A(u) + \frac{\gamma(q_j, \sigma, q_h)}{\rho(u)}\Bigg)
			\nonumber\\
			&=\A(u) + \frac{\min\limits_{j\in[1..n]} \Big(\Gap(q_j,u) + \gamma(q_j, \sigma, q_h) \Big)}{\rho(u)}
			=\mbox{\footnotesize{By the induction assumption}}
			\nonumber\\
			&=\D'(u) + \frac{ \min\limits_{j\in[1..n]} \Big(g_j + \gamma(q_j, \sigma, q_h) \Big)}{\rho(u)}
			\label{eqn:cost1}
		\end{align}
		\noindent 
		By the construction of $\D'$, the transition weight $\gamma'_{i}(t)$ assigned on the $i=|u|+1$ step is 
		
		$ \gamma'_{|u|+1}(t) = \min\limits_{h\in[1..n]}  \Big( \min\limits_{j\in[1..n]} (g_j + \gamma(q_j, \sigma, q_h) )\Big)$.
		Therefore, 
		\begin{align*}
			\D'(u\con\sigma) 
			&= \D'(u) + \frac{\gamma'_{|u|+1}(t)}{\rho(u)}\\
			&=  \D'(u) + \frac{\min\limits_{h\in[1..n]}   \min\limits_{j\in[1..n]} \Big(g_j + \gamma(q_j, \sigma, q_h) \Big)}{\rho(u)} \\
			&=  \min\limits_{h\in[1..n]} \Bigg(\D'(u) +   \frac{\min\limits_{j\in[1..n]} \Big(g_j + \gamma(q_j, \sigma, q_h) \Big)}{\rho(u)}\Bigg)\\
			&=  \min\limits_{h\in[1..n]} \Cost(q_h, u\con\sigma) =  \A(u\con\sigma)
		\end{align*}
		
		\item   By \Cref{eqn:cost1}, we get that for every $h\in[1..n]$:
		\begin{align*}
			\min\limits_{j\in[1..n]} (g_j + \gamma(q_j, \sigma, q_h)) =\rho(u)\Big(\Cost(q_h, u\con\sigma) - \D'(u) \Big)
		\end{align*}
		
		Let $t$ be the transition that was added in the $i=|u|+1$ step of the algorithm from the state $\delta'(u)$ over the $\sigma$ letter. 
		
		For every $h\in[1..n]$, we have
		\begin{align*}
			x_h &= \rho'_{i}(t)\cdot (c_h - \gamma'_{i}(t)) \\
			&= \rho'_{i}(t) \Big(\min\limits_{j\in[1..n]} (g_j + \gamma(q_j, \sigma, q_h) ) - \gamma'_{i}(t)\Big)\\
			& 
			\begin{aligned}
				=\rho'_{i}(t) \Bigg(\min\limits_{j\in[1..n]} (g_j + &\gamma(q_j, \sigma, q_h) )
				-\rho(u)\Big(\D'(u\con\sigma) - \D'(u)\Big)\Bigg)
			\end{aligned} \\
			& 
			\begin{aligned}
				=\rho'_{i}(t) \Bigg(\rho(u)\Big(&\Cost(q_h, u\con\sigma) - \D'(u) \Big) 
				-\rho(u)\Big(\D'(u\con\sigma) - \D'(u)\Big)\Bigg)
			\end{aligned}\\
			&= \rho'_{i}(t)\cdot\rho(u) \Big(\Cost(q_h, u\con\sigma) - \D'(u\con\sigma) \Big)\\
			&= \rho(u\con\sigma) \cdot \Big(\Cost(q_h, u\con\sigma) - \D'(u\con\sigma) \Big)
		\end{align*}
		And by the induction assumption we have
		\begin{align*}
			x_h&= \rho(u\con\sigma) \cdot\Big(\Cost(q_h, u\con\sigma) - \A(u\con\sigma) \Big) = \Gap(q_h,u\con\sigma)
		\end{align*}
		
	\end{enumerate}
	\noindent 
	It is left to show that the induction is also correct for the \emph{finite-state} automaton $\D$. 
	The only difference between the construction of $\D$ and of $\D'$ is that the former changes all gaps $(g_j$) above $2T$ to $\infty$. 
	We should thus show that if the gap $g_j$, for some $j\in[1..n]$, exceeds $2T$ at a step $i$ of the construction, and this $g_j$ influences the next gap of some state $h$ (we denoted this gap in the construction as $x_h$) then $x_h > 2T$. 
	This implies that $\D(u)=\D'(u)$, since at every step of the construction there is at least one $h\in[1..n]$, such that $x_h=0$, corresponding to an optimal run of $\A$ on $u$ ending in state $q_h$. 
	
	Formally, we should show that if $g_j> 2T$ and
	$x_h = \rho'_{i+1}(t)\cdot \Big(g_j + \gamma(q_j, \sigma, q_h) - \gamma'_{i+1}(t)\Big)$, 
	where $t$ is the transition added in the construction on step $i$ as defined in part (ii.) above, then $x_h > 2T$. 
	Indeed, according to the construction, there exists an index $k\in[1..n]$ such that $g_k=0$ and since $\A$ is complete, there is a transition from $q_k$ to some state $q_m$, implying that $\gamma'_{i+1}(t)\leq g_k+\gamma(q_k,\sigma,q_m)=\gamma(q_k,\sigma,q_m)$.
	Hence
	\begin{align*}
		x_h &> \rho'_{i+1}(t) \cdot\Big(2T + \gamma(q_j, \sigma, q_h) - \gamma'_{i+1}(t)\Big)  \geq 2 \cdot\Big(2T + \gamma(q_j, \sigma, q_h) - \gamma'_{i+1}(t)\Big) \\
		&\geq 2 \cdot\Big(2T + \gamma(q_j, \sigma, q_h) - \gamma(q_k, \sigma, q_m) \Big) \geq 2 \cdot (2T + (-T) ) = 2T
		\qedhere
	\end{align*}
\end{proof}

We show next that the DMDA created by the determinization procedure is indeed a $\theta$-DMDA.
\begin{lem}\label{lem:DetCorrectnessDiscountFunction}
	Consider a $\theta$-NMDA $\A$ over $\Sigma^+$ and a DMDA $\D$, constructed from $\A$ by the determinization procedure above. 
	Then $\D$ is a $\theta$-DMDA.
\end{lem}
\begin{proof}
	Consider a tidy NMDA $\A = \tuple{\Sigma, Q, \iota, \delta, \gamma, \rho}$, and the DMDA $\D = \tuple{\Sigma, Q', \iota', \delta', \gamma', \rho'}$ constructed from $\A$. 
	
	We show by induction on the length of an input word that for every finite word $u\in\Sigma^*$, we have $\rho'(u)=\rho(u)$. 
	The base case regarding the empty word obviously holds.	
	As for the induction step, we assume the claim holds for $u$ and show that it also holds for $u\con\sigma$, for every $\sigma\in\Sigma$.
	
	Let $t$ be the final transition of $\D$'s run on $u\con\sigma$.
	Due to the construction of $\D$, there exist $q,q' \in Q$ such that
	$\Gap(q,u)\neq\infty$, 
	$\Gap(q',u\con\sigma)\neq \infty$,
	and
	$
	\rho'(t)=\rho(q,\sigma,q')
	$.
	
	Hence,
	$
	\rho'(u\con\sigma) 
	= \rho'(u)\cdot \rho'(t) = \rho(u)\cdot \rho'(t) = \rho(u)\cdot \rho(q,\sigma,q')
	$ and since $\Gap(q,u)\neq\infty$, we get that $q \in\delta(u)$, and
	$
	\rho'(u\con\sigma) 
	= \rho(u)\cdot \rho(q,\sigma,q') = \rho(u\con\sigma)
	$.
\end{proof}

Finally, as a direct consequence of the above construction and 
 \Cref{lemma:finiteToInfinite,lem:DetCorrectness,lem:Termination,lem:DetCorrectnessDiscountFunction}:
\begin{thm}\label{thm:DetTidy}
	For every choice function $\theta$ and a $\theta$-NMDA $\A$, on finite or infinite words, there exists a $\theta$-DMDA $\D\equiv \A$ of size in $2^{O(|\A|)}$.
	Every state of $\D$ can be represented in space polynomial in $|\A|$.
\end{thm}

\subsection{Representing Choice Functions}\label{sec:RepresentingChoiceFunctions}

We show that, as opposed to the case of a general function $f: \Sigma^+ \to \Nat \setminus \{0,1\}$, every choice function $\theta$ can be finitely represented by a transducer.

A transducer $\T$ (Mealy machine) is a 6-tuple $\tuple{P,\Sigma,\Gamma,p_0,\delta,\rho}$, where $P$ is a finite set of states, $\Sigma$ and $\Gamma$ are finite sets called the input and output alphabets, $p_0\in P$ is the initial state, $\delta: P\times\Sigma\to P$ is the total transition function and $\rho: P\times\Sigma\to\Gamma$ is the total output function.

A transducer $\T$ represents a function, to which for simplicity we give the same name $\T:\Sigma^+\to \Gamma$, such that for every word $w$, the value $\T(w)$ is the output label of the last transition taken when running $\T$ on $w$.
The size of $\T$, denoted by $|\T|$, is the maximum between the number of transitions and the maximal binary representation of any output in the range of $\rho$.

Since in this work we only consider transducers in which the output alphabet $\Gamma$ is the natural numbers $\Nat$, we omit $\Gamma$ from their description, namely write $\tuple{P,\Sigma,p_0,\delta,\rho}$ instead of $\tuple{P,\Sigma,\Nat,p_0,\delta,\rho}$.
An example of a transducer $\T$ and a $\T$-NMDA is given in \Cref{fig:TransducerAndAutomaton}.

\begin{figure}
	\centering
	\begin{tikzpicture}[->,>=stealth',shorten >=1pt,auto,node distance=2cm, semithick, initial text=, every initial by arrow/.style={|->}]
		\node [midway, xshift=0cm] [above left of=q0] {$\T:$};
		\node[initial, state] (q0) {$q_0$};
		\node[state] (q1) [right of=q0] {$q_1$};
		
		\path 
		(q0) edge	[loop above, out=120, in=70,looseness=5] node {$a,2$} (q0)
		(q1) edge	[loop above, out=120,in=70,looseness=5] node {$a,3$} (q1)
		(q1) edge	[below, out=-160,in=-20] node {$b,2$} (q0)
		(q0) edge [above, out=20,in=160] node {$b,4$} (q1)
		;

		\node[initial, state, xshift=1.5cm] (p0) [above right of=q1] {$p_0$};
		
		\node[state, xshift=2.5cm] (p1) [right of=p0] {$p_1$};
		\node[state] (p2) [below of=p0] {$p_2$};
		\node[state] (p3) [below of=p1] {$p_3$};
		\node[state] [above of=p0, draw=none, xshift=-1.6cm, yshift=-1.8cm] {$\A:$};
		\path
		(p0) edge [left, in=140, out = -140] node {$a,1,2$} (p2)
		(p2) edge [left, out=40, in = -40] node {$a,\frac{1}{2},2$} (p0)
		(p0) edge node[yshift=-0.22cm] {$b,2,4$} (p3)
		(p1) edge	[loop, out=-30,in=10,looseness=5] node [right] {$a,1,2$} (p1)
		(p1) edge [right, in=140, out = -140] node {$b,\frac{1}{2},4$} (p3)
		(p3) edge [right, out=40, in = -40] node {$b,\frac{2}{3},2$} (p1)
		(p2) edge [above, in=170, out = 10] node[yshift=-0.05cm] {$b,1,4$} (p3)
		
		(p0) edge node[below,yshift=0.07cm] {$a,\frac{3}{2},2$} (p1)
		(p3) edge [below, out=-170, in = -10] node[above,yshift=-0.1cm] {$b,\frac{3}{4},2$} (p2)
		
		(p3) edge	[loop, out=-30,in=10,looseness=5] node [right]{$a,1,3$} (p3)
		;
	\end{tikzpicture}
	\caption{\label{fig:TransducerAndAutomaton}A transducer $\T$ and a $\T$-NMDA.}
\end{figure}

\begin{thm}\label{thm:transducersAreEnough}
	For every function $\theta: \Sigma^+ \to \Nat\setminus\{0,1\}$, $\theta$ is a choice function, namely there exists a $\theta$-NMDA, if and only if there exists a transducer $\T$ such that $\theta \equiv \T$.
\end{thm}
\begin{proof}
	Consider a function $\theta: \Sigma^+ \to \Nat\setminus\{0,1\}$.
	For the first direction, observe that given a transducer $\T=\tuple{P,\Sigma,p_0,\delta,\rho}$ representing $\theta$, 
	it holds that 
	the NMDA $\T'=\tuple{\Sigma,P,\{p_0\},\delta,\gamma,\rho}$, for every weight function $\gamma$, is a $\theta$-NMDA.
	
	For the other direction, consider a $\theta$-NMDA $\A'$. According to \Cref{thm:DetTidy}, there exists a $\theta$-DMDA $\A=\tuple{\Sigma, Q, q_0, \delta, \gamma, \rho}$ equivalent to $\A'$. Since the image of $\rho$ is a subset of $\Nat$, we have that $\theta$ can be represented by the transducer $\T=\tuple{Q,\Sigma,q_0,\delta,\rho}$.
\end{proof}
For a given choice function $\theta$, we refer to the class of all $\theta$-NMDAs. 
Observe that when considering such a class, only the choice function is relevant, regardless of the transducer defining it.

\subsection{Closure under Algebraic Operations}\label{sec:tidyAlgebraic}
We show that the family of $\theta$-NMDAs, for any fixed choice function $\theta$, is closed under algebraic operation. 
Namely, for every tidy-NMDAs $\A$ and $\B$ that share the same choice function ($\theta$), there exists tidy-NMDAs for $-\A$, $\A+\B$, $\A-\B$, $\min(\A,\B)$ and $\max(\A,\B)$, with that same choice function (all of them can be represented by $\theta$-NMDAs).
Observe that if the choice functions of $\A$ and $\B$ are not the same, the closure is not guaranteed (\Cref{thm:NoClosureGeneralIntegralNMDA}).
\begin{thm} \label{thm:closurealgebraicOps}
	For every choice function $\theta$, the set of $\theta$-NMDAs, on finite or infinite words, is closed under the operations of min, max, addition, subtraction, and multiplication by a rational constant.
\end{thm}
\begin{proof}
	Consider a choice function $\theta$ and $\theta$-NMDAs $\A$ and $\B$.
	
	\begin{itemize}
		\item 
		\emph{Multiplication by constant $c\geq 0$}: A $\theta$-NMDA for $c \cdot \A$ is straightforward from \Cref{prop:multiply}.
		\item
		\emph{Multiplication by $-1$}: A $\theta$-NMDA for $-\A$ can be achieved by first determinizing $\A$, as per \Cref{thm:DetTidy}, into a $\theta$-DMDA $\D$ and then multiplying all the weights in $\D$ by $-1$.
		\item 
		\emph{Addition}: 
		Considering  $\A=\tuple{\Sigma,Q_1,\iota_1,\delta_1,\gamma_1,\rho_1}$ and $\B=\tuple{\Sigma,Q_2,\iota_2,\delta_2,\gamma_2,\rho_2}$, a $\theta$-NMDA for $\A+ \B$ can be achieved by 
		constructing the product automaton
		$\C=\tuple{\Sigma,Q_1\times Q_2,\iota_1\times\iota_2,\delta,\gamma,\rho}$, where
		$\delta = \big\{ \big((q_1,q_2),\sigma,(p_1,p_2)\big)\ST (q_1,\sigma,p_1)\in\delta_1 \text{ and } (q_2,\sigma,p_2)\in\delta_2\big\}$,
		$\gamma\big((q_1,q_2), \sigma, (p_1,p_2)\big)= \gamma_1(q_1, \sigma,  p_1) + \gamma_2(q_2, \sigma,  p_2)$,
		$\rho\big((q_1,q_2), \sigma, (p_1,p_2)\big)= \rho_1(q_1, \sigma,  p_1) = \rho_2(q_2, \sigma,  p_2)$. The latter must hold since both $\rho_1$ and $\rho_2$ are compliant with $\theta$.
		
		\item 
		\emph{Subtraction}:
		A $\theta$-NMDA for $\A- \B$ can be achieved by i) Determinizing $\B$ to $\B'$;  ii) Multiplying $\B'$ by $-1$, getting $\B''$; and iii) Constructing a $\theta$-NMDA for $\A +\B''$.
		\item 
		\emph{min}: A $\theta$-NMDA for $\min (\A,\B)$ is straightforward by the nondeterminism on their union.
		\item 
		\emph{max}: A $\theta$-NMDA for $\max(\A, \B)$ can be achieved by i) Determinizing $\A$ and $\B$ to $\A'$ and $\B'$, respectively; ii) Multiplying $\A'$ and $\B'$ by $-1$, getting $\A''$ and $\B''$, respectively; iii) Constructing a $\theta$-NMDA $\C''$ for $\min(\A'',\B'')$; iv) Determinizing $\C''$ into a $\theta$-DMDA $\D$; and v) Multiplying $\D$ by $-1$, getting $\theta$-NMDA $\C$, which provides $\max(\A, \B)$.
		\qedhere
	\end{itemize}	
\end{proof}
\noindent 
We analyze next the size blow-up involved in algebraic operations. In addition to the general classes of $\theta$-NMDAs, we also consider the case where both input and output automata are deterministic.
Summation of the results can be seen in \Cref{tbl:blow-up}.

Most results in \Cref{tbl:blow-up} are straightforward from the constructions presented in the proof of \Cref{thm:closurealgebraicOps}: multiplying all the weights by a constant is linear, creating the product automaton is quadratic, and whenever determinization is required, we get an exponential blow-up.
However, the result of the size blow-up for the max operation on tidy NMDAs is a little more involved. At a first glance, determinizing back and forth might look like a doubly-exponential blow-up, however in this case an optimized determinization procedure can achieve a singly-exponential blow-up:
Determinizing a tidy NMDA that is the union of two DMDAs, in which the transition weights are polynomial in the number of states, is shown to only involve a polynomial size blow-up.

\begin{table}
	\begin{center}
		\begin{tabular}{|c|c||c||c|c|c|}
			\hline
			$c\cdot\A$ (for $c\geq 0$) &  $\min(\A,\B)$ & $\A+\B$ & $-\A$ & $\max(\A,\B)$ & $\A-\B$ \\ [0.5ex]
			\hline\hline
			\multicolumn{2}{|c||}{Linear} &Quadratic&\multicolumn{3}{c|}{Single Exponential}\\
			\hline
		\end{tabular}
	\end{center}
	\caption{The size blow-up involved in algebraic operations on tidy NMDAs.} \label{tbl:blow-up}
\end{table}	

\begin{thm}\label{thm:maxBlowUp}
	The size blow-up involved in the $\max$ operation on tidy NMDAs, on finite or infinite words, is at most single-exponential.
\end{thm}
\begin{proof}
	Consider a choice function $\theta$, $\theta$-NMDAs $\A$ and $\B$, and the automata $\A'',\B'',\C'',\D$ and $\C$, as constructed in the `$\max$' part of the proof of \Cref{thm:closurealgebraicOps}. Observe that $\C''$ is the the union of two $\theta$-DMDAs.
	As so, for every word $u$, there are only two possible runs of $\C''$ on $u$.
	In order to determinize $\C''$ into $\D$ we present a slightly modified procedure compared to the one presented in \Cref{sec:Determinizability}.
	Instead of the basic subset construction, we use the product automaton of $\A''$ and $\B''$ and instead of saving in every state of $\D$ the gap from the preferred state for every state of $\C''$, we only save the gap between the two runs of $\C''$.
	Combined with the observation we showed in the proof of \Cref{lem:DetCorrectness} that the weights of $\A''$ and $\B''$ are bounded by the weights of $\A$ and $\B$, we are able to reduce the overall blow-up to be only single-exponential.
	The procedure presented in \Cref{sec:Determinizability} requires the following modifications:
	\begin{itemize}
		\item Every state of $\D$ is a tuple $\tuple{q_1,q_2,g_1,g_2}$ where $q_1$ is a state of $\A''$, $q_2$ is a state of $\B''$, and $g_1,g_2\in G$ are the gaps from the preferred run.
		\item The initial state of $\D$ is $\tuple{q_\A,q_\B,0,0}$ where $q_\A$ and $q_\B$ are the initial states of $\A''$ and $\B''$, respectively.
		\item In the induction step, $\D_{i+1}$ extends $\D_{i}$ by (possibly) adding for every state $p=\tuple{q_1,q_2,g_1,g_2}$ and letter $\sigma\in\Sigma$, a state $p':=\tuple{q'_1,q'_2,g'_1,g'_2}$ and a transition $t:=\tuple{p,\sigma,p'}$ such that for every $h\in[1..2]$:
		\begin{itemize}
			\item 
			$c_h := g_h + \gamma\big(q_h,\sigma,\delta(q_h,\sigma)\big)$
			\item
			$\gamma'_{i+1}(t)= \min(c_1,c_2)$
			\item
			$\rho'_{i+1}(t) = \rho\big(q_1,\sigma,\delta(q_1,\sigma)\big)$
			\item
			$x_h:=\rho'_{i+1}(t)\cdot\big(c_h-\gamma'_{i+1}(t)\big)$. If $x_h> 2T$ then set $x_h:=\infty$
		\end{itemize}
	\end{itemize}
	\noindent 
	With the above modifications, similarly to \Cref{lem:Termination}, we get that the number of possible gaps is $2T d_\A d_\B+2$ where $d_\A$ and $d_\B$ are the denominators of weights in $\A''$ and $\B''$, respectively.
	Hence, there are no more than $(2T d_\A d_\B+2)^2\cdot N_{\A} \cdot N_{\B}$ possibilities for the states of $\D$, where $N_{\A}$ and $N_{\B}$ are the number of states in $\A''$ and $\B''$, respectively.
	
	According to the determinization procedure showed in \Cref{sec:Determinizability} and as explained in the proofs of \Cref{lem:Termination} and \Cref{lem:DetCorrectness}, the following observations hold:
	\begin{itemize}
		\item 
		$d_\A$ and $d_\B$ are also the denominators of weights in $\A$ and $\B$, respectively, and since we use binary representation of weights, $d_\A\cdot d_\B$ is up to single-exponential in $|\A|+|\B|$.
		\item
		All the weights in $\A''$ and $\B''$ are bounded by the weights of $\A$ and $\B$, hence $T$ is also up to single-exponential in $|\A|+|\B|$.
		\item $N_{\A}$ and $N_{\B}$ are up to single-exponential in $|\A|+|\B|$.
	\end{itemize}
	Concluding that the number of states in $\D$ is up to single-exponential in $|\A|+|\B|$, and since the number of states in $\C$ is equal to the number of states in $\D$, we get a single-exponential blow-up.
\end{proof}

Observe that if weights are represented in unary, we can achieve a quartic blow-up for the min and max operations on tidy-DMDAs, by using the above determinization procedure, and since $T$ is linear in unary representation.

We are not aware of prior lower bounds on the size blow-up involved in algebraic operations on NDAs.
For achieving such lower bounds, we develop a general scheme to convert every NFA to a $\lambda$-NDA of linearly the same size that defines the same language with respect to a threshold value $0$, and to convert some specific $\lambda$-NDAs back to corresponding NFAs.

The conversion of an NFA to a corresponding $\lambda$-NDA is quite simple. It roughly uses the same structure of the original NFA, and assigns four different transition weights, depending on whether each of the source and target states is accepting or rejecting.

\begin{lem}\label{lem:NFA_to_NDA}
	For every $\lambda\in\Nat\setminus\{0,1\}$ and NFA $\A$ with $n$ states, there exists a $\lambda$-NDA $\Tilde{\A}$ with $n+2$ states, such that for every word $u\in\Sigma^+$, we have $u\in L(\A)$ iff $\Tilde{\A}(u)<0$. That is, the language defined by $\A$ is equivalent to the language defined by $\Tilde{\A}$ and the threshold $0$.
\end{lem}
\begin{proof}
	Given an NFA $\A=\tuple{\Sigma,Q,\iota,\delta,F}$ and a discount factor $\lambda\in\Nat\setminus\{0,1\}$, we construct a $\lambda$-NDA $\Tilde{\A}=\tuple{\Sigma,Q',\{p_0\},\delta',\gamma'}$ for which there exists a bijection $f$ between the runs of $\A$ and the runs of $\Tilde{\A}$ such that for every run $r$ of $\Tilde{\A}$ on a word $u$,
	\begin{itemize}
		\item 
		$r$ is an accepting run of $\A$ iff $f(r)$ is a run of $\Tilde{\A}$ on $u$ with the value $\Tilde{\A}\big(f(r)\big)=-\frac{1}{\lambda^{|r|}}$.
		\item
		$r$ is a non-accepting run of $\A$ iff $f(r)$ is a run of $\Tilde{\A}$ on $u$ with the value $\Tilde{\A}\big(f(r)\big)=\frac{1}{\lambda^{|r|}}$.	
	\end{itemize}
	We first transform $\A$ to an equivalent NFA $\A'=\tuple{\Sigma,Q',\{p_0\},\delta',F}$ that is complete and in which there are no transitions entering its initial state, and later assign weights to its transitions to create $\Tilde{\A}$.
	
	To construct $\A'$ we add two states to $Q$, having $Q'=Q\cup\{p_0,q_{hole}\}$, duplicate all the transitions from $\iota$ to start from $p_0$, and add a transition from every state to $q_{hole}$, namely $$\delta'=\delta\cup \big\{ (p_0,\sigma,q)\ST \exists p\in\iota,(p,\sigma,q)\in\delta\big\}\cup\big\{(q,\sigma,q_{hole})\ST q\in Q', \sigma\in \Sigma\big\}$$
	Observe that $|Q'|=|Q|+2$, and $L(\A)=L(\A')$.
	Next, we assign the following transition weights:
	\begin{itemize}
		\item For every $t=(p_0,\sigma,q)\in\delta'$, $\gamma'(t)=-\frac{1}{\lambda}$ if $q\in F$ and $\gamma'(t)=\frac{1}{\lambda}$ if $q\notin F$.
		\item For every $t=(p,\sigma,q)\in\delta'$ such that $p\neq p_0$,
		$\gamma'(t)=\frac{\lambda-1}{\lambda}$ if $p,q\in F$;
		$\gamma'(t)=\frac{\lambda+1}{\lambda}$ if $p\in F$ and $q\notin F$;
		$\gamma'(t)=-\frac{\lambda+1}{\lambda}$ if $p\notin F$ and $q\in F$;
		and $\gamma'(t)=-\frac{\lambda-1}{\lambda}$ if $p,q\notin F$.
	\end{itemize}
	By induction on the length of the runs on an input word $u$, one can show that for every $u\in\Sigma^+$, $\Tilde{\A}(u)=-\frac{1}{\lambda^{|u|}}$ if $u\in L(\A)$ and $\Tilde{\A}(u)=\frac{1}{\lambda^{|u|}}$ if $u\notin L(\A)$.
\end{proof}
Converting an NDA to a corresponding NFA is much more challenging, since a general NDA might have arbitrary weights.
We develop a conversion scheme, whose correctness proof is quite involved, from every NDA $\Dot{\B}$ that is equivalent to $-\Tilde{\A}$, where $\Tilde{\A}$ is generated from an arbitrary NFA as per \Cref{lem:NFA_to_NDA}, to a corresponding NFA $\B$. Notice that the assumption that $\Dot{\B}\equiv-\Tilde{\A}$ gives us some information on $\Dot{\B}$, yet $\Dot{\B}$ might a priori still have arbitrary transition weights.
Using this scheme, we provide an exponential lower bound on the size blow-up involved in multiplying an NDA by $(-1)$. The theorem holds with respect to both finite and infinite words.

\begin{thm}\label{thm:arithmeticMinusoneLowerbound}
	For every $n\in\Nat$ and $\lambda\in \Nat\setminus\{0,1\}$, there exists a $\lambda$-NDA $\A$ with $n$ states over a fixed alphabet, such that every $\lambda$-NDA that is equivalent to $-\A$, w.r.t.\ finite or infinite words, has $\Omega(2^n)$ states.
\end{thm}
\begin{proof}
	Consider $n \in\Nat$ and $\lambda\in \Nat\setminus\{0,1\}$. 
	By \cite{SakodaS78,Jiraskova05} there exists an NFA $\A$ with $n$ states over a fixed alphabet of two letters, such that any NFA for the complement language $\overline{L(\A)}$ has at least $2^n$ states.
	
	\noindent{\emph{Finite words.}}
	
	Let $\Tilde{\A}$ be a $\lambda$-NDA that is correlated to $\A$ as per \Cref{lem:NFA_to_NDA}, and assume towards contradiction that there exists a $\lambda$-NDA $\Dot{\B}=\tuple{\Sigma, Q_{\Dot{\B}}, \iota_{\Dot{\B}}, \delta_{\Dot{\B}}, \gamma_{\Dot{\B}}}$ with less than $\frac{2^n}{4}$ states such that $\Dot{\B}\equiv -\Tilde{\A}$.
	
	We provide below a conversion opposite to \Cref{lem:NFA_to_NDA}, leading to an NFA for $\overline{L(\A)}$ with less than $2^n$ states, and therefore to a contradiction.
	The conversion of $\Dot{\B}$ back to an NFA builds on the specific values that $\Dot{\B}$ is known to assign to words, as opposed to the construction of \Cref{lem:NFA_to_NDA}, which works uniformly for every NFA, and is much more challenging, since $\Dot{\B}$ might have arbitrary transition weights.
	This conversion scheme can only work for $\lambda$-NDAs whose values on the input words converge to some threshold as the words length grow to infinity.
	
	For simplification, we do not consider the empty word, since one can easily check if the input NFA accepts it, and set the complemented NFA to reject it accordingly.
	
	By \Cref{lem:NFA_to_NDA} we have that for every word $u\in\Sigma^+$, $\Tilde{\A}(u)=-\frac{1}{\lambda^{|u|}}$ if $u\in L(\A)$ and $\Tilde{\A}(u)=\frac{1}{\lambda^{|u|}}$ if $u\notin L(\A)$. Hence, $\Dot{\B}(u)=-\frac{1}{\lambda^{|u|}}$ if $u\notin L(\A)$ and $\Dot{\B}(u)=\frac{1}{\lambda^{|u|}}$ if $u\in L(\A)$.
	We will show that there exists an NFA $\B$, with less than $2^n$ states, such that $u\in L(\B)$ iff $\Dot{\B}(u)=-\frac{1}{\lambda^{|u|}}$, implying that $L(B)=\overline{L(\A)}$.
	
	We first construct a $\lambda$-NDA $\B'=\tuple{\Sigma, Q_{\B'}, \iota, \delta, \gamma}$ that is equivalent to $\Dot{\B}$, 
	but has no transitions entering its initial states. 
	This construction eliminates the possibility that one run is a suffix of another, allowing to simplify some of our arguments. 
	Formally, $Q_{\B'}=Q_{\Dot{\B}}\cup\iota$, $\iota=\iota_{\Dot{\B}}\times\{1\}$, $\delta=\delta_{\Dot{\B}}\cup\big\{\big((p,1),\sigma,q\big)\ST (p,\sigma,q)\in\delta_{\Dot{\B}}\big\}$, and weights $\gamma(t)=\gamma_{\Dot{\B}}(t)$ if $t\in\delta_{\Dot{\B}}$ and $\gamma\big((p,1),\sigma,q\big)=\gamma_{\Dot{\B}}(p,\sigma,q)$ otherwise.
	
	Let $R^-$ be the set of all the runs of $\B'$ that entail a minimal value which is less than $0$, i.e., $R^-=\{r\ST r \text{ is a minimal run of }\B' \text{ on some word and } \B'(r)<0\}$.
	Let $\hat{\delta}\subseteq\delta$ be the set of all the transitions that take part in some run in $R^-$, meaning $\hat{\delta}=\{r(i) \ST r\in R^- \text{ and } i\in[0..|r|{-}1] \}$,
	and $\doublehat{\delta} \subseteq\delta$ the set of all transitions that are the last transition of those runs, meaning $\doublehat{\delta}=\big\{r\big(|r|-1\big) \ST r \in R^-\big\}$.
	
	We construct next the NFA $\B=\tuple{\Sigma, Q_{\B},\iota,\delta_{\B},F_{\B}}$. Intuitively, $\B$ has the states of $\B'$, but only the transitions from $\hat{\delta}$. 
	Its accepting states are clones of the target states of the transitions in $\doublehat{\delta}$, but without outgoing transitions. We will later show that the only runs of $\B$ that reach these clones are those that have an equivalent run in $R^-$.
	Formally,  
	$Q_{\B}=Q_\B'\cup F_{\B}$,
	$F_{\B}=\big\{(q,1)\ST \exists p,q\in Q_\B' \text{ and } (p,\sigma,q)\in \doublehat{\delta} \big\}$, and 
	$\delta_{\B}=\hat{\delta}\cup\big\{\big(p,\sigma,(q,1)\big)\ST(p,\sigma,q)\in \doublehat{\delta}\big\}$.
	
	Observe that the number of states in $\B$ is at most $3$ times the number of states in $\Dot{\B}$, and thus less than $2^n$.
	We will now prove that for every word $u$, $\B$ accepts $u$ iff $\B'(u)=-\frac{1}{\lambda^{|u|}}$.
	
	The first direction is easy: if $\B'(u)=-\frac{1}{\lambda^{|u|}}$, we get that all the transitions of a minimal run of $\B'$ on $u$ are in $\hat{\delta}$, and its final transition is in $\doublehat{\delta}$, hence there exists a run of $\B$ on $u$ ending at an accepting state.
	
	For the other direction, assume towards contradiction that there exists a word $u$, such that $\B'(u)=\frac{1}{\lambda^{|u|}}$, while there is an accepting run $r_u$ of $\B$ on $u$. 
	
	Intuitively, we define the ``normalized value'' of a run $r'$ of $\B'$ as the value of $\B'$ multiplied by the accumulated discount factor, i.e., $\B'(r')\cdot\lambda^{|r'|}$. 
	Whenever the normalized value reaches $-1$, we have an ``accepting'' run. 
	We will show that $r_u$ and the structure of $\B$ imply the existence of two ``accepting'' runs $r'_1,r'_2\in R^-$ that intersect in some state $q$, such that taking the prefix of $r'_1$ up to $q$ results in a normalized value $\lambda^k W_1$ that is strictly smaller than the normalized value $\lambda^j W_2$ of the prefix of $r'_2$ up to $q$.
	Since $r'_2$ is an ``accepting'' run, the suffix of $r'_2$ reduces $\lambda^j W_2$ to $-1$ and therefore it will reduce  $\lambda^k W_1$ to a value strictly smaller than $-1$, and the total value of the run to a value strictly smaller than $-\frac{1}{\lambda^n}$, which is not a possible value of $\B'$.
	
	Formally, let $r_u(|u|-1)=\big(p',u(|u|-1),(q',1)\big)$ be the final transition of $r_u$. 
	We replace it with the transition $t'=\big(p',u(|u|-1),q'\big)$. The resulting run $r'_u=r_u[0..|u|-2] \con t$ is a run of $\B'$ on $u$, and therefore $\B'(r'_u)\geq \frac{1}{\lambda^{|u|}}$.
	Since $(q',1)$ is an accepting state, we get by the construction of $\B$ that $t'$ is in $\doublehat{\delta}$.
	Consider a run $r'_1\in R^-$ that shares the maximal suffix with $r'_u$, meaning that if there exist $r'\in R^-$ and $x>0$ such that $r'[|r'|-x..|r'|-1]=r'_u[|u|-x..|u|-1]$ then also $r'_1[|r'_1|-x..|r'_1|-1]=r'_u[|u|-x..|u|-1]$.
	
	Recall that all the initial states of $\B'$ have no transitions entering them and $\B'(r'_1)\neq \B'(r'_u)$, hence $r'_1$ is not a suffix of $r'_u$ and $r'_u$ is not a suffix of $r'_1$.
	Let $i$ be the maximal index of $r'_u$ such that $r'_u[i..|u|-1]$ is a suffix of $r'_1$, but $r'_u[i-1..|u|-1]$ is not a suffix of $r'_1$.
	Let $k$ be the index in $r'_1$ such that $r'_1[k..|r'_1|-1]=r_u[i..|u|-1]$, and let $x = |r'_1|-k$ (see \Cref{fig:BlowupMinusOneLowerBound}).
	
	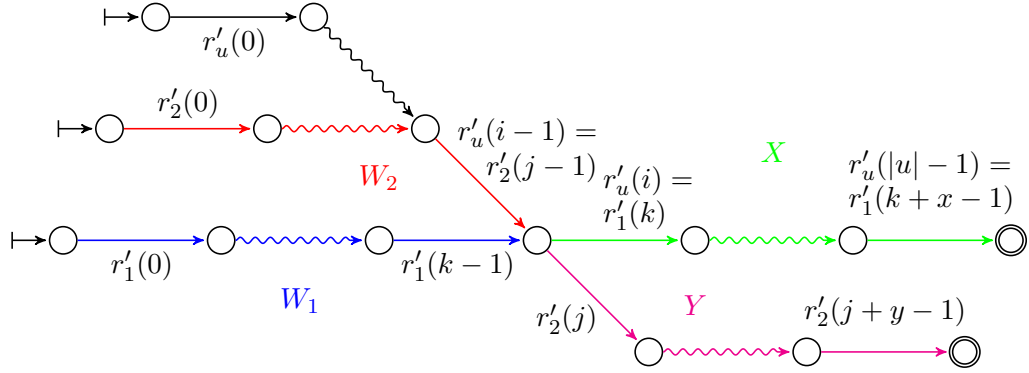
\begin{figure}
		\centering{}
		\begin{tikzpicture}[->,>=stealth',shorten >=1pt,auto,node distance=2.1cm, semithick, initial text=, every initial by arrow/.style={|->}, state/.style={circle, draw, minimum size=0pt}]
			\node[initial, state] (q0) {};
			\node[state] (q1) [right of=q0] {}; 
			\node[state] (q2) [right of=q1] {}; 
			\node[state] (q3) [right of=q2] {};
			\node[state] (q4) [right of=q3] {};
			\node[state] (q5) [right of=q4] {};
			\node[state, accepting] (q6) [right of=q5] {};
			\node[state] (p2) [above left of=q3] {};
			\node[state] (p1) [left of=p2] {};
			\node[initial, state] (p0) [left of=p1] {};
			\node[state] (p4) [below right of=q3] {};
			\node[state] (p5) [right of=p4] {};
			\node[state, accepting] (p6) [right of=p5] {};
			\node[state] (s1) [above left of=p2] {};
			\node[initial, state] (s0) [left of=s1] {};
			\draw[->,decorate,color=blue,decoration={snake,amplitude=.4mm,segment length=2mm,post length=1.5mm}]
			(q1) -- (q2);
			\draw[->,decorate,color=green,decoration={snake,amplitude=.4mm,segment length=2mm,post length=1.5mm}]
			(q4) -- (q5);
			\draw[->,decorate,color=red,decoration={snake,amplitude=.4mm,segment length=2mm,post length=1.5mm}]
			(p1) -- (p2);
			\draw[->,decorate,color=magenta,decoration={snake,amplitude=.4mm,segment length=2mm,post length=1.5mm}]
			(p4) -- (p5);
			\draw[->,decorate,decoration={snake,amplitude=.4mm,segment length=2mm,post length=1.5mm}]
			(s1) -- (p2);
			\path 
			(q0) edge [color=blue] node [color=black,below] {$r'_1(0)$} (q1)
			(q2) edge [color=blue] node [color=black,below] {$r'_1(k-1)$} (q3)
			(q3) edge [color=green] node [color=black,above,align=left, near end] {$r'_u(i)=$\\$r'_1(k)$} (q4)
			(q5) edge [color=green] node [color=black,above,align=left, yshift=0.2cm] {$r'_u(|u|-1)=$\\$r'_1(k+x-1)$} (q6)
			(p0) edge [color=red] node [above, color=black] {$r'_2(0)$} (p1)
			(p2) edge [color=red] node [color=black,right,align=right,very near start] {$r'_u(i-1)=$\\$r'_2(j-1)$} (q3)
			(q3) edge [color=magenta] node [color=black,left, near end, xshift=-0.1cm] {$r'_2(j)$} (p4)
			(p5) edge [color=magenta] node [color=black,above,yshift=0.2cm] {$r'_2(j+y-1)$} (p6)
			(s0) edge node [below] {$r'_u(0)$} (s1)
			(q1) edge [draw=none] node [below, color=blue,yshift=-0.5cm]{$W_1$} (q2)
			(q4) edge [draw=none] node [above, color=green,yshift=0.9cm]{$X$} (q5)
			(p4) edge [draw=none] node [above, color=magenta,near start, yshift=0.35cm]{$Y$} (p5)
			(p1) edge [draw=none] node [below, color=red,near end, yshift=-0.35cm]{$W_2$} (p2)
			;
		\end{tikzpicture}
		\caption{\label{fig:BlowupMinusOneLowerBound}The runs and notations used in the proof of \Cref{thm:arithmeticMinusoneLowerbound}.}
	\end{figure}
	
	Since $r'_u(i-1)\in\hat{\delta}$, there exists $r'_2\in R^-$ and index $j$ such that $r'_2(j-1)=r'_u(i-1)$. Let $y=|r'_2|-j$ (see \Cref{fig:BlowupMinusOneLowerBound}).
	Consider the run $r'_3=r'_2[0..j-1]\con r'_u[i..|u|-1]$, starting with the prefix of $r'_2$ up to the shared transition with $r'_u$, and then continuing with the suffix of $r'_u$. Observe that $\B'(r'_3)>-\frac{1}{\lambda^{|r'_3|}}$ as otherwise $r'_3\in R^-$ and has a larger suffix with $r'_u$ than $r'_1$ has.
	
	Let 
	$W_1=\B'\big(r'_1[0..k-1]\big)$,
	$W_2=\B'\big(r'_2[0..j-1]\big)$,
	$X=\B'\big(r'_1[k..k+x-1]\big)$  (which is also $\B'\big(r'_u[i..|u|-1]\big)$),
	and $Y=\B'\big(r'_2[j..j+y-1]\big)$ (see \Cref{fig:BlowupMinusOneLowerBound}).
	The following must hold:
	\begin{enumerate}
		\item 
		$W_1+\frac{X}{\lambda^k}=\B'(r'_1)=-\frac{1}{\lambda^{k+x}}$. Hence, $\lambda^k W_1= -\frac{1}{\lambda^{x}}-X$ . 
		\item \label{eqn:negBlowUp2} $W_2+\frac{X}{\lambda^j}=\B'(r'_3)>-\frac{1}{\lambda^{j+x}}$. Hence, $\lambda^j W_2> -\frac{1}{\lambda^{x}}-X$, and after combining with the previous equation, $\lambda^j W_2> \lambda^k W_1$.
		\item \label{eqn:negBlowUp3}
		$W_2+\frac{Y}{\lambda^j}=\B'(r'_2)=-\frac{1}{\lambda^{j+y}} $. Hence, $\lambda^j W_2 + Y = -\frac{1}{\lambda^{y}}$
	\end{enumerate}
	Consider now the run
	$r'_4 = r'_1[0..k-1]\con r'_2[j..j+y-1]$, and combine \Cref{eqn:negBlowUp2} and \Cref{eqn:negBlowUp3} above to get that $\lambda^k W_1 + Y < -\frac{1}{\lambda^{y}}$. 
	But this leads to $ \B'(r'_4)=W_1 + \frac{Y}{\lambda^k} < -\frac{1}{\lambda^{k+y}}=-\frac{1}{\lambda^{|r'_4|}}$, and this means that there exists a word $w$ of length $k+y$ such that $\B'(w)< -\frac{1}{\lambda^{k+y}}$, contradicting the assumption that $\B'\equiv \Dot{\B}\equiv -\Tilde{\A}$.
	
	\noindent{\emph{Infinite words.}}
	
	For showing the lower bound for the state blow-up involved in multiplying an NDA by $(-1)$ w.r.t.\ infinite words, we add a new letter $\#$ to the alphabet, and correlate every finite word $u$ to an infinite word $u\con \#^\omega$. 
	The proof is similar, applying the following modifications:
	\begin{itemize}
		\item 
		The scheme presented in the proof of \Cref{lem:NFA_to_NDA} now constructs a $\lambda$-NDA $\Tilde{\A}$ over the alphabet $\Sigma\cup\{\#\}$, adding  a $0$-weighted transition from every state of $\Tilde{\A}$ to $q_{hole}$.
		The function $f$ that correlates between the runs of $\A$ and $\Tilde{\A}$ is still a bijection, but with a different co-domain, correlating every run $r$ of $\A$ on a finite word $u\in\Sigma^+$ to the run $f(r)$ of $\Tilde{\A}$ on the word $u\con\#^\omega$.
		\item
		With this scheme, we get that $\Dot{\B}(u\con\#^\omega)=-\frac{1}{\lambda^{|u|}}$ if $u\notin L(A)$ and $\Dot{\B}(u\con\#^\omega)=\frac{1}{\lambda^{|u|}}$ if $u\in L(A)$, hence replacing all referencing to $\B'(u)$ with referencing to $\B'(u\con\#^\omega)$.
		\item
		$R^-$ is defined with respect to words of the form $u\con \#^\omega$, namely $R^-=\{r\ST u\in\Sigma^+, r$ is a minimal run of $\B' \text{ on } u\con \#^\omega \text{ and } \B'(r)<0\}$.
		\item
		$R^-_p$ is a new set of all the maximal (finite) prefixes of the runs of $R^-$ without any transitions for the $\#$ letter, meaning $R^-_p = \{r[0..i-1]\ST r\in R^-, r(i-1)=(p,\sigma,q) \text{ for some } \sigma\in\Sigma, \text{ and }  r(i)=(q,\#,s)\}$.
		$\Hat{\delta}$ and $\doublehat{\delta}$ are defined with respect to $R_p^-$ instead of $R^-$.
		\item
		Defining $r'_u$, we consider a run $r'_t\in R^-$ that is a witness for $t'\in\doublehat{\delta}$, meaning there exists $i\in\Nat$ for which $r'_t(i)=t'$, and $r'_t(i+1)$ is a transition for the $\#$ letter. 
		Then $r'_u=r_u[0..|u|-2]\con t\con r'[i+1..\infty]=r_u[0..|u|-2]\con r'[i..\infty]$, is a run of $\B'$ on $u\con\#^\omega$.
		\item
		For choosing $r'_1$ that ``shares the maximal suffix'' with $r'_u$, we take $r'_1\in R^-$ such that for every $r'\in R^-$ and $x>0$, if $r'_u[i..\infty]$ is a suffix of $r'$ then it is also a suffix of $r'_1$.
		\item
		For the different runs and their parts, we set $X=\B'\big(r'_1[k..\infty]\big)$, $Y=\B'\big(r'_2[j..\infty]\big)$, $r'_3=r'_2[0..j-1]\con r'_u[i..\infty] $ and $r'_4=r'_1[0..k-1]\con r'_2[j..\infty]$.
		\qedhere
	\end{itemize}
\end{proof}

\subsection{Basic Subfamilies}\label{sec:SimpleSubfamilies}
Tidy NMDAs constitute a rich family that also contains some basic subfamilies that are still more expressive than integral NDAs.
Two such subfamilies are integral NMDAs in which the discount factors depend on the transition letter or on the elapsed time.

Notice that closure of tidy NMDAs under determinization and under algebraic operations is related to a specific choice function $\theta$, namely every class of $\theta$-NMDAs enjoys these closure properties (\Cref{thm:DetTidy} and \Cref{thm:closurealgebraicOps}).  Since the aforementioned subfamilies of tidy NMDAs also consist of $\theta$-NMDA classes, their closure under determinization and under algebraic operations follows. For example, the class of NMDAs that assigns a discount factor of $2$ to the letter `a' and of $3$ to the letter `b' enjoys these closure properties.

\subsubsection{Letter-Oriented Discount Factors}
Allowing each action (letter) to carry its own discount factor is a basic extension of discounted summation, used in various models, such as Markov decision processes \cite{LM16,WG19}.

A $\theta$-NMDA over an alphabet $\Sigma$ is \emph{letter oriented} if all transitions over the same alphabet letter share the same discount factor; that is, if $\theta: \Sigma^+ \to \Nat\setminus\{0,1\}$ coincides with a function $\Lambda: \Sigma \to \Nat\setminus\{0,1\}$, in the sense that for every finite word $u$ and letter $\sigma$, we have $\theta(u \sigma)=\Lambda(\sigma)$.
(See an example in \Cref{fig:LetterOrrientedExample}.)

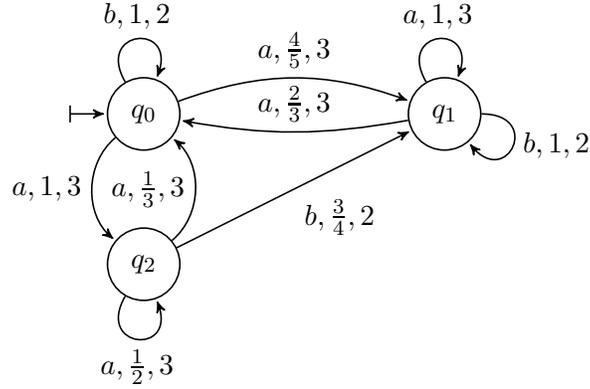
\begin{figure}
	\centering
	\begin{tikzpicture}[->,>=stealth',shorten >=1pt,auto,node distance=2cm, semithick, initial text=, every initial by arrow/.style={|->}]
		\node[initial, state] (q0) {$q_0$};
		
		\node[state, xshift=2cm] (q1) [right of=q0] {$q_1$};
		\node[state] (q2) [below of=q0] {$q_2$};
		\path
		(q0) edge [left, in=140, out = -140] node {$a,1,3$} (q2)
		(q2) edge [left, out=40, in = -40] node {$a,\frac{1}{3},3$} (q0)
		(q2) edge [below right] node {$b,\frac{3}{4},2$} (q1)
		(q1) edge [loop above, out=120,in=70,looseness=5] node {$a,1,3$} (q1)
		(q1) edge [loop below, right, out=0,in=-50,looseness=5] node {$b,1,2$} (q1)
		(q0) edge [loop above, out=120,in=70,looseness=5] node {$b,1,2$} (q0)
		(q0) edge [above, in=160, out = 20] node {$a,\frac{4}{5},3$} (q1)
		(q1) edge [above, out=-170, in = -10] node {$a,\frac{2}{3},3$} (q0)
		(q2) edge	[loop below, out=-120,in=-70,looseness=5] node {$a,\frac{1}{2},3$} (q2)
		;
	\end{tikzpicture}
	\caption[A letter-oriented discounted-sum automaton.]{\label{fig:LetterOrrientedExample}A letter-oriented discounted-sum automaton, for the discount factor function $\Lambda(a)=3$; $\Lambda(b)=2$.}
\end{figure}

Notice that every choice function $\theta$ for a letter-oriented $\theta$-NMDA can be defined via a simple transducer of a single state, having a self loop over every letter with its assigned discount factor.

We show that letter-oriented NMDAs, and in particular the NMDA $\A$ depicted in \Cref{fig:expressivness}, indeed add expressiveness over NDAs.
\begin{thm}\label{thm:expressivness}
	There exists a letter-oriented NMDA that no integral NDA is equivalent to, with respect to both finite and infinite words.
\end{thm}

	\begin{figure}
		\centering{}
		\begin{tikzpicture}[->,>=stealth',shorten >=1pt,auto,node distance=2cm, semithick, initial text=, every initial by arrow/.style={|->}]
			\node [midway] [above left of=q2, xshift=-1.5cm, yshift=-0.6cm] {$\A:$};
			\node[initial above,state] (q0) {$q_0$};
			\node[state] (q1) [right of=q0] {$q_1$}; 
			\node[state] (q2) [left of=q0] {$q_2$}; 
			\path 
			(q0) edge node [above] {$a,\frac{1}{2},2$} (q1)
			(q1) edge [loop right, in=-30,out=30,looseness=5, align=left] node {$a,-\frac{1}{2},2$\\$b,0,3$} (q1)
			(q0) edge node [above]{$b,\frac{1}{3},3$} (q2)
			(q2) edge [loop right, in=150,out=-150,looseness=5, align=left] node [left]{$a,0,2$\\$b,-\frac{2}{3},3$} (q2)
			;
		\end{tikzpicture}
		\caption[A letter-oriented discounted-sum automaton that no integral NDA is equivalent to.]{\label{fig:expressivness}A letter-oriented discounted-sum automaton, for the discount factor function $\Lambda(a)=2$; $\Lambda(b)=3$, that no integral NDA is equivalent to.}
	\end{figure}
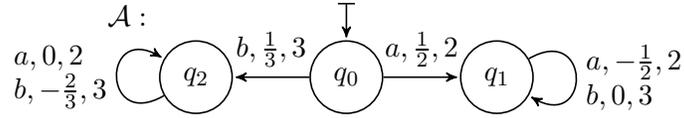
	
	\begin{proof}
	We show the result with respect to infinite words, and it also holds by  \Cref{lemma:finiteToInfinite} to finite words.	
	Consider the NMDA $\A$ depicted in \Cref{fig:expressivness}.
	Assume toward contradiction that there exists an integral NDA $\B'$ such that $\B'\equiv \A$. According to \cite{BH14}, there exists an integral deterministic NDA (integral DDA) $\B$ with transition function $\delta_\B$ and discount factor $\lambda$, such that $\B \equiv \B'\equiv\A$.
	
	Observe that for every $n\in\Nat$$\setminus\{0\}$, we have $\B(a^n b^\omega)=\A(a^n b^\omega)=\frac{1}{2^n}$.
	As $\B$ has finitely many states, there exists a state $q$ in $\B$ and $i,j\in\Nat\setminus\{0\}$ such that $\delta_\B(a^i)=\delta_\B(a^{i+j})=q$.
	Let $W_1=\B^{q}(a^j)$ and $W_2=\B^{q}(b^\omega)$.
	
	Observe that
	\begin{align}
		\frac{1}{2^i} 
		&= \B(a^i b^\omega) = \B(a^i) + \frac{W_2}{\lambda^i} \label{eqn:expressivness1}\\
		\frac{1}{2^{i+j}} 
		&= \B(a^{i+j} b^\omega) = \B(a^i) + \frac{W_1}{\lambda^i} + \frac{W_2}{\lambda^{i+j}} \label{eqn:expressivness2}\\
		\frac{1}{2^{i+2j}} 
		&= \B(a^{i+2j} b^\omega) = \B(a^i) + \frac{W_1}{\lambda^i} + \frac{W_1}{\lambda^{i+j}} + \frac{W_2}{\lambda^{i+2j}} \label{eqn:expressivness3}
	\end{align}
	Subtract \Cref{eqn:expressivness1} from \Cref{eqn:expressivness2}, and \Cref{eqn:expressivness2} from \Cref{eqn:expressivness3} to get
	\begin{align}
		\frac{1}{2^{i+j}} - \frac{1}{2^i} 
		&= \frac{W_1-W_2}{\lambda^i} + \frac{W_2}{\lambda^{i+j}}\label{eqn:expressivness4} \\
		\frac{1}{2^{i+2j}} - \frac{1}{2^{i+j}} 
		&= \frac{W_1-W_2}{\lambda^{i+j}} + \frac{W_2}{\lambda^{i+2j}}
		= \frac{1}{\lambda^j}\Big(\frac{W_1-W_2}{\lambda^{i}} + \frac{W_2}{\lambda^{i+j}}\Big)\label{eqn:expressivness5}
	\end{align}
	and combine \Cref{eqn:expressivness4,eqn:expressivness5} to get 
	$
	\frac{1}{2^j}\Big(\frac{1}{2^{i+j}} - \frac{1}{2^{i}}\Big) 
	= \frac{1}{2^{i+2j}} - \frac{1}{2^{i+j}} = \frac{1}{\lambda^j}\Big(\frac{1}{2^{i+j}} - \frac{1}{2^{i}}\Big)
	$, 
	which implies
	$\lambda=2$.
	
	Observe that for every $n\in\Nat$$\setminus\{0\}$, we have $\B(b^n a^\omega)=\A(b^n a^\omega)=\frac{1}{3^n}$. Symmetrically to the above, but with respect to `$b$' instead of `$a$' and `$3$' instead of `$2$', results in $\lambda=3$, leading to a contradiction.
\end{proof}

\subsubsection{Time-Oriented Discount Factors}
A $\theta$-NMDA over an alphabet $\Sigma$ is \emph{time oriented} if the discount factor on a transition is determined by the distance of the transition from an initial state; that is, if $\theta: \Sigma^+ \to \Nat\setminus\{0,1\}$ coincides with a function $\Lambda: \Nat\setminus\{0\} \to \Nat\setminus\{0,1\}$, in the sense that for every finite word $u$, we have $\theta(u)=\Lambda\big(|u|\big)$. 
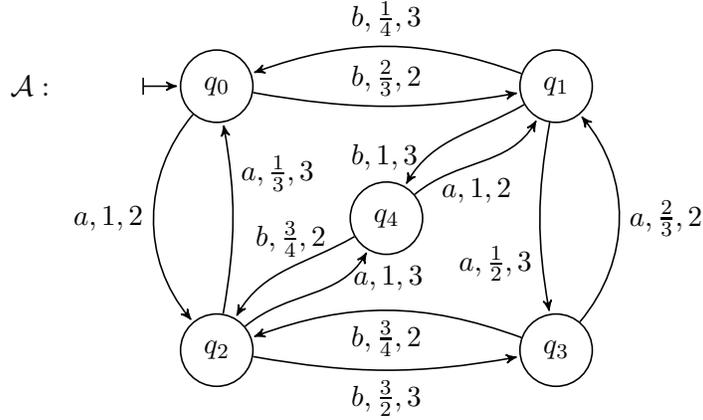
\begin{figure}
	\centering
	\begin{tikzpicture}[->,>=stealth',shorten >=1pt,auto,node distance=1.5cm, semithick, initial text=, every initial by arrow/.style={|->}]
		\node[state] (q4) {$q_4$};
		\node[initial, state] (q0) [above left= of q4, xshift=-0.5cm] {$q_0$};
		\node[state] (title) [left= of q0, draw=none] {$\A:$};
		\node[state] (q1) [above right= of q4, xshift=0.5cm] {$q_1$};
		\node[state] (q2) [below left= of q4, xshift=-0.5cm] {$q_2$};
		\node[state] (q3) [below right= of q4, xshift=0.5cm] {$q_3$};
		
		\path
		(q0) edge [left, in=130, out = -130] node {$a,1,2$} (q2)
		(q2) edge [right, out=80, in = -80, near end] node {$a,\frac{1}{3},3$} (q0)
		(q1) edge [left, in=100, out = -100, near end] node {$a,\frac{1}{2},3$} (q3)
		(q3) edge [right, out=50, in = -50] node {$a,\frac{2}{3},2$} (q1)
		(q0) edge [above, out=-10, in = -170] node {$b,\frac{2}{3},2$} (q1)
		(q1) edge [above, out=160, in = 20] node {$b,\frac{1}{4},3$} (q0)
		(q2) edge [below, out=-10, in = -170] node {$b,\frac{3}{2},3$} (q3)
		(q3) edge [below, out=160, in = 20] node {$b,\frac{3}{4},2$} (q2)
		(q1) edge [left, out=-150, in = 60, near end] node {$b,1,3$} (q4)
		(q4) edge [right, out=40, in = -120] node [very near start, yshift=-0.1cm] {$a,1,2$} (q1)
		(q4) edge [left, out=-150, in = 60] node [very near start, yshift=0.1cm] {$b,\frac{3}{4},2$} (q2)
		(q2) edge [right, out=40, in = -120, near end] node {$a,1,3$} (q4)
		;
	\end{tikzpicture}
	\caption[A time-oriented discounted-sum automaton.]{\label{fig:TimeOrientedExample}A time-oriented discounted-sum automaton $\A$.}
\end{figure}

For example, the NMDA $\A$ of \Cref{fig:TimeOrientedExample} is time-oriented, as all transitions taken at odd steps, in any run, have discount factor $2$, and those taken at even steps have discount factor $3$. The transducer $\T$ of \Cref{fig:NTDAF_Transducer} represents its choice function.

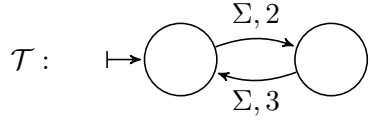
\begin{figure}
	\centering
	\begin{tikzpicture}[->,>=stealth',shorten >=1pt,auto,node distance=2cm, semithick, initial text=, every initial by arrow/.style={|->}]
		
		\node[initial, state] (qq0){};
		\node[state] (label) [left of=qq0, draw=none] {$\T:$};
		\node[state] (qq1) [right of=qq0] {};
		
		\path 
		(qq0) edge [above, out=20,in=160] node {$\Sigma,2$} (qq1)
		(qq1) edge	[below, out=-160,in=-20] node {$\Sigma,3$} (qq0)
		;
		
	\end{tikzpicture}
	\caption{\label{fig:NTDAF_Transducer}A transducer that represents the discount-factor choice function for the NMDA $\A$ of \Cref{fig:TimeOrientedExample}.}
\end{figure}

\begin{figure}
	\centering{}
	\begin{tikzpicture}[->,>=stealth',shorten >=1pt,auto,node distance=2cm, semithick, initial text=, every initial by arrow/.style={|->}]
		\node[initial,state] (q0) {$q_0$};
		\node[state] (label) [left of=q0, draw=none, xshift=0.5cm] {$\A:$};
		\node[state] (q1) [right of=q0] {$q_1$}; 
		\node[state] (q2) [right of=q1] {$q_2$};
		
		\node[state] (qq1) [left of=q0, xshift=-1cm] {};
		\node[initial, state] (qq0) [left of=qq1]{};
		\node[state] (label) [left of=qq0, draw=none] {$\T:$};
		\path
		(qq0) edge [above, out=20,in=160] node {$\Sigma,2$} (qq1)
		(qq1) edge	[below, out=-160,in=-20] node {$\Sigma,3$} (qq0)
		;
		
		\path 
		(q0) edge node [above, align=center] {$a,\frac{1}{6},2$\\$b,0,2$} (q1)
		(q1) edge [in=150,out=30] node [above, align=center,yshift=-0.1cm] {$a,0,3$\\[-3pt]$b,0,3$} (q2)
		(q2) edge [in=-30,out=-150] node [below, align=center] {$a,-\frac{5}{6},2$\\$b,0,2$} (q1)
		;
	\end{tikzpicture}
	
	\caption{\label{fig:expressivness2}A time-oriented NMDA that no integral NDA is equivalent to, and a transducer that defines its choice function.}
	
\end{figure}
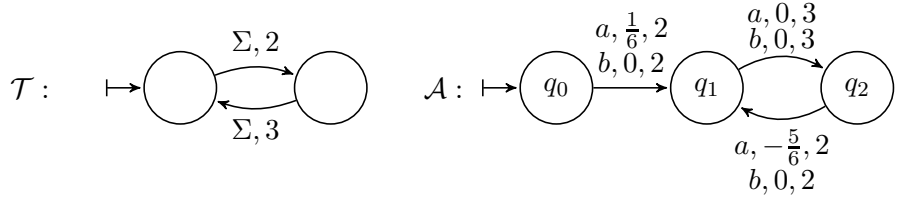

Time-oriented NMDAs extend the expressiveness of NDAs, as proved for the time-oriented NMDA depicted in \Cref{fig:expressivness2}.
\begin{thm}\label{thm:expressivnessTimeOriented}
	There exists a time-oriented NMDA that no integral NDA is equivalent to, with respect to both finite and infinite words.
\end{thm}

\begin{proof}
	We show the result with respect to infinite words, and it also holds by  \Cref{lemma:finiteToInfinite} to finite words.	
	Let $\A$ be the time-oriented NMDA depicted in \Cref{fig:expressivness2}.
	Observe that
	$\A(a^n b^\omega)= \frac{1}{6^{\lceil\frac{n}{2}\rceil}}$.
	Analogously to the proof of \Cref{thm:expressivness}, but with respect to ``$\sqrt{6}$'' instead of ``2'', we have that the discount factor of an equivalent DDA, if such exists, is $\lambda=\sqrt{6}$, hence no integral NDA can be equivalent to $\A$.
\end{proof}

\section{Tidy NMDAs -- Decision Problems}\label{sec:tidyDecisionProblems}
We show that all of the decision problems of tidy NMDAs are in the same complexity classes as the corresponding problems for discounted-sum automata with a single integral discount factor. That is, the nonemptiness problem is in PTIME, and the exact-value, universality, equivalence, and containment problems are in PSPACE (see \Cref{tbl:decisionProblems}).
In the equivalence and containment problems, we consider $\theta$-NMDAs with the same choice function $\theta$.
In addition, the problem of checking whether a given NMDA is tidy, as well as whether it is a $\theta$-NMDA, for a given choice function $\theta$, is decidable in PTIME.
The complexities are w.r.t.\ the automata size (as defined in \Cref{sec:Definitions}), and when considering a threshold $\nu$, w.r.t.\ its binary representation.

\subsection{Tidiness}\label{sec:Tidiness}
Given an NMDA $\A$, one can check in PTIME whether $\A$ is tidy.
The algorithm follows by solving a reachability problem in a Cartesian product of $\A$ with itself, to verify that for every word, the last discount factors are identical in all runs.

\begin{thm} \label{thm:GeneralCheckTidy}
	Checking if a given NMDA $\A$ is tidy is decidable in time $O\big(|\A|^2\big)$.
\end{thm}
\begin{proof}
	Consider an input NMDA $\A = \tuple{\Sigma, Q, \iota, \delta, \gamma, \rho}$.
	Observe that $\A$ is tidy iff there does not exist a finite word $u\in\Sigma^+$ of length $n=|u|$ and runs $r_1$ and $r_2$ of $\A$ on $u$, such that $\rho(r_1(n-1))\neq \rho(r_2(n-1))$.
	Intuitively, we construct the Cartesian product of $\A$ with itself, associating the weight of every transition in the product to the difference of the two discount factors of the transitions causing it.
	The problem then reduces to reachability in this product
	automaton of a transition with weight different from $0$.  
	
	Formally, construct a weighted automaton $P=\tuple{\Sigma, Q\times Q, \iota\times\iota, \delta', \gamma'}$ such that
	\begin{itemize}
		\item $\delta'= \Big\{
		\big((s_0,s_1),\sigma,(t_0,t_1)\big) \ST \sigma\in\Sigma \text{ and } (s_0,\sigma,t_0),(s_1,\sigma,t_1)\in \delta
		\Big\}$.
		\item $\gamma'\big((s_0,s_1),\sigma,(t_0,t_1)\big)= \rho(s_0,\sigma,t_0)-\rho(s_1,\sigma,t_1)$.
	\end{itemize}
	Every run in $P$ for a finite word $u$ corresponds to two runs in $\A$ for the same word $u$.
	A non-zero weighted transition in $P$ corresponds to two transitions in $\A$ for the same letter, but with different discount factors.
	Hence, $\A$ is tidy if and only if no run in $P$ takes a non-zero weighted transition.
	
	The graph underlying $P$ can be constructed in time quadratic in the size of $\A$, and the reachability check on it can be performed in time linear in the size of this graph.
\end{proof}

Given also a transducer $\T$, one can check in polynomial time whether $\A$ is a $\T$-NMDA.

\begin{thm}\label{thm:decideTidySpecific}
	Checking if a given NMDA $\A$ is a $\T$-NMDA, for a given transducer $\T$, is decidable in time $O\big(|\A|\cdot |\T|\big)$.
\end{thm}
\begin{proof}
	We show the procedure.
	Let $\A = \tuple{\Sigma, Q_\A, \iota, \delta_\A, \gamma, \rho_\A}$ be the input NMDA and $\T=\tuple{Q_\T,\Sigma,q_0,\delta_\T,\rho_\T}$ the input transducer.
	
	We construct a nondeterministic weighted automaton $\A'$ that resembles $\A$ and a deterministic weighted automaton $\T'$ that resembles $\T$, as follows. 
	$\A'=\tuple{\Sigma, Q_\A, \iota, \delta_{\A}, \rho_{\A}}$ is derived from $\A$ by taking the same basic structure of states, initial states and transition function, and having the discount factors of $\A$ as its weight function. $\T'=\tuple{\Sigma, Q_\T, q_0, \delta_\T, \rho_\T}$
	is derived from  $\T$, by having the same structure as $\T$ and having the output function of $\T$ as the weight function of $\T'$.
	
	Then, we construct the product automaton $\B = \A'\times\T'$, in which the weight on each transition is the weight of the corresponding transition in $\A'$ minus the weight of the corresponding transition in $\T'$. 
	
	It is only left to check whether or not all the weights on the reachable transitions of $\B$ are zero. Indeed, $\A$ is a $\T$-NMDA iff all its reachable discount factors, which are the weights in $\A'$, correspond to the outputs of $\T$, which are the weights in $\T'$.
\end{proof}

\subsection{Nonemptiness}\label{sec:Nonemptiness}

\begin{table}
	\begin{center}
		\begin{tabular}{|c||c|c|}
			\hline
			& Finite words & Infinite words \\ [0.5ex]
			\hline \hline
			Non-emptiness ($<$) &  PTIME {\footnotesize (\Cref{thm:nonemptinessFinite})}& \multirow{2}{*}{PTIME {\footnotesize (\Cref{thm:nonemptinessInf})}} \\ [0.2ex]
			\cline{1-2}
			Non-emptiness ($\leq$) & PTIME {\footnotesize (\Cref{thm:nonemptinessFiniteNonStrict})}& \\ [0.2ex]
			\hline
			Containment ($>$) & PSPACE-complete &  PSPACE \footnotesize{(\Cref{thm:containmentInfiniteStrict})}\\ [0.2ex]
			\cline{1-1}\cline{3-3}
			Containment ($\geq$) & {\footnotesize (\Cref{thm:containmentFinite})} &  PSPACE-complete \footnotesize{(\Cref{thm:containmentInfinite})}\\ [0.2ex]
			\hline
			Equivalence &  \multicolumn{2}{c|}{PSPACE-complete \footnotesize (\Cref{thm:equivalence})} \\ [0.2ex]
			\hline
			Universality ($<$) & PSPACE-complete &  PSPACE{\footnotesize(\Cref{thm:universality})}\\ [0.2ex]
			\cline{1-1}\cline{3-3}
			Universality ($\leq$) & \footnotesize{(\Cref{thm:universality})} &  PSPACE-complete \footnotesize{(\Cref{thm:universality})}\\ [0.2ex]
			\hline
			\multirow{2}{*}{Exact-value} & PSPACE-complete &  \multirow{2}{*}{PSPACE {\footnotesize (\Cref{thm:exact})}}\\ [0.2ex]
			& {\footnotesize (\Cref{thm:exact})} &  \\ [0.2ex]
			\hline
		\end{tabular}
	\end{center}
	\caption{The complexities of the decision problems of tidy NMDAs.} \label{tbl:decisionProblems}
\end{table}

Considering nonemptiness with respect to infinite words, for both strict and non-strict inequalities there is a simple reduction to one-player discounted-payoff games (\Cref{def:DPG}) that also applies to arbitrary NMDAs (which are not necessarily tidy, or even integral), showing that these problems are in PTIME.
This result can also be generalized to strict nonemptiness of arbitrary NMDAs w.r.t.\ finite words.
The non-strict problem w.r.t.\ finite words is solved differently, and applies to integral NMDAs (which are not necessarily tidy).

We follow the definition of discounted-payoff games with multiple discount factors (DPGs) given in \cite{Andersson06}:
\begin{defi}[\cite{Andersson06}]\label{def:DPG} A \emph{one-player discounted payoff game} (\emph{one-player DPG}) is a 4-tuple $\tuple{V,E,\gamma,\rho}$ such that,
\begin{itemize}
	\item $E=\set{e_1,...,e_m}$ is a set of directed edges between vertices in $V=\set{v_1,...,v_n}$. 
	DPGs allow multiple edges between the same ordered pair of source and destination vertices.
	\item $\gamma:E\to\Rat$ is a weight function.
	\item $\rho:E\to\set{x\in\Rat : 0<x<1}$ is a discount function.
	\item An \emph{infinite play} $\pi$ from some vertex $v\in V$ is an infinite sequence of edges, $e_{i_0}e_{i_1}...$, such that the head of $e_{i_0}$ is $v$, and the tail of every edge is the head of the following edge.
	\item The \emph{value} of an infinite play $\pi$ is defined by 
	$\mu(\pi)=\gamma(e_{i_0})+\rho(e_{i_0})\big(\gamma(e_{i_1})+\rho(e_{i_1})(...)\big)=
	\Sigma_{j=0}^{\infty}{\Big(\gamma(e_{i_j})\prod_{k\in[0..j{-}1]}{\rho(e_{i_k}})\Big)}$.
	\item A \emph{solution} to a MIN (respectively, MAX) one-player DPG is a function $s:V\to\Rat$, such that for every $v\in V$, $s(v)$ is a value of an infinite play from $v$, and for every infinite play $\pi$ from $v$, $s(v)\leq\mu(\pi)$ (respectively, $s(v)\geq\mu(\pi)$).
\end{itemize}
\end{defi}
\noindent 
Section 3.1 of \cite{Andersson06} presents a polynomial-time algorithm for finding a solution to MIN- and MAX-one-player DPGs.
Observe that our definition of the value of a walk in an NMDA is identical to the definition of the value of a play in \Cref{def:DPG}. 
Hence, we can transform a given NMDA to a one-player DPG by using the same states, transitions, weights and discount factors as the ones in the NMDA, and omit the letters on the transitions. Doing so, we can solve nonemptiness of NMDAs using the algorithm of solving MIN-one-player DPGs.

\begin{thm} \label{thm:nonemptinessInf}
	The nonemptiness problem of NMDAs w.r.t.\ infinite words is in PTIME.
\end{thm}
\begin{proof}
	Let $\A=\tuple{\Sigma, Q, \iota, \delta, \gamma, \rho}$ be an NMDA and $\nu\in\Rat$ a threshold.
	We will construct a one-player DPG $G=\tuple{Q,E,\gamma_G,\rho_G}$
	such that every infinite walk $\psi$ of $\A$ will have a corresponding infinite play $\pi$ of $G$, such that $\A(\psi)=\mu(\pi)$.
	
	For every transition $t=(q,\sigma,p)\in\delta$ we add a corresponding edge $(q,p)$ to $E$ with weight and discount factor of $\gamma_G(q,p)=\gamma(t)$ and $\rho_G(q,p)=1/\rho(t)$, respectively.
	Let $f$ be the function that matches a transition in $\A$ to the corresponding edge in $G$. 
	We extend $f$ to be a bijection between the set of walks of $\A$ and the set of plays of $G$. Observe that by construction, for every walk $\psi$, we have $\A(\psi)=\mu\big(f(\psi)\big)$, and for every play $\pi$, we have $\mu(\pi)=\A\big(f^{-1}(\pi)\big)$.
	Recall that the value of $\A$ on a word is the infimum value of its runs on the word, implying that the infimum value of $\A$ on all words is equal to the infimum value of all plays in $G$ that start from vertices that correspond to initial states of $\A$.
	
	Section 3.1 of \cite{Andersson06} presents a polynomial-time algorithm for finding the minimal value of a play starting from every vertex $v\in Q$. 
	All left to do is to iterate all the vertices that correspond to initial states in $\iota$, and check if the minimal value of a play from any of them is lower (or lower-or-equal for the non-strict case) than $\nu$. 
\end{proof}

For nonemptiness with respect to finite words, we cannot directly use the aforementioned DPG solution, as it stands for infinite plays. 
However, for nonemptiness with respect to strict inequality, we can reduce the finite-words case to the infinite-words case: If there exists an infinite word $w$ such that $\A(w)$ is strictly smaller than the threshold, the distance between them cannot be compensated in the infinity, implying the existence of a finite prefix that also has a value smaller than the threshold; As for the other direction, we add to every state a $0$-weight self loop, causing a small-valued finite word to also imply a small-valued infinite word.

\begin{thm} \label{thm:nonemptinessFinite}
	The nonemptiness problem of NMDAs w.r.t.\ finite words and strict inequality is in PTIME.
\end{thm}
\begin{proof}
	Let $\A=\tuple{\Sigma, Q, \iota, \delta, \gamma, \rho}$ be an NMDA and $\nu\in\Rat$ a threshold.
	We will construct in polynomial time an NMDA $\A'=\tuple{\Sigma, Q\cup \iota\times\{1\}\cup \{q_\infty\}, \iota\times\{1\}, \delta\cup \delta'\cup\delta'', \gamma\cup\gamma'\cup\gamma'', \rho\cup\rho'\cup\rho''}$, such that $\A'$ is empty($<$) with respect to infinite words if and only if $\A$ is empty($<$) with respect to finite words, getting from \Cref{thm:nonemptinessInf} the required result.
	
	The construction duplicates all the initial states of $\A$ and adds a new state $q_\infty$. 
	The new transitions are:
	\begin{itemize}
		\item 
		$\delta'= \big\{\big((q,1),\sigma,q'\big) \ST q\in \iota, \sigma\in\Sigma, (q,\sigma,q')\in\delta\big\}$;
		
		$\gamma':\delta'\to \Rat$ such that $\gamma'\big((q,1),\sigma,q'\big)=\gamma(q,\sigma,q')$;
		
		$\rho':\delta'\to \Nat\setminus\{0,1\}$ such that $\rho'\big((q,1),\sigma,q'\big)=\rho(q,\sigma,q')$.
		\item 
		$\delta''=\{(q,\tau,q_\infty) \ST q\in Q\} \cup \{(q_\infty,\sigma,q_\infty) \ST \sigma\in \Sigma\} $ for some letter $\tau\in\Sigma$;
		
		$\gamma'':\delta''\to \Rat$ such that $\gamma''\equiv 0$;
		
		$\rho'':\delta''\to \Nat\setminus\{0,1\}$ for any arbitrary discount factors.
	\end{itemize}
	\noindent 
	Observe that for every finite word $u\in\Sigma^+$ we have that $\A'(u\con \tau^\omega)\leq\A(u)$, since for every run of $A$ on $u$ there is an equivalent run of $A'$ on $u$ that has the same value.
	
	If $\A$ is not empty($<$) w.r.t.\ finite words, there exists $u\in\Sigma^+$ such that $\A(u)<\nu$. Hence $\A'(u\con \tau^\omega)\leq\A(u)<\nu$. Concluding that $\A'$ is not empty($<$) w.r.t.\ infinite words. 
	
	For the other direction, if $\A'$ is not empty($<$) w.r.t.\ infinite words, 
	there exists $w\in\Sigma^\omega$ such that $\A'(w)<\nu$. Let $r$ be the run of $\A'$ on $w$ that entails the minimum value.
	Assume $r$ contains some transitions from $\delta''$. 
	Let $r'$ be the maximal prefix run of $r$ that contains only transitions from $\delta$ and $\delta'$.
	Since all the transitions in $\delta''$ are targeted in $q_\infty$ and have a weight of 0, we get that $\A'(r')=\A'(r)<0$. 
	By changing the first transition of $r'$ from $\big((q,1),\sigma,q'\big)$ to $(q,\sigma,q')$ we get a run of $A$ on a finite prefix of $w$ with the same value of $\A'$ on $r$, which is a value strictly less than $\nu$.
	Meaning that there exists $v\in\Sigma^+$ such that $\A(v)<\nu$, which is our claim.
	Otherwise, $r$ contains only transitions from $\delta$ and $\delta'$. changing its first transition $\big((q,1),\sigma,q'\big)$ to $(q,\sigma,q')$ results in a run of $A$ on $w$ with the same value strictly less than $\nu$.
	
	We will now show that if the value of $\A$ on some infinite word $w$ is less than $\nu$ then there exists a prefix of $w$ for which the value of $\A$ is also less than $\nu$. Denote $\epsilon = \nu-\A(w)$. Let $W$ be the maximal absolute value of $\A$ on any infinite word, and $\lambda$ the minimal discount factor in $\A$.
	
	Observe that there exists $n_\epsilon\in\Nat$ such that $\frac{W}{\lambda^{n_\epsilon}}<\epsilon$ and consider the run $r_{n_\epsilon}=r[0..n_\epsilon-1]$ of $\A$ on the finite word $u=w[0..n_\epsilon -1]$.
	We will show that after reaching $\delta(r_{n_\epsilon})$, if $\A(r_{n_\epsilon})$ is not smaller than $\nu$, then the weight of the suffix $\A(r[n_\epsilon..\infty])$ reduced by the accumulated discount factor $\rho(r_{n_\epsilon})$ will be too small to compensate, resulting in $\A(r)\geq \nu$.
	
	Observe that $| \A^{\delta(u)}(w[n_\epsilon..\infty]) | \leq W < \epsilon\cdot \lambda^{n_\epsilon}$
	and $\rho(r_{n_\epsilon})\geq \lambda^{n_\epsilon}$, resulting in $\frac{1}{\rho(r_{n_\epsilon})}\leq\frac{1}{\lambda^{n_\epsilon}}$ and 
	\scalebox{1.3}{$\frac{| \A^{\delta(u)}(w[n_\epsilon..\infty]) |}{\rho(r_{n_\epsilon})}<\epsilon$}.
	
	And finally,
	\begin{align*}
		\nu-\epsilon = \A(w)=\A(r) &=\A(r_n)+\frac{\A^{\delta(u)}\big(w[n_\epsilon..\infty]\big)}{\rho(r_n)} \\
		&\geq \A(r_n) - \frac{\big| \A^{\delta(u)}(w[n_\epsilon..\infty]) \big|}{\rho(r_{n_\epsilon})} >\A(r_n)- \epsilon \geq \A(u)- \epsilon
	\end{align*} 
	
	Meaning that $\nu>\A(u)$ and $\A$ is not empty($<$) with respect to finite words.
\end{proof}
\noindent 
For nonemptiness with respect to finite words and non-strict inequality, we cannot use the construction used in the proof of \Cref{thm:nonemptinessFinite}, since its final part is inadequate: It is possible to have an infinite word with value that equals the threshold, while every finite prefix of it has a value strictly bigger than the threshold.
Yet, when considering \emph{integral} NMDAs, we can use a different approach for resolving the problem, applying linear programming to calculate the minimal value of a finite run ending in every state.

\begin{thm} \label{thm:nonemptinessFiniteNonStrict}
	The nonemptiness problem of integral NMDAs w.r.t.\ finite words and non-strict inequality is in PTIME.
\end{thm}
\begin{proof}
	Consider an integral NMDA $\A=\tuple{\Sigma,Q,\iota,\delta,\gamma,\rho}$ and a threshold $\nu$.
	For every finite run $r$ of $\A$, we define its normalized difference from $\nu$ as the accumulated discount factor multiplied by the difference, meaning $\Delta(r)=\rho(r)\big(\A(r)-\nu\big)$.
	For every state $q\in Q$, we define its minimal normalized difference from $\nu$ as the minimal normalized difference among all finite runs that end in $q$, meaning, $\Delta(q)=\inf\{\Delta(r)\St \delta(r)=q\}=\inf(D_q)$.
	
	$\A$ is not empty w.r.t.\ finite words and non-strict inequality iff there exists a run $r$ such that $\Delta(r)\leq 0$.
	We will show that for every state $q\in Q$ such that $\Delta(q)\leq 0$, there exists a finite run $r$ of $\A$ ending in $q$ such that $\Delta(r)\leq 0$, and combine it with the trivial opposite direction to conclude that $\A$ is not empty iff there exists $q\in Q$ such that $\Delta(q)\leq 0$.
	Consider a state $q\in Q$, 
	\begin{itemize}
		\item 
		If $\Delta(q)=-\infty$, then by the definition of $\Delta(q)$, for every $x<0$ there exists a run $r$ ending in $q$ such that $\Delta(r)<x$.
		\item
		If $\Delta(q)=x\in\Rat$, then for every $\epsilon>0$ there exists a run $r_\epsilon$ ending in $q$ such that $\epsilon>\Delta(r_\epsilon)-x\geq 0$.
		Since we are dealing with integral discount factors, every normalized difference of a run is of the form $\frac{k}{d}$, where $k\in\Nat$ and $d$ is the common denominator of the weights in $\gamma$ and $\nu$.
		We will show that the infimum of the set $D_q$ is its minimum, since every element of $D_q$ can have only discrete values.
		
		Let $k_x\in\Nat$ be the minimal integer such that $\frac{k_x}{d}\geq x$, meaning $k_x=\left \lceil{x\cdot d}\right \rceil $, and observe that for every run $r$ ending in $q$ we have $\Delta(r)\geq\frac{k_x}{d}$, leading to $\Delta(r)-x\geq\frac{k_x}{d}-x$.
		Since this difference needs to be arbitrarily small, we get that $\frac{k_x}{d}-x=0$.
		For every run $r$ ending in $q$ we have that $\Delta(r)-x$ is $0$ or at least $\frac{1}{d}$.
		And since this difference needs to be arbitrarily small, it must be $0$ for some of those runs. 
		Hence, there exists a run $r$ ending in $q$ such that $\Delta(r)=x$.
	\end{itemize}
	\noindent 
	We will now show a linear program that calculates the value of $\Delta(q)$ for every $q\in Q$, or determines that there exists some $q\in Q$ such that $\Delta(q)<0$.
	For simplicity, we assume that all the states in $\A$ are reachable (since otherwise, one can create in polynomial time an equivalent integral NMDA for which all states are reachable).
	Let $Q_{in}$ be the set of all states that have an incoming transition, and $n$ its size, meaning $Q_{in}=\{q\in Q\St \exists(p,\sigma,q)\in \delta\}=\{q_1,\cdots,q_n\}$.
	Our linear program is over the variables $x_1,x_2,\cdots,x_n$, such that if there exists a feasible solution to the program, meaning a solution that satisfies all the constraints, then $\tuple{\Delta(q_1),\Delta(q_2),\ldots,\Delta(q_n)}$ is its maximal solution, and otherwise there exists a state $q$ such that $\Delta(q)<0$.
	For the first case, after finding the minimal normalized difference from $\nu$ for every state in $Q_{in}$, we can check if any of them equals to $0$, and for the other case we can immediately conclude that $\A$ is not empty.
	
	For defining the linear program, we first make the following observations. For every $t=(q_i,\sigma,q_j)\in \delta$ s.t.\ $q_i\in\iota$, we have
	$\Delta(t)=\rho(t)\cdot \big(\gamma(t)-\nu\big)$,
	and for every run $r$ of length $|r|=m>1$ we have 
	\begin{align*}
		\Delta(r) &= \rho(r)\cdot \big(\A(r)-\nu\big)\\
		&=\rho\big(r[0..m-2]\big)\rho\big(r(m-1)\big)\cdot \Big(\A\big(r[0..m-2]\big)+\frac{\gamma\big(r(m-1)\big)}{\rho\big(r[0..m-2]\big)}-\nu\Big) \\
		&=\rho\big(r(m-1)\big)\cdot \Big(\Delta\big(r[0..m-2]\big)+\gamma\big(r(m-1)\big)\Big)
	\end{align*}
	Hence, $\tuple{x_1,x_2,\ldots,x_n}=\tuple{\Delta(q_1),\Delta(q_2),\ldots,\Delta(q_n)}$ must satisfy the following system of equations:
	\begin{align}		
		x_j &\leq \rho(t)\cdot \big(\gamma(t)-\nu\big) \quad \text{ }\text{ for every } t=(q_i,\sigma,q_j)\in \delta \text{ s.t.\ } q_i\in\iota \label{eqn:linearProg1}\\
		x_j &\leq \rho(t)\cdot \big(\gamma(t)+x_i\big) \quad  \text{ for every } t=(q_i,\sigma,q_j)\in \delta \text{ s.t.\ } q_i\in Q_{in}\label{eqn:linearProg2}
	\end{align}
	
	These equations have a single maximal solution $\tuple{x^*_1,\cdots, x^*_n}$ such that for any solution $\tuple{a_1,\cdots, a_n}$ and $i\in[1..n]$, we have $x^*_i\geq a_i$ .
	To see that $\tuple{\Delta(q_1),\ldots,\Delta(q_n)}$ is indeed the unique maximal solution, if such exists, consider a solution $\tuple{a_1,\cdots, a_n}$, a state $q_i\in Q_{in}$ and a run $r$ such that $\delta(r)=q_i$ and $\Delta(r)=\Delta(q_i)$. 
	For every $j\in[0..|r|{-}1]$, let $q_{i_j}$ be the target state after the $j$-sized prefix of $r$, meaning $q_{i_j}=\delta\big(r[0..j]\big)$. We will show by induction on $j$ that $a_{i_j}\leq \Delta(r[0..j])$ to conclude that $a_i=a_{i_{|r|-1}}\leq \Delta(r[0..|r|{-}1])=\Delta(r)=\Delta(q_i)$:
	\begin{itemize}
		\item For the base case, we have $a_{i_0}\leq \rho\big(r(0)\big)\big(\gamma(r(0))-\nu\big)=\Delta\big(r(0)\big)$.
		\item For the induction step,
		\begin{align*}
			a_{i_j}&\leq
			\rho\big(r(j)\big)\cdot \Big(\gamma\big(r(j)\big)+a_{i_{j-1}} \Big)\\
			&\leq
			\rho\big(r(j)\big)\cdot \Big(\gamma\big(r(j)\big)+\Delta\big(r[0..j-1]\big)\Big) = \Delta\big(r[0..j]\big)
		\end{align*}
	\end{itemize}
	\noindent 
	The implicit constraint of non-negative values for the variables of the linear program, meaning $x_i\geq 0 $ for every $i\in[1..n]$, handles the case of a possible divergence to $-\infty$.
	With these constraints, if there exists $q\in Q$ such that $\Delta(q)<0$, then the linear program has no feasible solution, and this case will be detected by the algorithm that solves the linear program.
	
	Meaning that the problem can be stated as the linear program:
	maximize $\sum_{i=0}^n{x_i}$ subject to \Cref{eqn:linearProg1,eqn:linearProg2} and $x_i\geq 0$ for every $i\in[1..n]$.
\end{proof}

Notice that when considering deterministic automata, complementation, namely multiplication by $(-1)$ is straightforward, and thus universality and nonemptiness are equally easy. Furthermore, containment and equivalence between deterministic automata $\A$ and $\B$ can also be reduced to nonemptiness, by considering $\A-\B$.

\begin{thm}\label{thm:decisionDMDAs}
	For every choice function $\theta$, the containment, equivalence and universality problems of $\theta$-DMDAs are in PTIME for both finite and infinite words.
\end{thm}
\begin{proof}
	We show that the containment problems can be reduced to the nonemptiness problems when swapping the strictness of the problem (``$>$'' becomes ``$\leq$'' and ``$\geq$'' becomes ``$<$''). 
	Consider $\theta$-DMDAs $\A$ and $\B$.
	By the proof of \Cref{thm:closurealgebraicOps}, we can construct an integral DMDA $\C\equiv \A-\B$ in linear time.
	Observe that for all words $w$, $\A(w)>\B(w)$ $\Leftrightarrow$ for all words $w$, $\C(w)>0$ $\Leftrightarrow$ there is no word $w$ s.t $\C(w)\leq 0$.	Meaning that $\A$ is contained($>$) in $\B$ iff $\C$ is empty($\leq$) with respect to the threshold $0$.
	Similarly, $\A$ is contained($\geq$) in $\B$ iff $\C$ is empty($<$) with respect to the threshold $0$.
	
	Equivalence is solved by checking containment($\geq$) in both directions, and the universality problems are special cases of the containment problems, by setting $\B$ to be the input DMDA and $\A$ to be a constant DMDA that gets the value of the input threshold on every word.
\end{proof}

Observe that since \Cref{thm:nonemptinessInf,thm:nonemptinessFinite} are valid for general NMDAs, having discount factors that are not necessarily integral (as opposed to \Cref{thm:nonemptinessFiniteNonStrict}, which requires the NMDAs to be integral), the results of \Cref{thm:decisionDMDAs} are also valid for general DMDAs (with the same choice function), considering all the problems with respect to infinite words, and the problems of equivalence, containment($\geq$), and universality($\leq$) w.r.t.\ finite words.

\subsection{Exact-Value, Universality, Equivalence, and Containment}\label{sec:PSPaceProblems}

We turn to the PSPACE-complete problems, to which we first provide hardness proofs, and then, in \Cref{sec:PspaceAlgorithms}, PSPACE algorithms.

\subsubsection{PSPACE-hardness}\label{sec:PspaceHardness}
Our hardness proofs are by reductions from the universality problem of NFAs, which is known to be PSPACE-complete \cite{MS72}. Notice that the provided hardness results already stand for integral NDAs, not only to tidy NMDAs.

PSPACE-hardness of the containment problem for NDAs with respect to infinite words and non-strict inequalities is shown in \cite{ComparatorAutomataInQuantitativeVerification}. 
We provide below more general hardness results, considering the universality, equivalence, and exact-value problems.
Notice that PSPACE-hardness of universality w.r.t.\ finite words directly follows from \cite{ComparatorAutomataInQuantitativeVerification} and \Cref{lem:NFA_to_NDA}. Yet, we include this case below, using slightly modified reduction, which also serves to show hardness of other decision problems.
\begin{lem} \label{lemma:lowerBoundContainmentAndEquivalence}
	The equivalence and universality($\leq$) problems of integral NDAs w.r.t.\ finite words are PSPACE-hard.
\end{lem}
\begin{proof}
	Given an NFA $\A=\tuple{\Sigma,Q,\iota,\delta,F}$, we construct in polynomial time an NDA $\Tilde{\A}=\tuple{\Sigma,Q\cup\{p_0,q_{hole}\},\{p_0\},\delta',\gamma'}$ with discount factor $2$, such that $\Tilde{\A}$ never gets a negative value, and $\A$ is universal if and only if $\Tilde{\A}$ is equivalent to a $0$ NDA, namely to an NDA that gets a value of $0$ on all finite words. For simplicity, we ignore the empty word, whose acceptance is easy to check in $\A$.
	
	The construction is similar to the one presented in the proof of \Cref{lem:NFA_to_NDA}, with the following modifications to the weights: 
	\begin{equation*}
		\gamma'(p,\sigma,q) = 
		\begin{cases}
			0 & p \in \{p_0\}\cup F,q\in F
			\\
			\frac{1}{2} & p\in \{p_0\}\cup F, q\notin F
			\\
			-\frac{1}{2} & p\notin \{p_0\}\cup F, q\notin F
			\\
			-1 & p\notin \{p_0\}\cup F, q\in F
		\end{cases}
	\end{equation*}

	An example of the construction is given in \Cref{fig:lowerBoundContainmentAndEquivalence}.
	\begin{figure}
		\centering
		\begin{tikzpicture}[->,>=stealth',shorten >=1pt,auto,node distance=1.5cm, semithick, initial text=, every initial by arrow/.style={|->}]
			\node[initial, state, inner sep=0.1cm,minimum size=0.4cm] (q0) {$q_0$};
			\node[state, accepting, inner sep=0.1cm,minimum size=0.4cm] (q1) [right of=q0] {$q_1$}; 
			\node[state] (dummy) [right of =q1, draw=none, xshift=-0.4cm] {\Huge $\Rightarrow$};
			\node[initial, state, inner sep=0.1cm, minimum size=0.4cm] (p01) [right of=dummy] {$p_0$};
			\node[state, inner sep=0.0001cm,minimum size=0.4cm] (phole) [right of=p01] {$q_{hole}$};
			\node[state, inner sep=0.1cm,minimum size=0.4cm] (p0) [right of=phole, xshift=0.2cm] {$q_0$};
			\node[state, inner sep=0.1cm,minimum size=0.4cm] (p1) [right of=p0] {$q_1$};  
			\path 
			(q0) edge [out=30, in=150] node [above] {$a$} (q1)
			(q1) edge [out=-150, in=-30] node [below] {$a$} (q0)
			(p0) edge [out=30, in=150] node [above, near start] {$a,-1$} (p1)
			(p1) edge [out=-150, in=-30] node [below, near start] {$a,\frac{1}{2}$} (p0)
			(p01) edge [out=90, in=90] node [left, near start, xshift=-0.25cm] {$a,0$} (p1)
			(p0) edge  node [above] {$a,-\frac{1}{2}$} (phole)
			(p01) edge  node [below] {$a,\frac{1}{2}$} (phole)
			(phole) edge [out=110, in=150,loop, looseness=5] node [above,near end, yshift=0.1cm] {$a,-\frac{1}{2}$} (phole)
			(p1) edge [out=110, in=70] node [above, near end, xshift=-0.1cm] {$a,\frac{1}{2}$} (phole)
			;
		\end{tikzpicture}
		\caption{\label{fig:lowerBoundContainmentAndEquivalence}An example of the reduction defined in the proof of \Cref{lemma:lowerBoundContainmentAndEquivalence}.}
	\end{figure}
	We can show by induction that for every $u\in\Sigma^+$,
	\begin{equation*}
		\Tilde{\A}(u) = 
		\begin{cases}
			0 & u\in L(\A)
			\\
			\frac{1}{2^{|u|}} & u \notin L(\A)
		\end{cases}
	\end{equation*}
	Hence, $L(\A)$ is universal iff $\Tilde{\A}$ is equivalent to a $0$ NDA iff it is universal($\leq$) with respect to the threshold $0$.
\end{proof}

\begin{lem}\label{lemma:lowerBoundContainmentAndEquivalenceInf}
	The equivalence and universality($\leq$) problems of integral NDAs w.r.t.\ infinite words are PSPACE-hard.
\end{lem}
\begin{proof}
	Similarly to the proof of \Cref{lemma:lowerBoundContainmentAndEquivalence}, we construct in polynomial time an NDA $\Tilde{\A}$ with discount factor $2$, such that the input NFA is universal if and only if $\Tilde{\A}$ is equivalent to a $0$ NDA with respect to infinite words.
	Also in this reduction, no negative values of words will be possible, so it is also valid for showing the PSPACE-hardness of the universality($\leq$) problem.
	The reduction is similar to the one provided in the proof of \Cref{lemma:lowerBoundContainmentAndEquivalence}, with the following additions to support the case of infinite words:
	\begin{itemize}
		\item A new letter $\#$ to the alphabet.
		\item A new state $q_\infty$ to $\Tilde{\A}$.
		\item $0$-weighted $\#$-transitions from every state of $\Tilde{\A}$ to $q_\infty$.
		\item $0$-weighted self loops $(q_\infty,\sigma,q_\infty)$ for every alphabet letter $\sigma$.
	\end{itemize}
	An example of the construction is given in \Cref{fig:lowerBoundContainmentAndEquivalenceInf}.
	\begin{figure}
	\centering
	\begin{tikzpicture}[->,>=stealth',shorten >=1pt,auto,node distance=1.5cm, semithick, initial text=, every initial by arrow/.style={|->}]
		\node[initial, state, inner sep=0.1cm,minimum size=0.4cm] (q0) {$q_0$};
		\node[state, accepting, inner sep=0.1cm,minimum size=0.4cm] (q1) [right of=q0] {$q_1$}; 
		\node[state] (dummy) [right of =q1, draw=none, xshift=-0.4cm] {\Huge $\Rightarrow$};
		\node[initial, state, inner sep=0.1cm, minimum size=0.4cm] (p01) [right of=dummy] {$p_0$};
		\node[state, inner sep=0.0001cm,minimum size=0.4cm] (phole) [right of=p01] {$q_{hole}$};
		\node[state, inner sep=0.1cm,minimum size=0.4cm] (p0) [right of=phole, xshift=0.2cm] {$q_0$};
		\node[state, inner sep=0.1cm,minimum size=0.4cm] (p1) [right of=p0] {$q_1$};
		\node[state, inner sep=0.1cm,minimum size=0.4cm] (inf) [below of=phole, yshift=-0.5cm] {$q_\infty$};  
		\path 
		(q0) edge [out=30, in=150] node [above] {$a$} (q1)
		(q1) edge [out=-150, in=-30] node [below] {$a$} (q0)
		(p0) edge [out=30, in=150] node [above, near start] {$a,-1$} (p1)
		(p1) edge [out=-150, in=-30] node [below, near start] {$a,\frac{1}{2}$} (p0)
		(p01) edge [out=90, in=90] node [left, near start, xshift=-0.25cm] {$a,0$} (p1)
		(p0) edge  node [above] {$a,-\frac{1}{2}$} (phole)
		(p01) edge  node [below] {$a,\frac{1}{2}$} (phole)
		(phole) edge [out=110, in=150,loop, looseness=5] node [above,near end, yshift=0.1cm] {$a,-\frac{1}{2}$} (phole)
		(p1) edge [out=110, in=70] node [above, near end, xshift=-0.1cm] {$a,\frac{1}{2}$} (phole)
		(p01) edge [out=-90, in=170] node [left] {$\#,0$} (inf)
		(p1) edge [out=-90, in=10] node [right, xshift=0.3cm] {$\#,0$} (inf)
		(phole) edge node [left,near end] {$\#,0$} (inf)
		(p0) edge node [right] {$\#,0$} (inf)
		(inf) edge [out=-20, in=-60,loop, looseness=5] node [right, xshift=0.0cm, align=center] {$a,0$\\$\#,0$} (inf)
		;
	\end{tikzpicture}
	\caption{\label{fig:lowerBoundContainmentAndEquivalenceInf}An example of the reduction defined in the proof of \Cref{lemma:lowerBoundContainmentAndEquivalenceInf}.}
\end{figure}
	
	By this construction, the value of $\Tilde{\A}$ on an infinite word $u\con\# \con w$, where $u$ does not contain $\#$, is $\Tilde{\A}(u)$, hence $0$ if and only if $\A$ accepts $u$ and greater than $0$ if and only if $\A$ does not accept $u$.
	Notice that the value of $\Tilde{\A}$ on an infinite word that does not contain $\#$ is $0$.
	
	Hence, $\A$ is universal iff the value of $\Tilde{\A}$ on all infinite words is $0$ iff $\Tilde{\A}$ is equivalent to a $0$ NDA with respect to infinite words iff $\Tilde{\A}$ is universal($\leq$) with respect to the threshold $0$ and infinite words.
\end{proof}

\begin{lem}\label{lemma:lowerBoundUniversalityAndExact}
	The universality($<$) and exact-value problems of integral NDAs w.r.t.\ finite words are PSPACE-hard.
\end{lem}
\begin{proof} 
	Similarly to the proof of \Cref{lemma:lowerBoundContainmentAndEquivalence}, we show a polynomial reduction from the problem of NFA universality such that for every finite word $u$, we have $\Tilde{\A}(u)<0$ if and only if $\A$ accepts $u$, and $\Tilde{\A}(u)=0$ otherwise.
	This provides reductions to both the universality($<$) and exact-value problems.
	Once again we use the construction of \Cref{lem:NFA_to_NDA}, while slightly adjusting the weights: 
	\begin{equation*}
		\gamma'(p,\sigma,q) = 
		\begin{cases}
			-\frac{1}{2} & p\notin F, q\in F
			\\
			0 & p\notin F, q\notin F
			\\
			\frac{1}{2} & p \in F,q\in F
			\\
			1 & p\in F, q\notin F
		\end{cases}
	\end{equation*}

	We can show by induction on the length of the runs on an input word $u$ that
	\begin{equation*}
		\Tilde{\A}(u) = 
		\begin{cases}
			-\frac{1}{2^{|u|}} & u\in L(\A)
			\\
			0 & u \notin L(\A)
		\end{cases}
		\vspace{-1.6\baselineskip}
	\end{equation*}

\end{proof}
\subsubsection{PSPACE Algorithms}\label{sec:PspaceAlgorithms}
Consider a choice function $\theta$ and $\theta$-NMDAs $\A$ and $\B$.
Our PSPACE algorithms relate to 12 problems (see \Cref{tbl:decisionProblems}): Exact-value and strict/non-strict universality w.r.t.\ $\A$, and equivalence and strict/non-strict containment between $\A$ and $\B$, each over finite or infinite words.

Since equivalence and universality are easily shown to be special cases of containment, our main algorithms are for the containment and exact-value problems. Observe that while the latter problem considers the \emph{existance} of a word $w$ (s.t.\ $\A(w)=\nu$), the former problem requires that \emph{for every} word $w$ ($\A(w) > \B(w)$ or $\A(w) \geq \B(w)$). Yet, since PSPACE = NPSPACE = coNPSPACE, we may consider the opposite of the former problem, namely whether there exists a word $w$, s.t.\ $\A(w) \leq \B(w)$ or $\A(w) < \B(w)$, and our algorithms may be nondeterministic.

Considering the containment problem, let $\C$ be the NMDA obtained by taking the union of $\A$ and $\B$. That is, the set of states of $\C$ is the union of $\A$'s and $\B$'s states, its transition function, when restricted to $\A$'s states is as of $\A$ and when restricted to $\B$'s states is as of $\B$, etc. (Notice that $\C$ is equivalent to $\min(\A,\B)$.)

Our algorithm for the containment problem non-deterministically generates a word $u$ letter by letter, and performs an on-the-fly determinization of $\C$, along the procedure described in \Cref{sec:Determinizability}, with respect to the input word $u$.

Recall that after reading a word prefix $u$, the determinization procedure maintains (in space polynomial in $|\C|$) for each state $q$ of $\C$, the gap between the best run of $\C$ on $u$ that ends in $q$ and the overall best run of $\C$ on $u$. (So the gap of a state in which an optimal run on $u$ ends is $0$, while the gap of other states is bigger than or equal to $0$. The gap $\infty$ stands for an irrecoverable gap, namely for a positive gap that cannot be reduced to $0$ as the word continues.) 
The determinization of $\C$ provides information on the possible runs of $\A$ and $\B$ on a (prefix) word $u$, hence holds all the required information:
\begin{enumerate}
	\item A $0$-gap for an $\A$-state means that $\A$ has an optimal run on $u$ among all runs of $\A$ and $\B$ on $u$, therefore $\A(u) \leq \B(u)$, namely $u$ witnesses that there is no strict containment between $\A$ and $\B$ over finite words.
	\item A $0$-gap for an $\A$-state when the gaps of all $\B$-states are strictly positive means that $\A$ has an optimal run on $u$ while $\B$ does not, therefore $\A(u) < \B(u)$, namely $u$ witnesses that there is no non-strict containment between $\A$ and $\B$ over finite words.
	\item\label{itm:eq_finite} $0$-gaps for both an $\A$-state and a $\B$-state means that both $\A$ and $\B$ have an optimal run on $u$, therefore $\A(u) = \B(u)$. This will be used for the exact-value problem.
	\item A $0$-gap for an $\A$-state on a prefix $u$ when the gaps of all $\B$-states are $\infty$ means that $\A$ has an optimal run on any continuation of $u$ while $\B$ does not, therefore for a word $w\in u\Sigma^\omega$, we have $\A(w) < \B(w)$, namely $w$ witnesses that there is no non-strict containment between $\A$ and $\B$ over infinite words.
	\item A configuration of gaps (i.e., the set of gaps of all of $\A$'s and $\B$'s states), repeating after two different prefixes $u_1$ and $u_1 u_2$, in which an $\A$-state has a non $\infty$-gap means that $\A$ has an optimal run on $w=u_1 u_2^\omega$ among all the runs of $\A$ and $\B$ on $w$, therefore $\A(w) \leq \B(w)$, namely $w$ witnesses that there is no strict containment between $\A$ and $\B$ over infinite words.
	\item\label{itm:eq_inf} A configuration of gaps with non $\infty$-gaps for both an $\A$-state and a $\B$ state, which is repeating after two different prefixes $u_1$ and $u_1 u_2$, means that both $\A$ and $\B$ have optimal runs on $w=u_1 u_2^\omega$, therefore $\A(w) = \B(w)$. This will be used for the exact-value problem.
\end{enumerate}
\noindent 
For the exact-value problem, of whether there exists a word $w$ s.t.\ $\A(w)=\nu$, 
we use the approach of \Cref{itm:eq_finite,itm:eq_inf} above, while letting $\B$ stand for a constant-$\nu$ DMDA.

\begin{lem} \label{lem:PSPACE}
	For every choice function $\theta$, the strict and non-strict containment problems of $\theta$-NMDAs w.r.t.\ finite or infinite words are in PSPACE.
\end{lem}
\begin{proof}
	Consider a choice function $\theta$, and $\theta$-NMDAs $\A=\tuple{\Sigma,Q_\A,\iota_\A,\delta_\A,\gamma_\A,\rho_\A}$ and $\B=\tuple{\Sigma,Q_\B,\iota_\B,\delta_\B,\gamma_\B,\rho_\B}$.
	Denote the states of $\A$ as $Q_\A=\set{p_1,...,p_n}$ and of $\B$ as $Q_\B=\set{p_{n+1},...,p_{n+m}}$.
	Building on the equivalence PSPACE = NPSPACE = coNPSPACE, we nondeterministically construct a word, letter by letter, that witnesses non-containment, namely a word $u$, s.t.\ $\A(u) \leq \B(u)$ or $\A(u) < \B(u)$. (In the infinite-word case, we construct a lasso word $u_1u_2^\omega$, where $u_1$ and $u_2$ are finite.)
	
	We validate the adequateness of $u$ by constructing, on-the-fly, a $\theta$-DMDA $\D$ equivalent to $\C=\tuple{\Sigma,Q_\A\cup Q_\B,\iota_\A \cup \iota_B,\delta_\A \cup \delta_\B,\gamma_\A \cup \gamma_\B,\rho_\A \cup \rho_\B}$, as per the proof of \Cref{thm:DetTidy}. Along the construction, we only save the current state of $\D$ after reading the current prefix $u$ (or two such states), which due to \Cref{thm:DetTidy} only requires space polynomial in $|\C|$ and thus polynomial in $|\A|$ and $|\B|$.

	On every step of the construction, after generating a finite word prefix $u$, we examine the current state $S=\tuple{g_1, \cdots, g_n, g_{n+1}, \cdots, g_{n+m}}$ of $\D$, which consists of the current gaps of each original state of $\A$ and of $\B$. (For the definition of gaps, see \Cref{sec:Determinizability}). $\D$'s state $S$ shows that $u$ (or a related word) witnesses non-containment, with respect to the following containment problems, iff each corresponding condition, as detailed below, holds.
	
	\begin{itemize}
		\item Strict containment($>$) finite words: There exists $i\in[0..n]$ such that $g_i = 0$. (The word $u$ is a witness.)
		\item Non-strict containment($\geq$) finite words: There exists $i\in[0..n]$ such that $g_i = 0$ and for every $j\in[n{+}1..n{+}m]$ we have $g_j > 0$. (The word $u$ is a witness.)
		\item Non-strict containment($\geq$) infinite words: There exists $i\in[0..n]$ such that $g_i = 0$ and for every $j\in[n{+}1..n{+}m]$ we have $g_j = \infty$. (Every word $w\in u\Sigma^\omega$ is a witness.)
		\item Strict containment($>$) infinite words: The algorithm also (nondeterminstically) remembers some previous state $S'$ of $\D$, and the condition is that $S=S'$ and that there exists some $i\in[1.. n]$, such that $g_i\neq\infty$. (The witness is a word $u_1u_2^\omega$, where $u_1$ leads $\D$ from the initial state to $S$ and $u_2$ leads $\D$ from $S$ to $S'$, namely back to $S$.)
	\end{itemize}
\noindent 
	For showing correctness of the above conditions, we use the constant $T$ and the functions $\Gap$ and $\Cost$ as defined in \Cref{sec:Determinizability}. Notice that by \Cref{lem:DetCorrectness}, for every $h\in[1 .. n{+}m]$, we have $g_h=\Gap(p_h,u)$ if $\Gap(p_h,u) \leq 2T$ and $\infty$ otherwise.	Observe that since $\A$ and $\B$ are $\theta$-NMDAs, they agree on the accumulated discount factor over every finite word $u$, which we denote by $\rho(u)=\rho_\A(u)=\rho_\B(u)$.
	\begin{itemize}
		\item Strict containment($>$) finite words: The containment $\A>\B$ does not hold 
		iff there exists a finite word $u$ s.t.\ $\A(u)\leq \B(u)$ 
		iff there exists a state $p\in Q_\A$ s.t.\ $p$ is the target state of an optimal run of $\C$ on $u$
		iff there exists $p\in Q_\A$ s.t.\ $\Cost(p,u) = \C(u)$ 
		iff there exists $p\in Q_\A$ s.t.\ $\Gap(p,u)=\rho(u)\big(\Cost(p,u) - \C(u)\big)=0$ 
		iff there exists $i\in[1..n]$ s.t.\ $g_i=0$.
		\item Non-strict containment($\geq$) finite words:  The containment $\A\geq\B$ does not hold iff there exists a finite word $u$ s.t.\ 
		$\A(u) < \B(u)$ 
		iff there exists $p\in Q_\A$ s.t.\ $p$ is the target state of an optimal run of $\C$ on $u$, and every $p' \in Q_\B$ is not a target state of an optimal run of $\C$ on $u$ 
		iff there exists $p\in Q_\A$ s.t.\ $\Gap(p,u)=0$ and for all $p' \in Q_\B$, $\Gap(p',u)>0$ 
		iff there exists $i\in[1..n]$  s.t.\ $g_i=0$ and for all $j\in[n{+}1..n{+}m]$, we have $g_j >0$.
		\item Non-strict containment($\geq$) infinite words: 
		The containment $\A\geq\B$ does not hold 
		iff there exists an infinite word $w$ s.t.\ $\A(w)< \B(w)$. Recall that we provided for that case the condition that (1) there exists $i\in[0..n]$ such that $g_i = 0$ and (2) for every $j\in[n+1..n+m]$ we have $g_j = \infty$. We show the two directions of the condition correctness:
		
		\noindent
		$\Rightarrow$: 
		If the condition holds, then by (1) we have $\A(u)=\C(u)$, and by (2) we have  for every $j\in[n{+}1..n{+}m]$ that $\Gap(p_j,u)=\rho(u)\big(\Cost(p_j,u)-\C(w)\big)>2T$, implying that $\rho(u)\big(\B(u)-\C(u)\big)>2T$. 
		Hence, $\B(u)-\A(u)>\frac{2T}{\rho(u)}$. 
		Since the difference between two infinite runs of $\C$ on an infinite word $v$ is bounded by
		$\sum_{i=0}^\infty \frac{T}{\prod_{j=0}^{i-1}{\rho\big(v[j]\big)}} \leq \sum_{i=0}^\infty \frac{T}{2^i} = 2T$, 
		reading an infinite prefix $v$ after reading $u$, will change the difference between $\B$ and $\A$ by no more than $\frac{2T}{\rho(u)}$.
		We get $\B(uv)-\A(uv)\geq \B(u)-\A(u) - \frac{2T}{\rho(u)} >0$, and in particular for some $w=uv$, we have $\B(w)> \A(w)$, as required.
		
		\noindent
		$\Leftarrow$: 
		If $\B(w) > \A(w)$ then since $\rho(u)$ grows exponentially with the length of a word $u$, for a long enough prefix $u$ of $w$, we have $\B(w)-\A(w)> \frac{4T}{\rho(u)}$.
		Since the difference between two runs of $\C$ on infinite continuations of $u$ is bounded by $2T$, we have $\B(u) - \A(u) + \frac{2T}{\rho(u)} \geq \B(w)-\A(w)$, implying that $\B(u) - \A(u) > \frac{2T}{\rho(u)}$. Hence, $\rho(u)(\B(u) - \A(u))>2T$, and therefore the condition holds for $u$.
		\item Strict containment($>$) infinite words: 
		The containment $\A>\B$ does not hold
		iff there exists an infinite word $w$ s.t.\ $\A(w)\leq\B(w)$
		iff ( \circled{1} there exists an infinite word $w$ s.t.\ $\A(w)<\B(w)$ or \circled{2} there exists an infinite word $w$ s.t.\ $\A(w)=\B(w)$).
		\begin{itemize}
			\item 
		By the previous argument, \circled{1} holds iff for some finite word $u$, there exists $i\in[0..n]$ such that $g_i = 0$ and for every $j\in[n+1..n+m]$ we have $g_j = \infty$. In this case, the values $g_j$s will remain $\infty$ for all continuations of $u$, so by the finiteness of $\D$, some state repeats at some point. 
		Observe that since $S$ must have a $0$-gap, requiring $g_j = \infty$ for every $j\in[n{+}1..n{+}m]$ leads to $g_i = 0$ for some $i\in[0..n]$.
		Hence \circled{1} holds iff $S=S'$, there exists $i\in[0..n]$ s.t.\ $g_i \neq \infty$ and for every $j\in[n{+}1..n{+}m]$ we have $g_j = \infty$.
		\item 
		We will show that \circled{2} holds iff $S=S'$, and contains $g_i, g_j \neq \infty$ for some $i\in[0..n]$ and $j\in[n{+}1..n{+}m]$.
		
		\noindent
		$\Leftarrow$
		If $\A(w)=\B(w)$, by the finiteness of $\D$, the run of $\D$ on $w$ must infinitely often return to some state $S$. By the previous argument, it cannot be that for every $j\in[n{+}1..n{+}m]$ we have $g_j = \infty$, and by symmetry, nor can it be that for every $i\in[1..n]$ we have $g_i = \infty$. Thus, a state $S$ repeats with some $g_i\neq\infty$ and $g_j\neq\infty$, for $i\in[1..n]$ and $j\in[n+1..n+m]$.
		
		\noindent
		$\Rightarrow$
		Due to the determinism of $\D$, if it reaches $S$ reading a word $u_1$, and returns to $S$ after further reading a word $u_2$, then it will infinitely often reach $S$ reading $u_1 u_2^\omega$.
		Without loss of generality, we consider $i, j$ s.t.\ $g_i$ (respectively $g_j$), is the minimal gap between all states of $\A$ (respectively $\B$). 
		Denote by $\Pi$ the accumulated discount factor in $\C$ over the word $u_2$ after $u_1$ was already read, that is $\Pi=\frac{\rho(u_1 u_2)}{\rho(u_1)}$. 
		Now, 
		$\A(u_1 u_2^\omega) - \C(u_1 u_2^\omega) = \lim_{k\to\infty}{\Cost(p_i, u_1 u_2^k)-\C(u_1 u_2^k)} = 
		\lim_{k\to\infty}{\frac{\Gap(p_i,u_1 u_2^k)}{\rho(u_1 u_2^k)}} = 
		\lim_{k\to\infty}{\frac{g_i}{\rho(u) \Pi^k}} = 0$.
		Similarly, we get 
		$\B(u_1 u_2^\omega) - \C(u_1 u_2^\omega) = 0$, hence $\A(u_1 u_2^\omega) = \B(u_1 u_2^\omega)$ as required.
		
		\end{itemize}
		\noindent 
		Combining both results to achieve (\circled{1} or \circled{2}) iff $S=S'$ with $g_i\neq\infty$ for some $i\in[1.. n]$.
		\qedhere
	\end{itemize}
\end{proof}

\begin{lem} \label{lem:exact_PSPACE}
	The exact-value problem of tidy NMDAs w.r.t.\ finite or infinite words is in PSPACE.
\end{lem}
\begin{proof}
	Consider a tidy NMDA $\A$ for some choice function $\theta$ and a constant $\nu\in\Rat$. 
	To check whether there exists a word $w$ s.t.\ $\A(w)=\nu$, we first construct a $\theta$-NMDA $B$ that expresses the constant function $\nu$.
	Such an NMDA is identical to a transducer $\T$ that represents $\theta$, while duplicating the initial state, so the initial state of $\B$ has no incoming transitions.
	All the transitions from the initial state of $\B$ have a weight of $\nu$, while the weight of all other transitions is $0$. 
	The discount factors are as of $\T$.
	Alternatively, if a transducer for $\theta$ is not provided, we can perform the same process on the input automaton $\A$ to achieve a $\theta$-NMDA for the constant function $\nu$.
		
	Then, similarly to the algorithms of \Cref{lem:PSPACE}, we check for a witness $\A(w) = \B(w)$.
	We nondeterministically generate a word $u$, letter by letter, and determinize on-the-fly an NMDA $\C$ that is the union of $\A$ and $\B$ into a DMDA $\D$.
	Denote the states of $\A$ as $\set{p_1,...,p_n}$ and the states of $\B$ as $\set{p_{n+1},...,p_{n+m}}$.
	
	Considering the exact-value problem with respect to finite words, we have $\D(u)=\nu$ iff $g_i, g_j= 0$ for some $i\in[1..n], j\in[n{+}1..n{+}m]$. Indeed, $g_i=0$ implies $\A(u)=\C(u)$ and $g_j = 0$ implies $\C(u)=\B(u)=\nu$.

	Considering the exact-value problem with respect to infinite words, the condition for a positive answer is that a state $S$ is repeated twice, and contains $g_i, g_j \neq \infty$ for some $i\in[0..n], j\in[n{+}1..n{+}m]$.
	(The witness is a word $u_1u_2^\omega$, where $u_1$ leads $\D$ from the initial state to $S$ and $u_2$ leads $\D$ from $S$ back to $S$.)
	The correctness argument for this condition is provided in the proof of \Cref{lem:PSPACE} for the case of strict containment($>$) on infinite words.
\end{proof}

We continue with the universality problems which are special cases of the containment problems.
\begin{thm} \label{thm:universality}
	The universality problems of tidy NMDAs are in PSPACE.
	
	\noindent
	The universality($<$) w.r.t.\ finite words, universality($\leq$) w.r.t.\ finite words, and universality($\leq$) w.r.t.\ infinite words are PSPACE-complete.
\end{thm}
\begin{proof}
	Consider a tidy NMDA $\B$ for some choice function $\theta$, and a threshold $\nu$.
The universality($<$) is a special case of the containment($>$) problem, with a (PSPACE) initialization phase that creates a $\theta$-NMDA $\A$ for the constant function $\nu$ (the process of creating such an automaton is shown in the proof of \Cref{lem:exact_PSPACE}).
Similarly, the non-strict universality is a special case of the non-strict containment, when using a $\theta$-NMDA for $\nu$ as $\A$.

	Hardness directly follows from \Cref{lemma:lowerBoundUniversalityAndExact} for universality($<$) with respect to finite words, from \Cref{lemma:lowerBoundContainmentAndEquivalence} for universality($\leq$) with respect to finite words, and  from \Cref{lemma:lowerBoundContainmentAndEquivalenceInf} for universality($\leq$) with respect to infinite words.
\end{proof}

We summarize below the PSPACE algorithms of \Cref{lem:PSPACE,lem:exact_PSPACE} and the hardness proofs given in \Cref{sec:PspaceHardness}. Notice that while for some of the problems we provide PSPACE-completeness, for others we only show membership in PSPACE. Observe that since universality is a special case of containment, hardness of the former also shows hardness of the latter.

\begin{thm} \label{thm:containmentFinite}
	For every choice function $\theta$, the containment problem of $\theta$-NMDAs on finite words is PSPACE-complete for both strict and non-strict inequalities.
\end{thm}
\begin{proof}
	A PSPACE algorithm is provided in \Cref{lem:PSPACE}, and hardness in  \Cref{lemma:lowerBoundContainmentAndEquivalence,lemma:lowerBoundUniversalityAndExact}.
\end{proof}

\begin{thm} \label{thm:containmentInfinite}
	For every choice function $\theta$, the containment problem of $\theta$-NMDAs w.r.t.\ infinite words and non-strict inequality is PSPACE-complete.
\end{thm}
\begin{proof}
	A PSPACE algorithm is provided in \Cref{lem:PSPACE}, and hardness in \Cref{lemma:lowerBoundContainmentAndEquivalenceInf}. 
\end{proof}

\begin{thm} \label{thm:equivalence}
	For every choice function $\theta$, the equivalence problem of $\theta$-NMDAs,  w.r.t.\ both finite and infinite words, is PSPACE-complete.
\end{thm}
\begin{proof}
A PSPACE algorithm for equivalence directly follows from the fact that $\A\equiv\B$ if and only if $\A\geq \B$ and $\B\geq \A$, thus from \Cref{thm:containmentFinite,thm:containmentInfinite}.
Hardness is provided in \Cref{lemma:lowerBoundContainmentAndEquivalence,lemma:lowerBoundContainmentAndEquivalenceInf}
\end{proof}

\begin{thm} \label{thm:containmentInfiniteStrict}
	For every choice function $\theta$, the containment problem of $\theta$-NMDAs w.r.t.\ infinite words and strict inequality is in PSPACE.
\end{thm}
\begin{proof}
	Directly follows from \Cref{lem:PSPACE}. 
\end{proof}

\begin{thm}\label{thm:exact}
	The exact-value problem of tidy NMDAs is in PSPACE (and PSPACE-complete w.r.t.\ finite words).
\end{thm}
\begin{proof}
	PSPACE algorithms are given in \Cref{lem:exact_PSPACE}, and hardness with respect to finite words in \Cref{lemma:lowerBoundUniversalityAndExact}.
\end{proof}

\section{Conclusions and Future Work}\label{sec:Conclusions}
The measure functions most commonly used in the field of quantitative verification, whether for describing system properties \cite{DiscountingInSystems,Cha07,DiscountedDeterministicMarkov}, automata valuation schemes \cite{CDH10,BH12,BH14,ComparatorAutomataInQuantitativeVerification}, game winning conditions \cite{ZP96,Andersson06,DDGRT10}, or temporal specifications \cite{AFHMS05,THHY12,ABK14,BCHK14}, are the limit-average (mean payoff) and the discounted-sum functions.

Limit-average automata cannot always be determinized \cite{CDH10} and checking their (non-strict) universality is undecidable \cite{DDGRT10}. Therefore, the tendency is to only use deterministic such automata, possibly with the addition of algebraic operations on them \cite{CDEHR10}. 

Discounted-sum automata with an arbitrary rational discount factor also cannot always be determinized \cite{CDH10} and are not closed under algebraic operations \cite{BH14}. Yet, with an arbitrary integral discount factor, they do enjoy all of these closure properties and their decision problems are decidable \cite{BH14}. They thus provide a very interesting automata class for quantitative verification. Yet, they have a main drawback of only allowing a single discount factor. 

We define a rich class of discounted-sum automata with multiple integral factors (tidy NMDAs) that strictly extends the expressiveness of automata with a single factor, while enjoying all of the good properties of the latter, including the same complexity of the required decision problems.

While we show that the containment problem of two tidy NMDAs with the same choice function is decidable, and of 
general integral NMDAs is undecidable, we leave for future work the question with respect to two tidy NMDAs with different choice functions. Though the problem with respect to two NDAs with different discount factors is decidable in PSPACE \cite{BH23}, we believe that considering two different choice functions requires more involved techniques.

Another natural future direction is to consider NMDAs with real, as opposed to rational, discount factors, and in particular NMDAs with Pisot discount factors -- while NDAs with a single arbitrary rational discount factor do not behave well, it was recently shown that there are irrational discount factors with which NDAs are well behaved, and specifically that for every Pisot number $\lambda$, the class of $\lambda$-NDAs enjoys all of the closure properties that integral NDAs enjoy \cite{Bok24}. There is thus an interesting potential for extending tidy NMDAs to allow for multiple Pisot discount factors, while preserving their good properties.

\section*{Acknowledgments}
We thank an anonymous reviewer for their insightful comments and constructive suggestions, which helped clarify the explanations and simplify the proofs.

\bibliographystyle{alphaurl}
\bibliography{bib}

\end{document}